%% file: cc-v1.tex
\documentclass[a4paper,12pt,index=totoc,bibliography=totoc]{scrreprt}

\usepackage[colorlinks]{hyperref}
\usepackage{makeidx}
\usepackage{url}
\makeindex

\usepackage{pgfplots}
\usepackage{pgf}
\usepackage{tikz}

\usepackage{amsmath,amsthm,amssymb}
\usepackage{scrpage2}
\usepackage{xcolor}
\usepackage{mdframed}
\title{Lecture Notes on Channel Coding}
\author{Georg B\"ocherer\\
Institute for Communications Engineering\\
Technical University of Munich, Germany\\
\url{georg.boecherer@tum.de}}

\input{notation}

\numberwithin{equation}{chapter}

\newcounter{problemcount}[chapter]
\renewcommand{\theproblemcount}{\thechapter.\arabic{problemcount}}
\newcommand{\myproblem}{\noindent\refstepcounter{problemcount}\textbf{Problem \theproblemcount.\ }}

\theoremstyle{plain}
\newtheorem{lemma}{Lemma}[chapter]
\newtheorem{theorem}{Theorem}[chapter]
\newtheorem{proposition}{Proposition}[chapter]
\theoremstyle{definition}
\newtheorem{definition}{Definition}[chapter]

\definecolor{examplegray}{rgb}{.95,.95,.95}
\newtheorem{mdexample}{Example}[chapter]
\newenvironment{example}%
  {\begin{mdframed}[backgroundcolor=examplegray,hidealllines=true]\begin{mdexample}}%
  {\end{mdexample}\end{mdframed}}

\begin{document}
\maketitle

\begin{abstract}
These lecture notes on channel coding were developed for a one-semester course for graduate students of electrical engineering. Chapter 1 reviews the basic problem of channel coding. Chapters 2--5 are on linear block codes, cyclic codes, Reed-Solomon codes, and BCH codes, respectively. The notes are self-contained and were written with the intent to derive the presented results with mathematical rigor. The notes contain in total 68 homework problems, of which 20\% require computer programming.
\end{abstract}

\tableofcontents

\pagestyle{scrheadings}

\input{introduction}

\input{short}

\input{linear}

\input{cyclic}

\input{rs}

\input{bch}

\bibliographystyle{IEEEtran}
\bibliography{IEEEabrv,confs-jrnls,cc}

\printindex

\end{document}

%% file: notation.tex
\usepackage{bbm}
\usepackage{graphicx}
\usepackage{xspace}

\newcommand{\four}{\mathcal{F}}

\providecommand{\set}[1]{\ensuremath{\mathcal{#1}}}
\newcommand{\seta}{\set{A}}
\newcommand{\setb}{\set{B}}
\newcommand{\setc}{\set{C}}
\newcommand{\setg}{\set{G}}
\newcommand{\setr}{\set{R}}
\newcommand{\setu}{\set{U}}
\newcommand{\setv}{\set{V}}
\newcommand{\setx}{\set{X}}
\newcommand{\sety}{\set{Y}}

\newcommand{\setz}{\ensuremath{\mathbf{Z}}\xspace}

\providecommand{\field}[1]{\ensuremath{\mathbb{#1}}}
\newcommand{\myfield}{\field{F}}
\newcommand{\fieldtwo}{\field{F}_2}

\newcommand{\veczero}{\boldsymbol{0}}
\newcommand{\vecone}{\boldsymbol{1}}
\newcommand{\veca}{\boldsymbol{a}}
\newcommand{\vecb}{\boldsymbol{b}}
\newcommand{\vecc}{\boldsymbol{c}}

\newcommand{\vece}{\boldsymbol{e}}
\newcommand{\vecg}{\boldsymbol{g}}
\newcommand{\vech}{\boldsymbol{h}}
\newcommand{\vecs}{\boldsymbol{s}}
\newcommand{\vecu}{\boldsymbol{u}}
\newcommand{\vecv}{\boldsymbol{v}}
\newcommand{\vecw}{\boldsymbol{w}}
\newcommand{\vecx}{\boldsymbol{x}}
\newcommand{\vecy}{\boldsymbol{y}}
\newcommand{\vecz}{\boldsymbol{z}}

\newcommand{\matg}{\boldsymbol{G}}
\newcommand{\mathh}{\boldsymbol{H}}
\newcommand{\mati}{\boldsymbol{I}}
\newcommand{\matp}{\boldsymbol{P}}

\newcommand{\hd}{\ensuremath{\de_\text{H}}\xspace}
\newcommand{\hw}{\ensuremath{\we_\text{H}}\xspace}

\newcommand{\enc}{\ensuremath{\mathrm{ev}}\xspace}
\newcommand{\pc}{\ensuremath{P_c}\xspace}
\newcommand{\pe}{\ensuremath{P_e}\xspace}
\newcommand{\bc}{\ensuremath{\setc}\xspace}
\newcommand{\dml}{\ensuremath{\mathrm{d}_\text{ML}}\xspace}

\renewcommand{\mod}{\ensuremath{\bmod}\xspace}
\newcommand{\supp}{\ensuremath{\mathrm{supp}}\xspace}
\newcommand{\matik}{\ensuremath{\mati_k}\xspace}

\DeclareMathOperator{\im}{im}
\DeclareMathOperator{\de}{d}
\DeclareMathOperator{\we}{w}
\DeclareMathOperator*{\argmax}{\ensuremath{\arg\,\max}\xspace}
\DeclareMathOperator*{\argmin}{\arg\,\min}

\newcommand{\bpm}{\begin{pmatrix}}
\newcommand{\epm}{\end{pmatrix}}
\newcommand{\bbm}{\begin{bmatrix}}
\newcommand{\ebm}{\end{bmatrix}}

\newcommand{\oleq}[1]{\overset{\text{(#1)}}{\leq}}
\newcommand{\oeq}[1]{\overset{\text{(#1)}}{=}}

\newcommand{\oiff}[1]{\overset{\text{(#1)}}{\Leftrightarrow}}

\newcommand{\dmin}{d_{\min}}

%% file: introduction.tex

\chapter*{Preface}

The essence of reliably transmitting data over a noisy communication medium by channel coding is captured in the following diagram.
\begin{align*}
U\rightarrow\boxed{\text{encoder}}\rightarrow X \rightarrow\boxed{\text{channel}}\rightarrow Y \rightarrow\boxed{\text{decoder}}\rightarrow\hat{X}\rightarrow \hat{U}.
\end{align*}
Data $U$ is encoded by a codeword $X$, which is then transmitted over the channel. The decoder uses its observation $Y$ of the channel output to calculate a codeword estimate $\hat{X}$, from which an estimate $\hat{U}$ of the transmitted data is determined.

These notes start in Chapter 1 with an invitation to channel coding, and provide in the following chapters a sample path through algebraic coding theory. The destination of this path is decoding of BCH codes. This seemed reasonable to me, since BCH codes are widely used in standards, of which DVB-T2 is an example. This course covers only a small slice of coding theory. However within this slice, I have tried to derive all results with mathematical rigor, except for some basic results from abstract algebra, which are stated without proof. The notes can hopefully serve as a starting point for the study of channel coding.

\subsubsection{References}

The notes are self-contained. When writing the notes, the following references were helpful:
\begin{itemize}
\item Chapter~\ref{chap:short}: \cite{cover2006elements},\cite{gallager2013stochastic}.
\item Chapter~\ref{chap:linear}: \cite{mceliece2004theory}.
\item Chapter~\ref{chap:cyclic}: \cite{mceliece2004theory}.
\item Chapter~\ref{chap:rs}: \cite{forney2005principles},\cite{blahut2003algebraic}.
\item Chapter~\ref{chap:bch}: \cite{moon2005error},\cite{massey1997applied}.
\end{itemize}
Please report errors of any kind to \texttt{georg.boecherer@tum.de}.

\hfill G. B\"ocherer

\chapter*{Achnowledgments}

I used these notes when giving the lecture ``Channel Coding" at the Technical University of Munich in the winter terms from 2013 to 2015. Many thanks to the students Julian Leyh, Swathi Patil, Patrick Schulte, Sebastian Baur, Christoph Bachhuber, Tasos Kakkavas, Anastasios Dimas, Kuan Fu Lin, Jonas Braun, Diego Su\'arez, Thomas Jerkovits, and Fabian Steiner for reporting the errors to me, to Siegfried B\"ocherer for proofreading the notes, and to Markus Stinner and Hannes Bartz, who were my teaching assistants and contributed many of the homework problems.

\hfill G. B\"ocherer

%% file: short.tex
\chapter{Channel Coding}
\label{chap:short}

In this chapter, we develop a mathematical model of data transmission over unreliable communication channels. Within this model, we identify a trade-off between reliability, transmission rate, and complexity. We show that the exhaustive search for systems that achieve the optimal trade-off is infeasible. This motivates the development of the algebraic coding theory, which is the topic of this course.

\section{Channel}

We model a communication channel by a discrete and finite input alphabet $\setx$, a discrete and finite output alphabet $\sety$, and transition probabilities
\begin{align}
P_{Y|X}(b|a):=\Pr(Y=b|X=a),\qquad b\in\sety,a\in\setx.
\end{align}
The probability $P_{Y|X}(b|a)$ is called the \emph{likelihood}\index{likelihood} that the output value is $b$ given that the input value is $a$.
For each input value $a\in\setx$, the output value is a random variable $Y$ that is distributed according to $P_{Y|X}(\cdot|a)$.

\begin{example}
The \emph{binary symmetric channel}\index{binary symmetric channel} (BSC)\index{BSC \see{binary symmetric channel}} has the input alphabet $\setx=\{0,1\}$, the output alphabet $\sety=\{0,1\}$ and the transition probabilities
\begin{align}
\text{input }0\colon\quad &P_{Y|X}(1|0)=1-P_{Y|X}(0|0)=\delta\\
\text{input }1\colon\quad&P_{Y|X}(0|1)=1-P_{Y|X}(1|1)=\delta.
\end{align}
The parameter $\delta$ is called the \emph{crossover probability}\index{crossover probability}. Note that 
\begin{align}
P_{Y|X}(0|0)+P_{Y|X}(1|0)=(1-\delta)+\delta=1
\end{align}
which shows that $P_{Y|X}(\cdot|0)$ defines a distribution on $\sety=\{0,1\}$.
\end{example}

\begin{example}
The \emph{binary erasure channel}\index{binary erasure channel} (BEC)\index{BEC \see{binary erasure channel}} has the input alphabet $\setx=\{0,1\}$, the output alphabet $\sety=\{0,1,e\}$ and the transition probabilities
\begin{align}
\text{input }0\colon\quad&P_{Y|X}(e|0)=1-P_{Y|X}(0|0)=\epsilon,\quad P_{Y|X}(1|0)=0\\
\text{input }1\colon\quad&P_{Y|X}(e|1)=1-P_{Y|X}(1|1)=\epsilon,\quad P_{Y|X}(0|1)=0.
\end{align}
The parameter $\epsilon$ is called the \emph{erasure probability}\index{erasure probability}.
\end{example}

\section{Encoder}

For now, we model the encoder as a device that chooses the channel input $X$ according to a distribution $P_X$ that is defined as
\begin{align}
P_X(a):=\Pr(X=a),\qquad a\in\setx.
\end{align}
For each symbol $a\in\setx$, $P_X(a)$ is called the \emph{a priori probability} of the input value $a$. In Section~\ref{sec:cyclic:encoder}, we will take a look at how an encoder generates the channel input $X$ by encoding data.

\section{Decoder}
Suppose we want to use the channel once. This corresponds to choosing the input value according to a distribution $P_X$ on the input alphabet $\setx$. The probability to transmit a value $a$ and to receive a value $b$ is given by
\begin{align}
P_{XY}(ab)=P_X(a)P_{Y|X}(b|a).
\end{align}
We can think of one channel use as a random experiment that consists in drawing a sample from the joint distribution $P_{XY}$. We assume that both the a priori probabilities $P_X$ and the likelihoods $P_{Y|X}$ are known at the decoder.
\subsection{Observe the Output, Guess the Input}
At the decoder, the channel output $Y$ is observed. Decoding consists in guessing the input $X$ from the output $Y$. More formally, the decoder consists in a deterministic function
\begin{align}
f\colon\sety\to\setx.
\end{align}
We want to design an optimal decoder, i.e., a decoder for which some quantity of interest is maximized. A natural objective for decoder design is to maximize the average probability of correctly guessing the input from the output, i.e., we want to maximize the \emph{average probability of correct decision}\index{average probability of correct decision}\index{Pc@\pc \see{average probability of correct decision}}, which is given by
\begin{align}
P_c:=\Pr[X=f(Y)].
\end{align}
The \emph{average probability of error}\index{average probability of error}\index{Pe@\pe \see{average probability of error}} is given by
\begin{align}
P_e:=1-P_c.
\end{align}

\subsection{MAP Rule}
We now derive the decoder that maximizes $P_c$.
\begin{align}
P_c=\Pr[X=f(Y)]=&\sum_{ab\in\setx\times\sety\colon a=f(b)}P_{XY}(ab)\\
&=\sum_{ab\in\setx\times\sety\colon a=f(b)}P_{Y}(b)P_{X|Y}(a|b)\\
&=\sum_{b\in\sety}P_{Y}(b)P_{X|Y}[f(b)|b].
\end{align}
From the last line, we see that maximizing the average probability of correct decision is equivalent to maximizing for each observation $b\in\sety$ the probability to guess the input correctly. The optimal decoder is therefore given by
\begin{align}
f(b)=\argmax_{a\in\setx} P_{X|Y}(a|b).\label{eq:map}
\end{align}
The operator `$\argmax$'\index{arg\,max} returns the argument where a function assumes its maximum value, i.e.,
\begin{align*}
a^*=\argmax_{a\in\setx}\Leftrightarrow f(a^*)=\max_{a\in\setx}f(a).
\end{align*}
The probability $P_{X|Y}(a|b)$ is called the \emph{a posteriori probability}\index{a posteriori probability} of the input value $a$ given the output value $b$. The rule \eqref{eq:map} is called the \emph{maximum a posteriori probability} (MAP) rule\index{maximum a posteriori probability}\index{MAP \see{maximum a posteriori probability}}. We write $f_\text{MAP}$ to refer to a decoder that implements the rule defined in \eqref{eq:map}.

We can write the MAP rule \eqref{eq:map} also as 
\begin{align}
\argmax_{a\in\setx} P_{X|Y}(a|b)=&\argmax_{a\in\setx}\frac{P_{XY}(ab)}{P_Y(b)}\nonumber\\
=&\argmax_{a\in\setx} P_{XY}(ab)\nonumber\\
=&\argmax_{a\in\setx} P_{X}(a)P_{Y|X}(b|a).\label{eq:short:mapapriori}
\end{align}
From the last line, we see that the MAP rule is determined by the a priori information $P_X(a)$ and the likelihood $P_{Y|X}(b|a)$. 
\begin{example}
We calculate the MAP decoder for the BSC with crossover probability $\delta=0.2$ and $P_X(0)=0.1$. We calculate
\begin{align}
&P_{XY}(0,0)=0.1\cdot(1-0.2)=0.08\\
&P_{XY}(0,1)=0.1\cdot0.2=0.02\\
&P_{XY}(1,0)=0.9\cdot0.2=0.18\\
&P_{XY}(1,1)=0.9\cdot(1-0.2)=0.72.
\end{align}
Thus, by \eqref{eq:short:mapapriori}, the MAP rule is
\begin{align}
f_\text{MAP}(0)=1,\quad f_\text{MAP}(1)=1\label{eq:short:mapruleex}.
\end{align}
For the considered values of $P_X(0)$ and $\delta$, the MAP decoder that maximizes the probability of correct decision always decides for $1$, irrespective of the observed value $b$.
\end{example}

\subsection{ML Rule}

By neglecting the a priori information in \eqref{eq:short:mapapriori} and by choosing our guess such that the likelihood is maximized, we get the \emph{maximum likelihood} (ML) rule\index{maximum likelihood}\index{ML \see{maximum likelihood}}
\begin{align}
f_\text{ML}(b)=\argmax_{a\in\setx}P_{Y|X}(b|a).\label{eq:short:mlrule}
\end{align}
\begin{example}
We calculate the ML rule for the BSC with crossover probability $\delta=0.2$ and $P_X(0)=0.1$. The likelihoods are
\begin{align}
&P_{Y|X}(0|0)=0.8\\
&P_{Y|X}(0|1)=0.2\\
&P_{Y|X}(1|0)=0.2\\
&P_{Y|X}(1|1)=0.8.
\end{align}
By \eqref{eq:short:mlrule}, the ML rule becomes
\begin{align}
f_\text{ML}(0)=0,\quad f_\text{ML}(1)=1.\label{eq:short:mlruleex}
\end{align}
Note that for this example, the ML rule \eqref{eq:short:mlruleex} is different from the MAP rule \eqref{eq:short:mapruleex}.
\end{example}

\section{Block Codes}

\subsection{Probability of Error vs Transmitted Information}

So far, we have only addressed the decoding problem, namely how to (optimally) guess the input having observed the output taking into account the a priori probabilities $P_X$ and the channel likelihoods $P_{Y|X}$. However, we can also design the encoder, i.e., we can decide on the input distribution $P_X$. Consider a BSC with crossover probability $\delta=0.11$. We plot the probability of transmitting a zero $P_X(0)$ versus the error probability $P_e$ of a MAP decoder. The plot is shown in Figure~\ref{fig:bsc_px_pe}. For $P_X(0)=1$, we have $P_e=0$, i.e., we decode correctly with probability one! The reason for this is that the MAP decoder does not use its observation at all to determine the input. Since $P_X(0)=1$, the decoder knows for sure that the input is equal to zero irrespective of the output value. Although we always decode correctly, the configuration $P_X(0)=1$ is useless, since we do not transmit any information at all. We quantify how much information is contained in the input by
\begin{align}
H(X)=\sum_{a\in\supp P_X}P_X(a)\log_2 \frac{1}{P_X(a)}\label{eq:short:entropy}
\end{align}
where $\supp P_X:=\{a\in\setx\colon P_X(a)>0\}$ denotes the \emph{support}\index{support}\index{s@\supp \see{support}} of $P_X$, i.e., the set of values $a\in\setx$ that occur with positive probability. The quantity $H(X)$ is called the \emph{entropy}\index{entropy}\index{H@\textit{H} \see{entropy}} of the random variable $X$. Since entropy is calculated with respect to $\log_2$ in \eqref{eq:short:entropy}, the unit of information is called \emph{bits}\index{bits}. Entropy has the property (see Problem~\ref{prob:short:entropyproperties})
\begin{align}
&0\leq H(X)\leq\log_2|\setx|.\label{eq:short:entropyproperties}
\end{align}
We plot entropy versus probability of error. The plot is displayed in Figure~\ref{fig:bsc_r_pe}. We now see that there is a trade-off between the amount of information that we transmit over the channel and the probability of error. For $P_X(0)=\frac{1}{2}$, information is maximized, but also the probability of error takes its greatest value. This observation is discouraging. It suggest that the only way to increase reliability is to decrease the amount of transmitted information. Fortunately, this is not the end of the story, as we will see next. 
\begin{figure}
\centering
\footnotesize
\includegraphics{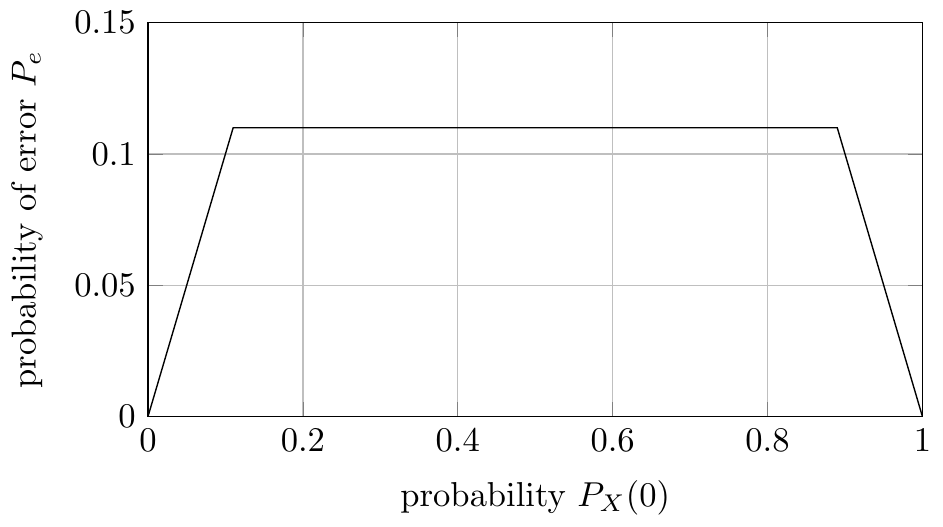}
\caption{Channel input probability $P_X(0)$ versus probability of error $P_e$ of a MAP decoder for a BSC with crossover probability $\delta=0.11$.}
\label{fig:bsc_px_pe}
\end{figure}
\begin{figure}
\centering
\footnotesize
\includegraphics{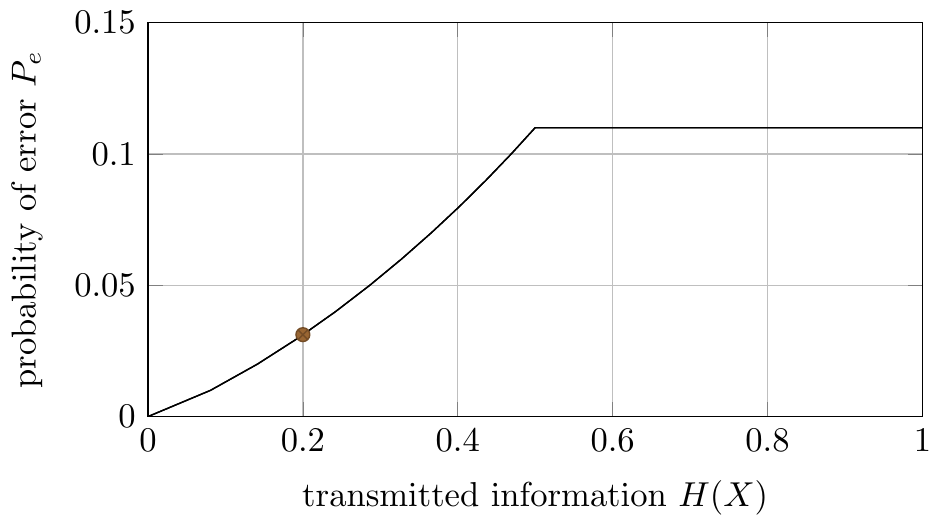}
\caption{Transmitted information $H(X)$ versus probability of error $P_e$ of a MAP decoder for a BSC with crossover probability $\delta=0.11$.}
\label{fig:bsc_r_pe}
\end{figure}
\subsection{Probability of Error, Information Rate, \text{Block Length}}

In Figure~\ref{fig:bsc_r_pe}, we see that transmitting $H(X)=0.2$ bits per channel use over the BSC results in $P_e=0.03$. We can do better than that by using the channel more than once. Suppose we use the channel $n$ times. The parameter $n$ is called the \emph{block length}\index{block length}\index{n@\textit{n} \see{block length}}. The input consists in $n$ random variables $X^n=X_1X_2\dotsb X_n$ and the output consists in $n$ random variables $Y^n=Y_1Y_2\dotsb Y_n$. The joint distribution of the random experiment that corresponds to $n$ channel uses is
\begin{align}
P_{X^nY^n}(a^nb^n)&=P_{X^n}(a^n)P_{Y^n|X^n}(b^n|a^n)\\
&=P_{X^n}(a^n)\prod_{i=1}^n P_{Y|X}(b_i|a_i).
\end{align}
In the last line, we assume that conditioned on the inputs, the outputs are independent, i.e.,
\begin{align}
P_{Y^n|X^n}(b^n|a^n)=\prod_{i=1}^n P_{Y|X}(b_i|a_i).
\end{align}
Discrete channels with this property are called \emph{discrete memoryless channels}\index{discrete memoryless channel} (DMC)\index{DMC \see{discrete memoryless channel}}. To optimally guess blocks of $n$ inputs from blocks of $n$ outputs, we define a \emph{super channel}\index{super channel} $P_{X'|Y'}$ with input $X':=X^n$ and output $Y':=Y^n$ and then use our MAP decoder for the super channel. The \emph{information rate}\index{information rate}\index{R@\textit{R} \see{information rate}} $R$ is defined as the information we transmit per channel use, which is given by
\begin{align}
R:=\frac{H(X^n)}{n}.
\end{align}
For a fixed block length $n$, we can trade probability of error for information rate by choosing the joint distribution $P_{X^n}$ of the input appropriately. 

From now on, we restrict ourselves to special distributions $P_{X^n}$. First, we define a \emph{block code}\index{block code}\index{C@\bc \see{block code}} as the set of input vectors that we choose with non-zero probability, i.e.,
\begin{align}
\setc:=\supp P_{X^n}\subseteq \setx^n.
\end{align}
The elements $c^n\in\setc$ are called \emph{code words}\index{code word}. Second, we let $P_{X^n}$ be a uniform distribution on $\setc$, i.e,
\begin{align}
P_{X^n}(a^n)=\begin{cases}
\frac{1}{|\setc|}&a^n\in\setc\\
0&\text{otherwise}.
\end{cases}
\end{align}
The rate can now be written as
\begin{align}
R=&\frac{H(X^n)}{n}\\
=&\frac{\displaystyle\sum_{a^n\in\supp P_{X^n}} P_{X^n}(a^n)\log_2\frac{1}{P_X(a^n)}}{n}\\
=&\frac{\displaystyle\sum_{c^n\in\setc} \frac{1}{|\setc|}\log_2|\setc|}{n}\\
=&\frac{\log_2|\setc|}{n}.
\end{align}
For a fixed block-length $n$, we can now trade probability of error for information rate via the code $\setc$. First, we would decide on the rate $R$ and then we would choose among all codes of size $2^{nR}$ the one that yields the smallest probability of error.
\begin{example}
For the BSC with crossover probability $\delta=0.11$, we search for the best codes for block length $n=2,3,\dotsc,7$. For complexity reasons, we only evaluate code sizes $|\setc|=2,3,4$. For each pair $(n,|\setc|)$, we search for the code with these parameters that has the lowest probability of error under MAP decoding. The results are displayed in Figure~\ref{fig:optimal_codes_bsc}. 

In Figure~\ref{fig:bsc_r_pe}, we observed for the code $\{0,1\}$ of block length $1$ the information rate $0.2$ and the error probability $0.03$. We achieved this by using the input distribution $P_X(0)=1-P_X(1)=0.034$. This can be improved upon by using the code
\begin{align}
\setc=\{00000,11111\}
\end{align}
with a uniform distribution. The block length is $n=5$ and the rate is
\begin{align}
\frac{\log_2|\setc|}{n}=\frac{1}{5}=0.2.
\end{align}
The resulting error probability is $P_e=0.0112$, see Figure~\ref{fig:optimal_codes_bsc}. Thus, by increasing the block length from $1$ to $5$, we could lower the probability of error from $0.03$ to $0.0112$. In fact, the longer code transmits $0.2\cdot 5$ information bits correctly with probability $1-0.0112$, while the short code only transmits $0.2$ information bits correctly with probability $1-0.03$.

We want to compare the performance of codes with different block length. To this end, we calculate for each code in Figure~\ref{fig:optimal_codes_bsc} the probability $P_\text{cb}$ that it transmits $840$ bits correctly, when it is applied repeatedly. The number $840$ is the least common multiple of the considered block lengths $2,3,\dotsc,8$. For a code with block length $n$ and error probability $P_e$, the probability $P_\text{cb}$ is calculated by
\begin{align}
P_\text{cb}=(1-P_e)^\frac{840}{n}.
\end{align}
The results are displayed in Figure~\ref{fig:optimal_codes_bsc_pb}. Three codes are marked by a circle. They exemplify that by increasing the block length from $4$ to $6$ to $7$, both probability of correct transmission \emph{and} information rate are increased.
\end{example}
\begin{figure}
\footnotesize
\includegraphics{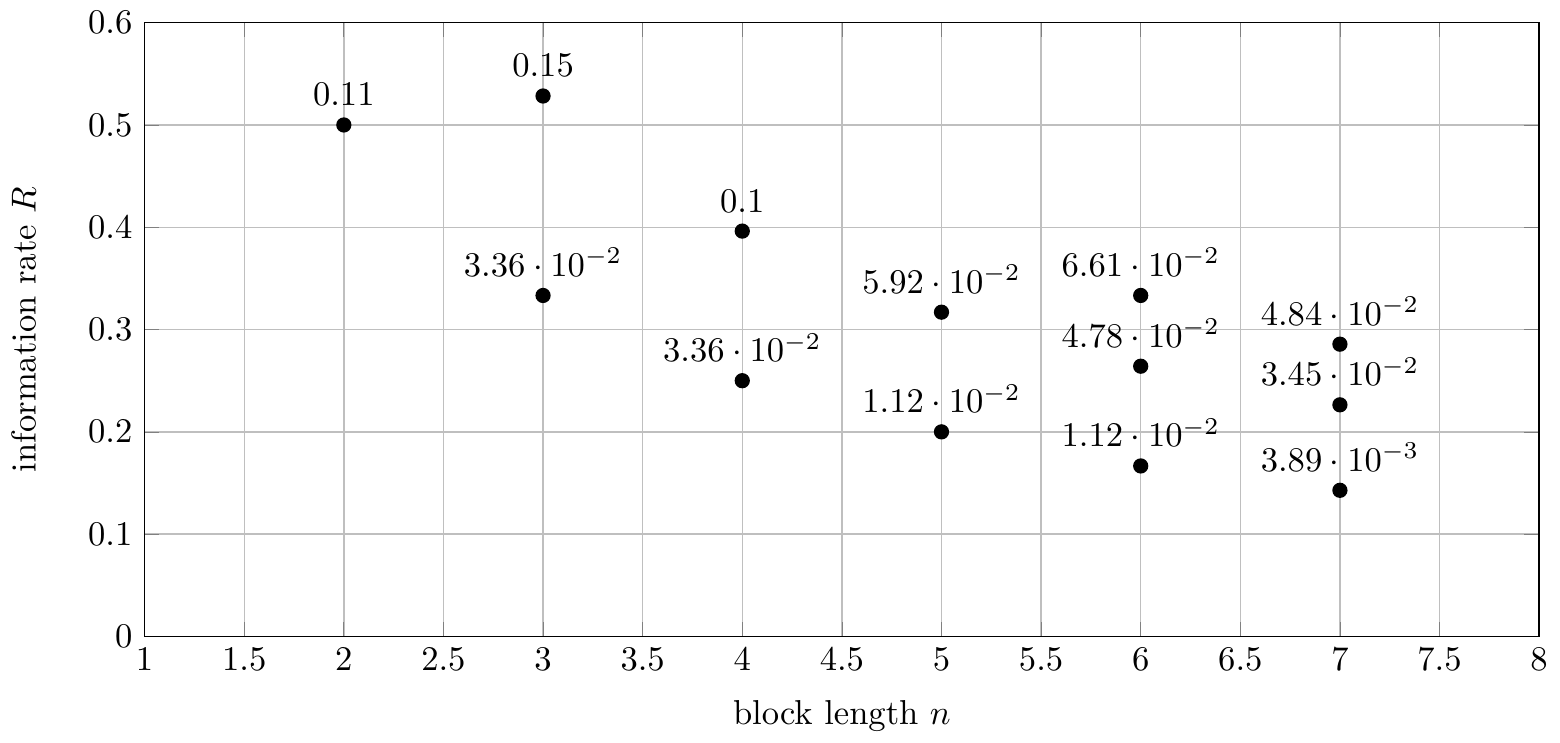}
\caption{Optimal codes for the BSC with crossover probability $\delta=0.11$. Block length and rate are displayed in horizontal and vertical direction, respectively. Each code is labeled by the achieved error probability $P_e$.}
\label{fig:optimal_codes_bsc}
\end{figure}
\begin{figure}
\footnotesize
\includegraphics{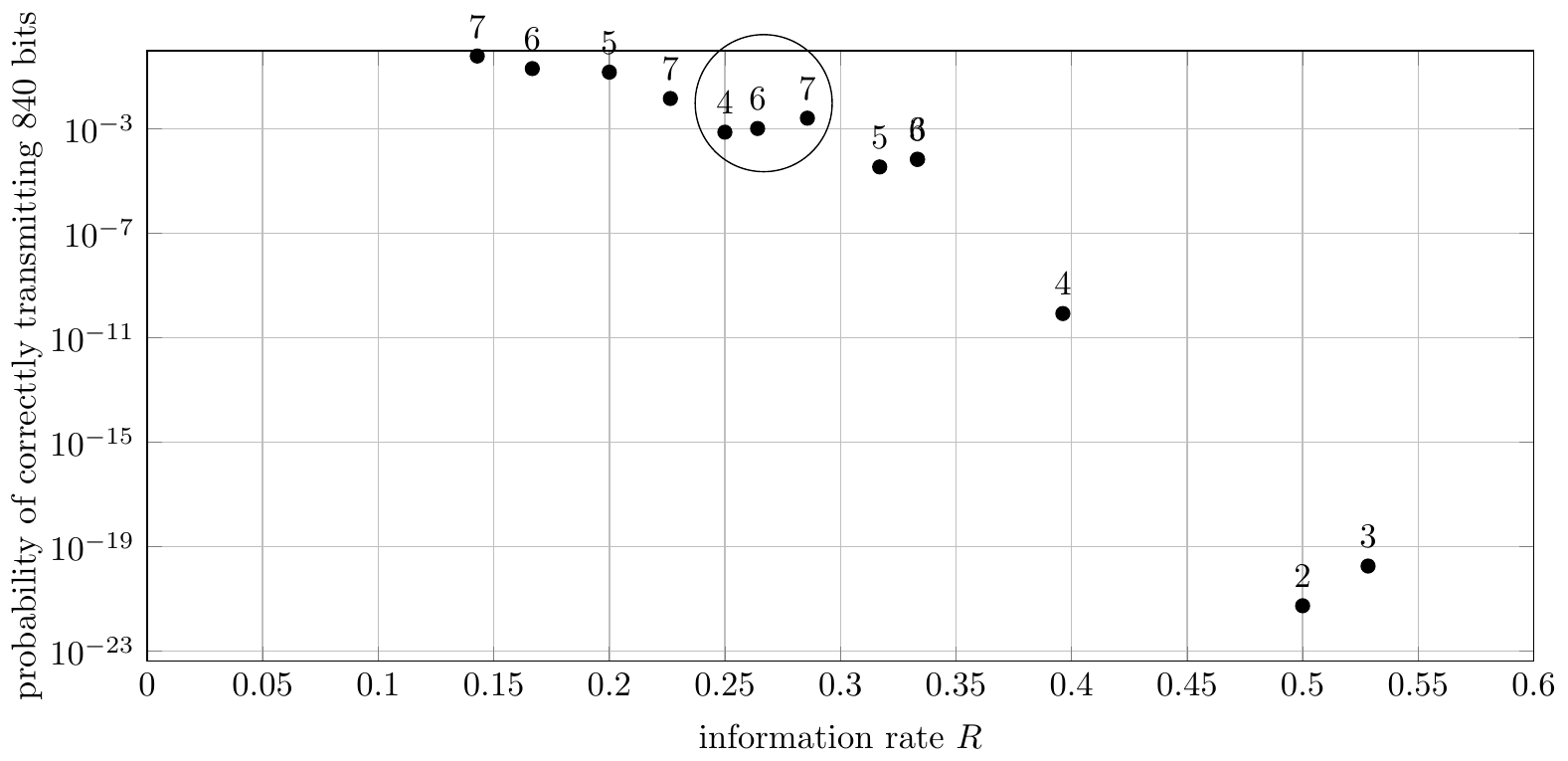}
\caption{Optimal codes for the BSC with crossover probability $\delta=0.11$. In horizontal direction, the transmission rate is displayed. For fair comparison of codes of different length, the probability to transmit 840 bits correctly is displayed in vertical direction. Each code point is labeled by its block length.}
\label{fig:optimal_codes_bsc_pb}
\end{figure}
\subsection{ML Decoder}

Let the input distribution be uniform on the code $\setc$, i.e.,
\begin{align}
P_X(a)=\begin{cases}
\frac{1}{|\setc|},&a\in\setc\\
0,&\text{otherwise}.
\end{cases}\label{eq:short:uniformpxonc}
\end{align}
When \eqref{eq:short:uniformpxonc} holds, the MAP rule can be written as
\begin{align}
f_\text{MAP}(b)&\oeq{a}\argmax_{a\in\setx} P_{X}(a)P_{Y|X}(b|a)\label{eq:short:mapdecoder}\\
&\oeq{b}\argmax_{c\in\setc}P_{Y|X}(b|c)\label{eq:short:mldecoder1}
\end{align}
where we used \eqref{eq:short:mapapriori} in (a) and where (b) is shown in Problem~\ref{prob:short:mldecoder}. 
Note that the maximization in \eqref{eq:short:mapdecoder} is over the whole input alphabet while the maximization in \eqref{eq:short:mldecoder1} is over the code. This shows that when \eqref{eq:short:uniformpxonc} holds, knowing the a priori information $P_X$ is equivalent to knowing the code $\setc$. The rule in \eqref{eq:short:mldecoder1} resembles the ML rule \eqref{eq:short:mlrule}, with the difference that the likelihood is maximized over $\setc$. In accordance with the literature, we define the \emph{ML decoder}\index{ML decoder}\index{dml@\dml \see{ML decoder}} by
\begin{align}
\dml(b):=\argmax_{c\in\setc}P_{Y|X}(b|c)\label{eq:short:mldecoder}
\end{align}
Whenever we speak of an ML decoder in the following chapters, we mean \eqref{eq:short:mldecoder}.

\clearpage

\section{Problems}

\myproblem \label{prob:short:entropyproperties} Let $X$ be a random variable with the distribution $P_X$ on $\mathcal{X}$. Show that
\begin{align*}
0\overset{\text{(a)}}{\leq}H(X)\overset{\text{(b)}}{\leq}\log_2|\mathcal{X}|.
\end{align*}
For which distributions do we have equality in (a) and (b), respectively?

\myproblem \label{prob:short:mldecoder}
Consider a channel $P_{Y|X}$ with input alphabet $\mathcal{X}$ and output alphabet $\mathcal{Y}$. Let $\mathcal{C}\subseteq\mathcal{X}$ be a code and let the input be distributed according to
\begin{align}
P_X(a)=\begin{cases}
\frac{1}{|\mathcal{C}|}&a\in\mathcal{C}\\
0&\text{otherwise}.
\end{cases}
\end{align}
\begin{enumerate}
\item Show that decoding by using the MAP rule to choose a guess from the alphabet $\mathcal{X}$ is equivalent to using the ML rule to choose a guess from the code $\mathcal{C}$.\\\emph{\textbf{Remark}: This is why a MAP decoder for an input that is uniformly distributed over the code is usually called an ML decoder. We also use this convention.}
\end{enumerate}

\myproblem Consider a BEC with erasure probability $\epsilon=0.1$. The input distribution is $P_X(0)=1-P_X(1)=0.3$. 
\begin{enumerate}
\item Calculate the joint distribution $P_{XY}$ of input $X$ and output $Y$.
\item What is the probability that an erasure is observed at the output?
\end{enumerate}

\myproblem Consider a BSC with crossover probability $\delta=0.2$. We observe the output statistics  $P_Y(0)=0.26$ and $P_Y(1)=0.74$.
\begin{enumerate}
\item Calculate the input distribution $P_X$.
\end{enumerate}

\myproblem\label{short:prob:superchannel} A channel $P_{Y|X}$ with input alphabet $\mathcal{X}=\{1,2,\dotsc,|\mathcal{X}|\}$ and output alphabet $\mathcal{Y}=\{1,2,\dotsc,|\mathcal{Y}|\}$ can be represented by a stochastic matrix $\boldsymbol{H}$ with $|\mathcal{X}|$ rows and $|\mathcal{Y}|$ columns that is defined by
\begin{align}
\boldsymbol{H}_{ij}=P_{Y|X}(j|i).
\end{align}
In particular, the $i$th row contains the distribution $P_{Y|X}(\cdot|i)$ on the output alphabet when the input is equal to $i$ and the entries of each row sum up to one. 
\begin{enumerate}
\item What is the stochastic matrix that describes a BSC with crossover probability $\delta$?
\item Suppose you use the BSC twice. What is the stochastic matrix that describes the channel from the length 2 input vector to the length 2 output vector?
\item Write a function in Matlab that calculates the stochastic matrix of $n$ BSC uses.
\end{enumerate}

\myproblem Consider the code
\begin{align}
\mathcal{C}=\{\underbrace{00\dotsb0}_{n\text{ times}},\underbrace{11\dotsb 1}_{n\text{ times}}\}.
\end{align}
Such a code is called a \emph{repetition code}. Each codeword is transmitted equally likely. We transmit over a BSC.  
\begin{enumerate}
\item What is the blocklength and the rate of the code?
\item For crossover probabilities $\delta=0.1,0.2$ and blocklength $n=1,2,3,4,5$, calculate the error probability of an ML decoder.
\item Plot rate versus error probability.
\end{enumerate}
\emph{Hint:} You may want to use your Matlab function from Problem~\ref{short:prob:superchannel}.

\myproblem  Consider a BSC with crossover probability $\delta$. For blocklength $n=5$, we want to find the best code with $3$ code words, under the assumption that all three code words are transmitted equally likely.
\begin{enumerate}
\item How many different codes are there?
\item Write a Matlab script that finds the best code by exhaustive search. What are the best codes for $\delta=0.1,0.2$ and what are the error probabilities?
\item Add rate and error probability to the plot from Problem~\ref{short:prob:bestcodes}.
\end{enumerate}

\myproblem Two random variables $X$ and $Y$ are stochastically independent if
\begin{align*}
P_{XY}(ab)=P_X(a)P_X(b),\quad\text{for all }a\in\mathcal{X},b\in\mathcal{Y}.
\end{align*}
Consider a binary repetition code with block length $n=4$ and let the input distribution be given by
\begin{align*}
P_{X^4}(a^4)=\begin{cases}
\frac{1}{2}&a^4\in\{0000,1111\}\\
0&\text{otherwise}.
\end{cases}
\end{align*}
Show that the entries $X_2$ and $X_4$ of the input $X^4=X_1X_2X_3X_4$ are stochastically dependent.

\myproblem\label{short:prob:bestcodes} Consider a BSC with crossover probability $\delta=0.11$. You are asked to design a transmission system that operates at an information rate of $R=0.2$ bits per channel use. You decide for evaluating the performance of repetition codes.
\begin{enumerate}
\item For block lengths $n=1,2,3,\dotsc$, calculate the input distribution $P_{X^n}$ for which the information rate is equal to $0.2$. What is the maximum block length $n_{\max}$ for which you can achieve $R=0.2$ with a repetition code?
\item For each $n=1,2,3,\dotsc,n_{\max}$, calculate the probability of error $P_e$ that is achieved by a MAP decoder and plot $P_e$ versus the block length $n$.
\item For fair comparison, calculate for each $n=1,2,3,\dotsc,n_{\max}$ the probability $P_\text{cb}$ of correctly transmitting $R\cdot K$ bits, where $R=0.2$ and where $K$ is the least common multiple of $n=1,2,3,\dotsc,n_{\max}$.

\emph{Hint:} First show that for each $n$, $P_\text{cb}=(1-P_e)^\frac{K}{n}$ where $P_e$ is the error probability of the block length $n$ code under consideration.
\item For each block length $n=1,2,3,\dotsc,n_{\max}$, plot information rate versus probability of error $1-P_\text{cb}$ for rates of $0,0.01,0.02,\dotsc,0.2$. Does $n>1$ improve the rate-reliability trade-off?	
\end{enumerate}

\myproblem The code $\setc=\{110,011,101\}$ is used for transmission over a binary erasure channel with input $X$, output $Y$ and erasure probability $\epsilon=\frac{1}{2}$. Each code word is used equally likely.

\begin{enumerate}
\item Calculate the block length and the rate in bits/channel use of the code $\setc$.
\item Suppose the codeword $110$ was transmitted. Calculate the distribution $P_{Y^3|X^3}(\cdot|110)$. 
\item Calculate the probability of correct decision of an ML decoder, given that $110$ was transmitted, i.e., calculate
\begin{align*}
\Pr(f_\text{ML}(Y^3)=110|X^3=110).
\end{align*}
\end{enumerate}

\myproblem Consider a channel with input alphabet $\setx=\{a,b,c\}$ and output alphabet $\sety=\{1,2,3\}$. Let $X$ and $Y$ be random variables with distribution $P_{XY}$ on $\setx\times\sety$. The probabilities are given by
\begin{center}
\begin{tabular}{rl}
$xy$&$P_{XY}(xy)$\\\hline
$a1$&$0.02$\\
$a2$&$0.02$\\
$a3$&$0$\\
$b1$&$0$\\
$b2$&$0.1$\\
$b3$&$0.15$\\
$c1$&$0.31$\\
$c2$&$0$\\
$c3$&$0.4$
\end{tabular}
\end{center}
\begin{enumerate}
\item Calculate the input distribution $P_X$.
\item Calculate the conditional distribution $P_{Y|X}(\cdot|i)$ on $\sety$ for $i=a,b,c$.
\item A decoder uses the function
\begin{align*}
f\colon\sety&\to\setx\\
1&\mapsto a\\
2&\mapsto b\\
3&\mapsto c.
\end{align*}
What is the probability of decoding correctly?
\item Suppose $X=a$ was transmitted. Using $f$, what is the probability of erroneous decoding? What are the respective probabilities of error if $X=b$ and $X=c$ are transmitted?
\item Suppose a MAP decoder is used. Calculate $f_\text{MAP}(1),f_\text{MAP}(2)$, and $f_\text{MAP}(3)$. With which probability does the MAP decoder decide correctly?
\item Suppose an ML decoder is used. Calculate $f_\text{ML}(1),f_\text{ML}(2)$, and $f_\text{ML}(3)$. With which probability does the ML decoder decide correctly?
\end{enumerate}

\myproblem
The binary input $X$ is transmitted over a channel and the binary output $Y=X+Z$ is observed at the receiver. The noise term $Z$ is also binary and addition is in $\fieldtwo$. Input $X$ and noise term $Z$ are independent. The input distribution is $P_X(0)=1-P_X(1)=1/4$ and the output distribution is $P_Y(0)=1-P_Y(1)=3/8$.
\begin{enumerate}
\item Calculate the noise distribution $P_Z$.
\item Is this channel a BSC?
\item Suppose an ML decoder is used. Calculate the ML decisions for $Y=0$ and $Y=1$, i.e., calculate $f_\text{ML}(0)$ and $f_\text{ML}(1)$. With which probability does the ML decoder decide correctly on average?
\item Is there a decoding function that achieves an average error probability that is strictly lower than the average error probability of the ML decoder?
\end{enumerate}

\myproblem Consider the following channel with input alphabet $\setx=\{0,1,2\}$ and output alphabet $\sety=\{0,1,2\}$. Each arrow indicates a transition, which occurs with the indicated probability. The input is distributed according to $P_{X}(0)=P_X(1)=\frac{1}{2}$.	
\begin{center}
\footnotesize
\includegraphics{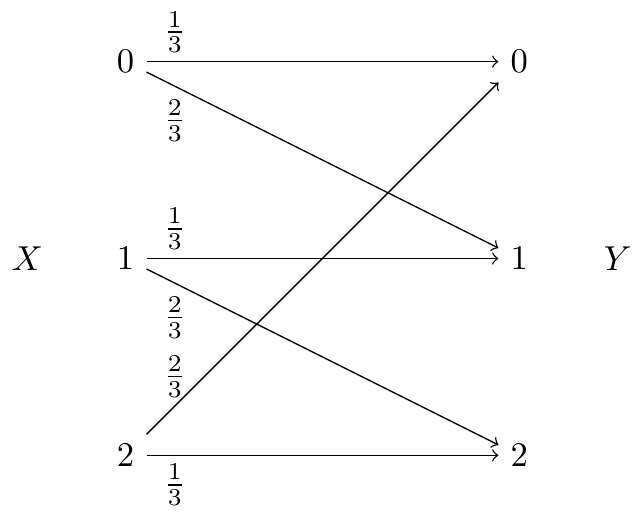}

\end{center}
\begin{enumerate}
\item Calculate the output distribution $P_Y$.
\item Suppose an ML decoder is used. Calculate the ML decisions for $Y=0$, $Y=1$, and $Y=2$.
\item Calculate the average probability of error of the ML decoder.
\item Show that for the considered scenario, the MAP decoder performs strictly better than the ML decoder. 
\item Show that the considered channel is equivalent to a ternary channel with an additive noise term $Z$.
\end{enumerate}
%
%
%

%% file: linear.tex
\chapter{Linear Block Codes}
\label{chap:linear}

In Chapter~\ref{chap:short}, we searched for codes for the BSC that perform well when an ML decoder is used. From a practical point of view, our findings were not very useful. First, the exhaustive search for good codes became infeasible for codes with more than $4$ code words and block lengths larger than $7$. Second, our findings suggested that increasing the block length further would lead to codes with better performance in terms of information rate and probability of error. Suppose we want a binary code with rate $1/2$ and block length $n=512$. Then, there are
\begin{align*}
|\setc|=2^{nR}=2^{256}
\end{align*}
different code words, each of size $512/8=64$ bytes. One gigabyte is $2^{30}$ bytes, so to store the whole code, we need
\begin{align}
2^{256}\cdot 64 \text{ bytes}=2^{262} \text{ bytes}=2^{232} \text{ gigabytes}.
\end{align}
To store this amount of data is impossible. Furthermore, this is the amount of data we need to store \emph{one} code, let alone the search for the best code with the desired parameters. This is the reason why we need to look at codes that have more structure so that they have a more compact description. \emph{Linear codes} have more structure and that is why we are going to study them in this chapter.

\section{Basic Properties}

Before we can define linear block codes, we first need to state some definitions and results of linear algebra.

\subsection{Groups and Fields}

\begin{definition}
A \emph{group}\index{group} is a set of elements $\setg=\{a,b,c,\dotsc\}$ together with an operation $\circ$ for which the following axioms hold:
\begin{enumerate}
\item\emph{Closure:} for any $a\in \setg$, $b\in \setg$, the element $a\circ b$ is in $\setg$.
\item\emph{Associative law:} for any $a,b,c\in \setg$, $(a\circ b)\circ c=a\circ(b\circ c)$.
\item\emph{Identity:} There is an identity element $0$ in $\setg$ for which $a\circ 0=0\circ a=a$ for all $a\in \setg$.
\item\emph{Inverse:} For each $a\in \setg$, there is an inverse $-a$ such that $a\circ(-a)=0$.
\end{enumerate} 
If $a\circ b=b\circ a$ for all $a,b\in \setg$, then $\setg$ is called \emph{commutative}\index{commutative group} or \emph{Abelian}\index{Abelian group}.
\end{definition}

\begin{example}\label{ex:linear:group}
Consider the binary set $\{0,1\}$ with the modulo-$2$ addition and multiplication specified by
\begin{center}
\begin{tabular}{c|cc}
+&0&1\\\hline
0&0&1\\
1&1&0
\end{tabular}
\hspace{2cm}
\begin{tabular}{c|cc}
$\cdot$&0&1\\\hline
0&0&0\\
1&0&1
\end{tabular}
\end{center}
It can be verified that $(+,\{0,1\})$ is an Abelian group. However, $(\cdot,\{0,1\})$ is not a group. This can be seen as follows. The identity with respect to $\cdot$ in $\{0,1\}$ is $1$, since $0\cdot 1=0$ and $1\cdot 1=1$. However, $0\cdot 0=0$ and $0\cdot 1=0$, i.e., the element $0$ has no inverse in $\{0,1\}$.
\end{example}
\begin{example}
The set of \emph{integers}\index{integers}\index{Z@\setz \see{integers}} $\setz=\{\dotsc,-2,-1,0,1,2,3,\dotsc\}$ together with the usual addition is an Abelian group. The set of \emph{positive integers}\index{positive integers} $\{1,2,3,\dotsc\}$, which is also called the set of \emph{natural numbers}\index{natural numbers}, is \emph{not} a group.
\end{example}

\begin{definition}
A \emph{field} is a set $\myfield$ of at least two elements, with two operations $+$ and $\cdot$, for which the following axioms are satisfied:
\begin{enumerate}
\item The set $\myfield$ forms an Abelian group (whose identity element is called $0$) under the operation $+$. $(\myfield,+)$ is called the \emph{additive group}\index{additive group of a field} of $\myfield$.
\item The operation $\cdot$ is associative and commutative on $\myfield$. The set $\myfield^*=\myfield\setminus\{0\}$ forms an Abelian group (whose identity element is called $1$) under the operation $\cdot$. $(\myfield\setminus\{0\},\cdot)$ is called the \emph{multiplicative group}\index{multiplicative group of a field} of $\myfield$.
\item \emph{Distributive law:} For all $a,b,c\in \myfield$, $(a+ b)\cdot c=(a\cdot c)+ (b\cdot c)$.
\end{enumerate}
\end{definition}

\begin{example}
Consider $\{0,1\}$ with ``$+$'' and ``$\cdot$'' as defined in Example~\ref{ex:linear:group}. $(\{0,1\},+)$ forms an Abelian group with identity $0$. $(\{1\},\cdot)$ is an Abelian group with identity $1$, so $(\{0,1\},+,\cdot)$ is a field. We denote it by $\fieldtwo$.
\end{example}
\begin{example}
The integers $\setz$ with the modulo-3 addition and multiplication specified by \begin{center}
\begin{tabular}{c|ccc}
+&0&1&2\\\hline
0&0&1&2\\
1&1&2&0\\
2&2&0&1
\end{tabular}
\hspace{2cm}
\begin{tabular}{c|ccc}
$\cdot$&0&1&2\\\hline
0&0&0&0\\
1&0&1&2\\
2&0&2&1\\
\end{tabular}
\end{center}
form a field, which we denote by $\field{F}_3$.
\end{example}
We study finite fields in detail in Section~\ref{sec:finite}.

\subsection{Vector Spaces}
\begin{definition}
A \emph{vector space}\index{vector space} $\setv$ over a field $\myfield$ is an Abelian group $(\setv,+)$ together with an additional operation ``$\cdot$'' (called the \emph{scalar multiplication}\index{scalar multiplication}) 
\begin{align}
\myfield\times \setv&\to \setv\\
(a,\vecv)&\mapsto a\cdot \vecv
\end{align}
that satisfies the following axioms:
\begin{enumerate}
\item $(a\cdot b)\cdot \vecv = a\cdot(b\cdot \vecv)$ for all $\vecv\in \setv$ and for all $a,b\in \myfield$.
\item $(a+b)\cdot \vecv = a\cdot \vecv+b\cdot \vecv$ for all $\vecv\in \setv$ and for all $a,b\in \myfield$.
\item $a\cdot(\vecv+\vecw)=a\cdot \vecv+ a\cdot \vecw$ for all $\vecv,\vecw\in \setv$ and for all $a\in \myfield$.
\item $1\cdot \vecv=\vecv$ for all $\vecv\in \setv$.
\end{enumerate}
Elements of $\setv$ are called \emph{vectors}\index{vector}. Elements of $\myfield$ are called \emph{scalars}\index{scalar}.
\end{definition}
\begin{example}\label{ex:linear:fn}
Let $n$ be a positive integer. The $n$-fold Cartesian product
\begin{align}
\myfield^n:=\underbrace{\myfield\times\myfield\times\dotsb\times\myfield}_{n\text{ times}}
\end{align}
with the operations
\begin{align}
(v_1,\dotsc,v_n)+(w_1,\dotsc,w_n) &:= (v_1+w_1,\dotsc,v_n+w_n)\\
a\cdot (v_1,\dotsc,v_n)&:=(a\cdot v_1,\dotsc,a\cdot v_n)
\end{align}
is the most important example of a vector space.
\end{example}

In the following definitions, $\setv$ is a vector space over $\myfield$, $\setg\subseteq \setv$ is a set of vectors and $n$ is a finite positive integer. 
\begin{definition}
A vector $\vecv\in \setv$ is \emph{linearly dependent}\index{linearly dependent} of the vectors in $\setg$ if there exist finitely many scalars $a_i\in \myfield$ and appropriate vectors $\vecw_i\in \setg$ such that
\begin{align}
\vecv=\sum_{i=1}^n a_i \vecw_i.
\end{align}
\end{definition}

\begin{definition}
$\setg$ is a \emph{generating set}\index{generating set} of $\setv$, if every vector $\vecv\in \setv$ is linearly dependent of $\setg$.
\end{definition}

\begin{definition}
The vectors $\vecv_1,\dots,\vecv_n$ are \emph{linearly independent}\index{linearly independent}, if for all $a_i\in \myfield$,
\begin{align}
\sum_{i=1}^n a_i \vecv_i=\veczero\;\Rightarrow\;\text{all $a_i$ are equal to zero}.
\end{align}
\end{definition}

\begin{definition}
The vectors $\vecv_1,\dotsc,\vecv_n$ form a \emph{basis}\index{basis} of $\setv$ if they are linearly independent and $\{\vecv_1,\dotsc,\vecv_n\}$ is a generating set of $\setv$.
\end{definition}

\begin{proposition}
A non-empty subset $\setu\subseteq \setv$ is itself a vector space if
\begin{align}
\vecv,\vecw\in \setu\;\Rightarrow\;a\cdot \vecv+b\cdot \vecw\in \setu,\quad\forall a,b\in \myfield.
\end{align}
$\setu$ is then called a \emph{subspace}\index{subspace} of $\setv$.
\end{proposition}
We state the following theorem without giving a proof.
\begin{theorem}\label{theo:linear:vs}
Let $\setv$ be a vector space over $\myfield$ with a basis $\setb$ and $n=|\setb|<\infty$. Any set of $n$ linearly independent vectors in $\setv$ forms a basis of $\setv$. The number $n=|\setb|$ is called the \emph{dimension}\index{dimension} of $\setv$. 
\end{theorem}

\subsection{Linear Block Codes}

\begin{definition}
An $(n,k)$ \emph{linear block code}\index{linear block code} over a field $\myfield$ is a $k$-dimensional subspace of the $n$-dimensional vector space $\myfield^n$.
\end{definition}
\begin{example}
The set $\setc=\{(0,0),(1,1)\}$ is a one-dimensional subspace of the two-dimensional vector space $\fieldtwo^2$. The set $\setc$ is called a \emph{binary linear block code}\index{binary linear block code}.
\end{example}
In the introductory paragraph of this chapter, we argued that in general, we would need $2^{232}$ gigabytes of storage to store a binary code with block length $n=512$ and $2^{256}$ code words. Now suppose the code is linear. Then its dimension is
\begin{align}
|\field{F}_2|^k\overset{!}{=}2^{256}\;\Rightarrow\;k=256.
\end{align}
By Theorem~\ref{theo:linear:vs}, the code is completely specified by $k$ linearly independent vectors in $\myfield_2^n$. Thus, we need to store $256$ code words of length $512$ to store the linear code. This amounts to
\begin{align}
512\cdot 256\cdot \frac{1}{8}=16\,384\text{ bytes}
\end{align}
which is the storage needed to store a $140\times 180$ pixel portrait photo in JPEG format.

The rate of an $(n,k)$ linear code is given by (see Problem~\ref{prob:linear:rate})
\begin{align}
R=\frac{k\log_2|\myfield|}{n}\quad\left[\frac{\text{bits}}{\text{code symbol}}\right].\label{eq:linear:rate}
\end{align}

\subsection{Generator Matrix}
\begin{definition}
Let $\setc$ be a linear block code. A matrix $\matg$ whose rows form a basis of $\setc$ is called a \emph{generator matrix}\index{generator matrix} for $\setc$. Conversely, the row space of a matrix $\matg$ with entries in $\myfield$ is called the \emph{code generated by} $\matg$.
\end{definition}
\begin{example}\label{linear:ex:nonuniqueG}
Consider the two matrices
\begin{align}
\matg_1=\bpm 1&0\\0&1\epm,\quad\matg_2=\bpm 1&0\\1&1\epm.
\end{align}
The rows of each of the matrices are vectors in some vector space over $\fieldtwo$. Since they are linearly independent, they form a basis of a vector space. By calculating all linear combinations of the rows, we find that both row spaces are equal to the vector space $\fieldtwo^2$.
\end{example}
Example~\ref{linear:ex:nonuniqueG} shows that different generator matrices can span the same vector space. In general, suppose now you have two codes specified by two generator matrices $\matg_1$ and $\matg_2$. A natural question is if these two codes are identical. To answer this question, we want to represent each linear block code by a unique canonical generator matrix and we would like to have a procedure that allows us to bring an arbitrary generator matrix into this canonical form.

The following \emph{elementary row operations}\index{elementary row operations} leave the row space of a matrix unchanged.
\begin{enumerate}
\item \emph{Row switching:} Any row $\vecv_i$ within the matrix can be switched with any other row $\vecv_j$:
\begin{align}
\vecv_i\leftrightarrow \vecv_j.
\end{align}
\item \emph{Row multiplication:} Any row $\vecv_i$ can be multiplied by any non-zero element $a\in \myfield$:
\begin{align}
\vecv_i\leftarrow a\cdot \vecv_i.
\end{align}
\item \emph{Row addition:} We can add a multiple of any row $\vecv_j$ to any row $\vecv_i$:
\begin{align}
\vecv_i\leftarrow \vecv_i+a\cdot \vecv_j.
\end{align} 
\end{enumerate}
With these three operations, we can bring any generator matrix into the so called \emph{reduced row echelon} form.
\begin{definition}
A matrix is in \emph{reduced row echelon} (RRE) \emph{form}\index{reduced row echelon form}\index{RRE form \see{reduced row echelon form}}, if it has the following three properties.
\begin{enumerate}
\item The leftmost nonzero entry in each row is $1$.
\item Every column containing such a leftmost $1$ has all its other entries equal to $0$.
\item If the leftmost nonzero entry in a row $i$ occurs in column $t_i$, then $t_1<t_2<\dotsb$.
\end{enumerate}
\end{definition}
We can now state the following important property of linear block codes.
\begin{theorem}\label{theo:linear:rre}
Every linear block code has a unique generator matrix in RRE form. This matrix can be obtained by applying elementary row operations to any matrix that generates the code.
\end{theorem}
The transformation into RRE form can be done efficiently by the \emph{Gaussian elimination}. Theorem~\ref{theo:linear:rre} gives us the tool we were seeking for: to check if two codes are identical, we first bring both generator matrices into RRE form. If the resulting matrices are identical, then so are the codes. Conversely, if the two generator matrices in RRE form differ, then they generate different codes.
\begin{example} The binary \emph{repetition code}\index{repetition code} is a $(n,1)$ linear block code over $\fieldtwo$ with the generator matrix
\begin{align}
\matg_\text{rep}=(\underbrace{111\dotsb 1}_{n\text{ times}}).
\end{align}
The code has only one generator matrix, which already is in RRE form.
\end{example}

\begin{example}
The $(7,4)$ \emph{Hamming code}\index{Hamming code} is a code over $\fieldtwo$ with the generator matrix in RRE form
\begin{align}
\matg_\text{ham} = \bpm 1&0&0&0&0&1&1\\
0&1&0&0&1&0&1\\
0&0&1&0&1&1&0\\
0&0&0&1&1&1&1
\epm.
\end{align}
\end{example}

\section{Code Performance}

In the previous section, we have defined linear block codes and we have stated some basic properties. Our goal is to analyze and design codes. In this section, we develop important tools to assess the quality of linear block codes.

\subsection{Hamming Geometry}

Consider an $(n,k)$ linear block code $\setc$ over some finite field $\myfield$.
\begin{definition}
The \emph{Hamming weight}\index{Hamming weight}\index{wh@\hw \see{Hamming weight}} of a code word $\vecv$ is defined as the number of non-zero entries of $\vecv$, i.e.,
\begin{align}
\hw(\vecv):=\sum_{i=1}^n\mathbbm{1}(v_i\neq 0).\label{eq:linear:defhw}
\end{align}
\end{definition}
The mapping $\mathbbm{1}$ in \eqref{eq:linear:defhw} is defined as
\begin{align}
\mathbbm{1}\colon \{\text{true},\text{false}\}&\to\{0,1\}\\
\text{true}&\mapsto 1\\
\text{false}&\mapsto 0.
\end{align}
The summation in \eqref{eq:linear:defhw} is in $\setz$. We illustrate this in the following example.
\begin{example}
Consider the code word $\vecv=(0,1,0,1)$ of some linear block code over $\fieldtwo$ and the codeword $\vecw=(0,2,0,1)$ of some linear block code over $\field{F}_3$. The Hamming weights of the two code words are given by
\begin{align}
\hw(\vecv)=\hw(\vecw)=0+1+0+1=2.
\end{align}
\end{example}

\begin{definition}
The \emph{Hamming distance}\index{Hamming distance}\index{dh@\hd \see{Hamming distance}} of two code words $\vecv$  and $\vecw$ is defined as the number of entries at which the code words differ, i.e.,
\begin{align}
\hd(\vecv,\vecw):=\sum_{i=1}^n\mathbbm{1}(v_i\neq w_i)=\hw(\vecv-\vecw).\label{eq:linear:defhd}
\end{align}
\end{definition}

The Hamming distance defines a metric on the vector space $\setc$, see Problem~\ref{prob:linear:metric}.

The \emph{minimum distance}\index{minimum distance}\index{d@\textit{d} \see{minimum distance}} of a linear code $\setc$ is defined as
\begin{align}
d:=\min_{\veca\neq\vecb\in\setc}\hd(\veca,\vecb).
\end{align}
It is given by
\begin{align}
d=\min_{\vecc\in\setc\setminus\veczero}\hw(\vecc)\label{eq:linear:minimumweight}
\end{align}
that is, the minimum distance of a linear code is equal to the minimum weight of the non-zero code words. See Problem~\ref{prob:linear:minimumweight} for a proof of this statement.

For an $(n,k)$ linear code $\setc$, we define $A_i$ as the number of code words with Hamming weight $i$, i.e.,
\begin{align}
A_i:=|\{\vecv\in\setc\colon\hw(\vecv)=i\}|.
\end{align}
We represent the sequence $A_0,A_1,A_2,\dotsc,A_n$ by the \emph{weight enumerator}\index{weight enumerator} 
\begin{align}
A(x):=\sum_{i=0}^n A_i x^i.
\end{align}
The weight enumerator $A(x)$ is a \emph{generating function}\index{generating function}, i.e., a formal power series in the indeterminate $x$.

Let $\vecv\in\setc$ be some code word. How many code words are in $\setc$ with Hamming distance $i$ from $\vecv$? To answer this question, we use the identity that is proved in Problem~\ref{prob:linear:coset}, namely
\begin{align}
\vecv+\setc=\setc.
\end{align}
We now have
\begin{align}
|\{\vecw\in\setc\colon\hd(\vecv,\vecw)=i\}|=&|\{\vecw\in\setc+\vecv\colon\hd(\vecv,\vecw)=i\}|\\
=&|\{\vecu\in\setc\colon\hd(\vecv,\vecu+\vecv)=i\}|\\
=&|\{\vecu\in\setc\colon\hw(\vecu)=i\}|\\
=& A_i.
\end{align}

\subsection{Bhattacharyya Parameter}

\begin{definition}
Let $P_{Y|X}$ be a DMC with binary input alphabet $\setx=\{0,1\}$ and output alphabet $\sety$. The \emph{channel Bhattacharyya parameter}\index{Bhattacharyya parameter} is defined as
\begin{align}
\beta:=\sum_{b\in\sety}\sqrt{P_{Y|X}(b|0)P_{Y|X}(b|1)}.
\end{align}
\end{definition}
\begin{example}
For a BSC with crossover probability $\delta$, the Bhattacharyya parameter is $\beta_\text{BSC}(\delta)=2\sqrt{\delta(1-\delta)}$.
\end{example}
The Bhattacharyya parameter is a measure for how ``noisy'' a channel is.

\subsection{Bound on Probability of Error}\label{sec:linear:bound}

Suppose we want to transmit code words of a linear code over a binary input channel and suppose further that we use an ML decoder to recover the transmitted code words from the channel output. The following theorem states an upper bound on the resulting average probability of error. 
\begin{theorem}\label{theo:linear:bhattacharyya}
Let $\setc$ be an $(n,k)$ binary linear code with weight enumerator $A$. Let $P_{Y|X}$ be a DMC with input alphabet $\setx=\{0,1\}$ and output alphabet $\sety$. Let $\beta$ be the Bhattacharyya parameter of the channel. Then the error probability of an ML decoder is bounded by
\begin{align}
P_\text{ML}\leq A(\beta)-1.
\end{align}
\end{theorem}
Before we give the proof, let's discuss the implication of this theorem. The bound is in terms of the weight enumerator of the code and the Bhattacharyya parameter of the channel. Let $d$ be the minimum distance of the considered code. We write out the weight enumerator.
\begin{align}
P_\text{ML}&\leq A(\beta)-1\\
&=\sum_{i=0}^n A_i \beta^i-1\\
&=1+\sum_{i=d}^n A_i \beta^i-1\\
&=A_d\beta^d+A_{d+1}\beta^{d+1}+\dotsb+A_n\beta^n.
\end{align}
By Problem~\ref{prob:linear:bhattacharyya}, if the channel is not completely useless, $\beta<1$. If $\beta$ is small enough, then the term $A_d\beta^d$ dominates the others. In this case, the minimum distance of the code determines the code performance. This is one of the reasons why a lot of research was done to construct linear codes with large minimum distance.
\begin{proof}[Proof of Theorem~\ref{theo:linear:bhattacharyya}]
 The code is $\setc=\{\vecv_1,\vecv_2,\dotsc,\vecv_{2^k}\}$. Suppose $\vecv_i\in\setc$ is transmitted. The probability of error is
\begin{align}
\Pr(f_\text{ML}(Y^n)\neq \vecv_i|X^n=\vecv_i)=\sum_{j\neq i}\underbrace{\Pr(f_\text{ML}(Y^n)=\vecv_j|X^n=\vecv_i)}_{:=P_{i\to j}}.
\end{align}
We next bound the probability $P_{i\to j}$ that $\vecv_i$ is erroneously decoded as $\vecv_j$. The ML decoder does \emph{not} decide for $\vecv_j$ if $P_{Y^n|X^n}(\vecw|\vecv_i)>P_{Y^n|X^n}(\vecw|\vecv_j)$. Define
\begin{align}
\sety_{ij}:=\bigl\{\vecw\in\sety^n\colon P_{Y^n|X^n}(\vecw|\vecv_i)\leq P_{Y^n|X^n}(\vecw|\vecv_j)\bigr\}.\label{eq:linear:setdef}
\end{align}
We have
\begin{align}
\Bigl\{\vecw\colon f_\text{ML}(\vecw)=\vecv_j\Bigr\}\subseteq\sety_{ij}.\label{eq:linear:setbound}
\end{align}
We can now bound $P_{i\to j}$ as
\begin{align}
P_{i\to j}&\oleq{a}\sum_{\vecw\in\sety_{ij}}P_{Y^n|X^n}(\vecw|\vecv_i)\nonumber\\
&\oleq{b}\sum_{\vecw\in\sety_{ij}}P_{Y^n|X^n}(\vecw|\vecv_i)\sqrt{\frac{P_{Y^n|X^n}(\vecw|\vecv_j)}{P_{Y^n|X^n}(\vecw|\vecv_i)}}\nonumber\\
&= \sum_{\vecw\in\sety_{ij}}\sqrt{P_{Y^n|X^n}(\vecw|\vecv_i)P_{Y^n|X^n}(\vecw|\vecv_j)}\nonumber\\
&\leq \sum_{\vecw\in\sety^n}\sqrt{P_{Y^n|X^n}(\vecw|\vecv_i)P_{Y^n|X^n}(\vecw|\vecv_j)}\nonumber\\
&\oeq{c} \sum_{\vecw\in\sety^n}\prod_{\ell=1}^n\sqrt{P_{Y|X}(w_\ell|v_{i\ell})P_{Y|X}(w_\ell|v_{j\ell})}\label{eq:linear:pijbound1}\\
&=\prod_{\ell=1}^n\sum_{b\in\sety}\sqrt{P_{Y|X}(b|v_{i\ell})P_{Y|X}(b|v_{j\ell})}.\label{eq:linear:pijbound2}
\end{align}
Inequality (a) follows by \eqref{eq:linear:setbound}, (b) follows by \eqref{eq:linear:setdef} and we used in (c) that the channel is memoryless. Note that the sum in \eqref{eq:linear:pijbound1} is over $\sety^n$ and the sum in \eqref{eq:linear:pijbound2} is over $\sety$. We evaluate the sum in \eqref{eq:linear:pijbound2}:
\begin{align}
\sum_{b\in\sety}\sqrt{P_{Y|X}(b|v_{i\ell})P_{Y|X}(b|v_{j\ell})}=\begin{cases}
1,&\text{if }v_{i\ell}=v_{j\ell}\\
\beta,&\text{otherwise}.
\end{cases}\label{eq:linear:pijcases}
\end{align}
The number of times the second case occurs is the Hamming distance of $\vecv_i$ and $\vecv_j$. We use \eqref{eq:linear:pijcases} in \eqref{eq:linear:pijbound2} and get
\begin{align}
P_{i\to j}\leq \beta^{\hd(\vecv_i,\vecv_j)}.
\end{align}
We can now bound the error probability of an ML decoder when $\vecv_i$ is transmitted by
\begin{align}
\Pr(f_\text{ML}(Y^n)\neq\vecv_i|X^n=\vecv_i)&=\sum_{j\neq i}P_{i\to j}\\
&\leq\sum_{j\neq i}\beta^{\hd(\vecv_i,\vecv_j)}\\
&=\sum_{\ell=1}^n A_\ell\beta^\ell\\
&=A(\beta)-A_0\\
&=A(\beta)-1.
\end{align}
The probability of error is now given by
\begin{align}
P_\text{ML}&=\sum_{\vecv\in\setc}P_{X^n}(\vecv)\Pr(f_\text{ML}(Y^n)\neq\vecv|X^n=\vecv)\\
&\leq \sum_{\vecv\in\setc}P_{X^n}(\vecv)[A(\beta)-1]\\
&=A(\beta)-1.
\end{align}
\end{proof}

\section{Syndrome Decoding}

Suppose we have a channel where the input alphabet is a field $\myfield_q$ with $|\myfield_q|=q$ elements and where for each input value $a\in\myfield_q$, the output is given by
\begin{align}
Y=a+Z.\label{eq:linear:qary}
\end{align}
The noise random variable $Z$ takes values in $\myfield_q$ according to the distribution $P_Z$. The addition in \eqref{eq:linear:qary} is in $\myfield_q$. Consequently, the output $Y$ also takes values in $\myfield_q$ and has the conditional distribution 
\begin{align}
P_{Y|X}(b|a)=P_Z(b-a).
\end{align}
The channel defined in \eqref{eq:linear:qary} is called a \emph{$q$-ary channel}\index{q@\textit{q}-ary channel}. If in addition, the noise distribution is of the form
\begin{align}
P_Z(a)=\begin{cases}
\delta,&a\neq 0\\
1-(q-1)\delta,&\text{otherwise}
\end{cases}\label{eq:linear:qarynoise}
\end{align}
then the channel is called a \emph{$q$-ary symmetric channel}\index{q@\textit{q}-ary symmetric channel}.
\begin{example}
Let the input alphabet of a channel be $\myfield_2$ and for $a\in\myfield_2$, define the output by
\begin{align}
Y=a+Z\label{eq:linear:symmetricbsc}
\end{align}
where $P_Z(1)=1-P_Z(0)=\delta$. The transition probabilities are
\begin{align}
&P_{Y|X}(0|0)=P_Z(0)=P_Z(0-0)=1-\delta\\
&P_{Y|X}(1|0)=P_Z(1)=P_Z(1-0)=\delta\\
&P_{Y|X}(0|1)=P_Z(1)=P_Z(0-1)=\delta\\
&P_{Y|X}(1|1)=P_Z(0)=P_Z(1-1)=1-\delta
\end{align}
where $X$ represents the channel input. We conclude that \eqref{eq:linear:symmetricbsc} defines a BSC with crossover probability $\delta$. The BSC is thus an instance of the class of $q$-ary symmetric channels.
\end{example}
The probability of a specific error pattern $\vecz$ on a $q$-ary symmetric channel is
\begin{align}
P_Z^n(\vecz)=\delta^{\hw(\vecz)}[1-(q-1)\delta]^{n-\hw(\vecz)}.
\end{align}
We define
\begin{align}
\text{a $q$-ary symmetric channel is \emph{not too noisy}}\Leftrightarrow\delta<1-(q-1)\delta. 
\end{align}
\index{not too noisy}
Suppose a $q$-ary symmetric channel is not too noisy. Then for two error patterns $\vecz_1$ and $\vecz_2$, we have
\begin{align}
P_{Z}^n(\vecz_1)>P_{Z}^n(\vecz_2)\Leftrightarrow \hw(\vecz_1)<\hw(\vecz_2).\label{eq:linear:weightcondition}
\end{align}
We formulate the ML decoder for a $q$-ary channel. Let $\setc$ be a (not necessarily linear) block length $n$ code over $\myfield_q$. Suppose the decoder observes $\vecy\in\myfield_q^n$ at the channel output. The ML decoder is
\begin{align}
\dml(\vecy)&=\argmax_{\vecc\in\setc}P_{Y|X}(\vecy|\vecc)\\
&\oeq{a}\argmax_{\vecc\in\setc}P_Z(\vecy-\vecc)\label{eq:linear:mlqary}
\end{align}
where (a) follows because the channel is $q$-ary. If in addition the channel is symmetric and not too noisy, we have by \eqref{eq:linear:weightcondition}
\begin{align}
\dml(\vecy)=\argmin_{\vecc\in\setc}\hw(\vecy-\vecc)\label{eq:linear:mldistance}
\end{align}
The expression \eqref{eq:linear:mldistance} has the following interpretation.
\begin{quote}
On a not too noisy $q$-ary symmetric channel, optimal decoding consists in searching for the code word that is closest in terms of Hamming distance to the observed channel output. 
\end{quote}
This observation suggests the construction of codes with large minimum distance, since then, only improbable error patterns of large weight could move the channel output too far away from the code word that was actually transmitted.

The rest of this section is dedicated to develop tools that allow us to implement the decoding rule \eqref{eq:linear:mlqary} efficiently. The resulting device is the so called \emph{syndrome decoder}.
\subsection{Dual Code}
\begin{definition}
For $\vecv,\vecw\in\myfield^n$, the scalar
\begin{align}
\sum_{i=1}^n v_iw_i=\vecv\vecw^T.
\end{align}
is called the \emph{inner product}\index{inner product} of $\vecv$ and $\vecw$.
\end{definition}
The inner product has the following properties. For all $a,b\in\myfield$ and $\vecv,\vecw,\vecu\in\myfield^n$:
\begin{align}
&\vecv\vecw^T=\vecw\vecv^T\\
&(a\vecv+b\vecw)\vecu^T=a\vecv\vecu^T+b\vecw\vecu^T\\
&\vecu(a\vecv+b\vecw)^T=a\vecu\vecv^T+b\vecu\vecw^T.
\end{align}
In the following, let $\setc\subseteq\myfield^n$ be a linear block code. 
\begin{definition}\label{def:linear:dual}
The \emph{dual code}\index{dual code} $\setc^\perp$ is the \emph{orthogonal complement}\index{orthogonal complement} of $\setc$, i.e.,
\begin{align}
\setc^\perp:=\bigl\{\vecv\in\myfield^n\colon \vecv\vecw^T=0\text{ for every }\vecw\in\setc\bigr\}.
\end{align}
\end{definition}
\begin{proposition}\label{prop:linear:dualtest}
Let $\matg$ be a generator matrix of $\setc$. Then $\vecv\in\setc^\perp\Leftrightarrow \vecv\matg^T=\veczero$.
\end{proposition}
\begin{proof}
``$\Rightarrow$'': If $\vecv\in\setc^\perp$, then $\vecv\vecw^T=0$ for all $\vecw\in\setc$. The rows of $\matg$ are in $\setc$, so $\vecv\matg^T=\veczero.$

``$\Leftarrow$'': Suppose $\vecv\matg^T=\veczero$. Let $\vecw\in\setc$. Since the rows of $\matg$ form a basis of $\setc$,
\begin{align}
\vecw=\vecu\matg
\end{align}
for some $\vecu\in\myfield^k$. We calculate
\begin{align}
\vecv\vecw^T&=\vecv(\vecu\matg)^T\\
&=\vecv\matg^T\vecu^T\\
&=\veczero\vecu^T\\
&=0\\
\Rightarrow &\vecv\in\setc^\perp.
\end{align} 
\end{proof}

\begin{proposition}
$\setc^\perp$ is a linear block code.
\end{proposition}
\begin{proof}
$\setc^\perp$ is a linear block code if it is a subspace of $\myfield^n$. By definition, $\setc^\perp\subseteq\myfield^n$. It remains to show that $\setc^\perp$ is closed under addition and scalar multiplication. To this end, let $\vecv,\vecw\in\setc^\perp$ and $a,b\in\myfield$. Then
\begin{align}
(a\vecv+b\vecw)\matg^T&=a\vecv\matg^T+b\vecw\matg^T\\
&=a\cdot 0+b\cdot 0\\
&=0\\
\Rightarrow &a\vecv+b\vecw\in\setc^\perp.
\end{align} 
\end{proof}
\begin{proposition}\ \label{prop:linear:dualdimension}
$\dim \setc+\dim\setc^\perp=n$.
\end{proposition}
\begin{proof}
1. Suppose $\dim\setc=k$. Statement 1. is true in general. We prove it for the special case when $\matg=[\mati_k,\matp]$, where $\mati_k$\index{I@\matik \see{identity matrix}} denotes the $k\times k$ \emph{identity matrix}\index{identity matrix} and $\matp$ is some $k\times (n-k)$ matrix.
\begin{align}
\vecv\in\setc^\perp&\Leftrightarrow \vecv\matg^T=\veczero\\
&\Leftrightarrow v_i+\sum_{j=1}^{n-k}p_{ij}v_{k+j}=0,\quad \forall i=1,2,\dotsc,k\\
&\Leftrightarrow v_i=-\sum_{j=1}^{n-k}p_{ij}v_{k+j},\quad \forall i=1,2,\dotsc,k.
\end{align}
Each $(n-k)$-tuple $(v_{k+1},\dotsc,v_n)$ determines a unique $(v_1,\dotsc,v_k)$ so that the resulting vector $\vecv$ fulfills the set of equations. Thus, a generator matrix for $\setc^\perp$ is
\begin{align}
[-\matp^T,\mati_{n-k}]
\end{align}
and the dimension of $\setc^\perp$ is $\dim\setc^\perp=n-k$.
\end{proof}

\begin{proposition}\label{prop:linear:dualperpperp}
$(\setc^\perp)^\perp=\setc$.
\end{proposition}
\begin{proof}
Let $\vecv\in\setc$. Then, for any $\vecw\in\setc^\perp$, $\vecv\vecw^T=0$, so $\setc\subseteq(\setc^\perp)^\perp$. Suppose $\dim\setc=k$. Then by Proposition~\ref{prop:linear:dualdimension}, $\dim(\setc^\perp)^\perp=n-(n-k)=k$, so $\setc=(\setc^\perp)^\perp$.
\end{proof}

\subsection{Check Matrix}

The notion of dual spaces allows for an alternative representation of linear block codes.
\begin{definition}\label{def:linear:checkmatrix}
A generator matrix $\mathh$ of $\setc^\perp$ is called a \emph{check matrix}\index{check matrix} of $\setc$.
\end{definition}
\begin{theorem}\label{theo:linear:checkmatrix}
If $\mathh$ is a check matrix for $\setc$ then $\setc=\{\vecv\in\myfield^n\colon \vecv\mathh^T=\veczero\}$.
\end{theorem}
\begin{proof}
Combining Proposition~\ref{prop:linear:dualtest} and Proposition~\ref{prop:linear:dualperpperp} proves the statement.
\end{proof}
\begin{theorem}\label{theo:linear:checkmatrixd}
The minimum distance of a code $\setc$ is equal to the minimum number of columns of the check matrix that sum up to $\veczero$.
\end{theorem}
\begin{proof}
See Problem~\ref{prob:linear:checkmatrixd}.
\end{proof}

\begin{theorem}
Suppose $\matg=[\mati,\matp]$. Then $\mathh=[-\matp^T,\mati_{n-k}]$.
\end{theorem}
\begin{proof}
The statement is shown in the proof of Proposition~\ref{prop:linear:dualdimension}.
\end{proof}
\begin{example}
Suppose we have an $(n,k)$ linear block code. We can represent it by a $k\times n$ generator matrix. By Theorem~\ref{theo:linear:checkmatrix}, it can alternatively be represented by a check matrix, which by Proposition~\ref{prop:linear:dualdimension} has size $(n-k)\times n$. If $k>n/2$, then the check matrix representation is more compact than the generator matrix representation.
\end{example}

\subsection{Cosets}
\label{subsec:cosets}

\begin{proposition}\label{prop:equivalence relation}
Let $\setg$ be a group and let $\setu\subseteq\setg$ be a subgroup of $\setg$. Then
\begin{align}
g_1\sim g_2 \Leftrightarrow g_1-g_2\in\setu
\end{align}
defines an \emph{equivalence relation}\index{equivalence relation} on $\setg$, i.e., it is \emph{reflexive}, \emph{transitive} and \emph{symmetric}. The equivalence classes are $\{g+\setu\colon g\in\setg\}$ and are called \emph{cosets}\index{coset} of $\setu$ in $\setg$.
\end{proposition}
\begin{proof}\ 

\emph{reflexive:} Let $g\in\setg$. Then $g-g=0\in\setu$, so $g\sim g$.

\emph{transitive:} For $g_1,g_2,g_3\in\setg$, suppose $g_1\sim g_2$ and $g_2\sim g_3$ then
\begin{align}
g_1-g_3&=(\underbrace{g_1-g_2}_{\in\setu}+\underbrace{g_2-g_3}_{\in\setu})\in\setu\\
\Rightarrow & g_1\sim g_3.
\end{align}

\emph{symmetric:} 
\begin{align}
g_1\sim g_2 &\Leftrightarrow g_1-g_2\in\setu\\
&\Rightarrow -(g_1-g_2)=g_2-g_1\in\setu\\
&\Rightarrow g_2\sim g_1.
\end{align}
\end{proof}
Since $``\sim$'' is an equivalence relation, we have for any $g_1,g_2\in\setg$
\begin{align}
\{g_1+\setu\} = \{g_2+\setu\}\quad\text{or}\quad \{g_1+\setu\} \cap \{g_2+\setu\}=\emptyset.
\end{align}
Furthermore, since $0\in\setu$, we have
\begin{align}
\bigcup_{g\in\setg}\{g+\setu\}=\setg.
\end{align}
The cardinality is 
\begin{align}
|\{g+\setu\}|=|\setu|
\end{align}
since $f(u):=g+u$ is an invertible mapping from $\setu$ to $\{g+\setu\}$. We conclude that the number of cosets is given by
\begin{align}
\frac{|\setg|}{|\setu|}.
\end{align}
In particular, this shows that if $\setu$ is a subgroup of $\setg$ then $|\setu|$ divides $|\setg|$.
\begin{example}
Let $\setc$ be an $(n,k)$ linear block code over $\myfield_q$. By definition, $\setc$ is a subspace of the vector space $\myfield_q^n$. In particular, $\setc$ together with the operation ``$+$'' is a subgroup of $\myfield_q^n$. Then for each $\vecv\in\myfield_q^n$, the coset $\{\vecv+\setc\}$ has cardinality $q^k$. The number of cosets is
\begin{align}
\frac{|\myfield_q^n|}{|\setc|}=\frac{q^n}{q^k}=q^{n-k}.
\end{align}
The $q^{n-k}$ cosets partition the vector space $\myfield_q^n$ into $q^{n-k}$ disjoint sets, each of which is of size $q^k$. One of the cosets is the code $\setc$.
\end{example}

\subsection{Syndrome Decoder}

Suppose some code word $\vecc$ from an $(n,k)$ linear code $\setc\subseteq\myfield_q^n$ is transmitted over a $q$-ary channel. Suppose the channel output $\vecy\in\myfield_q^n$ is observed by the decoder, i.e.,
\begin{align}
\vecy=\vecc+\vecz.
\end{align}
The decoder knows $\vecy$ and $\setc$. Since $\setc$ is a subspace of $\myfield_q^n$, it is in particular a subgroup of $\myfield_q^n$ with respect to addition. Thus, he knows that the error pattern $\vecz$ has to be in the coset
\begin{align}
\{\vecy+\setc\}.
\end{align}
Therefore, the ML decoder can be written as
\begin{align}
\dml(\vecy)=\vecy-\argmax_{\vecz\in\{\vecy+\setc\}}P_{Z^n}(\vecz).
\end{align}
For each of the $q^{n-k}$ cosets, the vector $\vecz$ that maximizes $P_Z(\vecz)$ can be calculated offline and stored in a lookup table. It remains to identify to which coset $\vecy$ belongs. To this end, we need the following property.
\begin{theorem}\label{theo:linear:syndromes}
Let $\setc$ be a linear code with check matrix $\mathh$ and let $\vecy_1$ and $\vecy_2$ be two vectors in $\myfield_q^n$. Then
\begin{align}
\{\vecy_1+\setc\}=\{\vecy_2+\setc\}\Leftrightarrow \vecy_1\mathh^T=\vecy_2\mathh^T.
\end{align}
\end{theorem}
\begin{proof}
We have
\begin{align}
\{\vecy_1+\setc\}=\{\vecy_2+\setc\}&\oiff{a} \vecy_1-\vecy_2\in\setc\\
&\oiff{b} (\vecy_1-\vecy_2)\mathh^T=0\\
&\Leftrightarrow \vecy_1\mathh^T=\vecy_2\mathh^T.
\end{align}
where (a) follows by the definition of cosets and where (b) follows by Theorem~\ref{theo:linear:checkmatrix}.
\end{proof}
For a vector $\vecy\in\myfield_q^n$, the vector
\begin{align}
\vecs=\vecy\mathh^T
\end{align}
is called the \emph{syndrome}\index{syndrome} of $\vecy$. Theorem~\ref{theo:linear:syndromes} tells us that we can index the $q^{n-k}$ cosets by the syndromes $\vecy\mathh^T\in\myfield_q^{n-k}$. The \emph{syndrome decoder}\index{syndrome decoder} now works as follows.
\begin{enumerate}
\item Calculate the syndrome $\vecs=\vecy\mathh^T$.
\item Choose $\hat{\vecz}$ in the $\vecs$th coset that maximizes $P_{Z^n}(\hat{\vecz})$.
\item Estimate the transmitted code word as $\hat{\vecx}=\vecy-\hat{\vecz}$. 
\end{enumerate}
\clearpage
\section{Problems}

\myproblem Let $\myfield$ be a field. Show that $a\cdot 0=0$ for all $a\in \myfield$.

\myproblem Prove the following statement: The vectors $\vecv_1,\dotsc,\vecv_n$ are linearly independent if and only if every vector $\vecv\in \setv$ can be represented at most in one way as a linear combination
\begin{align}
\vecv = \sum_{i=1}^n a_i \vecv_i.
\end{align}

\myproblem\label{prob:linear:dim} 
\begin{enumerate}
\item Give a basis for the vector space $\fieldtwo^n$ (see Example~\ref{ex:linear:fn}). 
\item What is the dimension of $\fieldtwo^n$? 
\item How many vectors are in $\fieldtwo^n$?
\end{enumerate}

\myproblem\label{prob:linear:rate} What is the rate of an $(n,k)$ linear block code over a finite field $\field{F}$?

\myproblem Show the following implication:
\begin{align}
\text{$\setc$ is a linear code}\Rightarrow \veczero \in\setc.
\end{align}
That is, every linear code has the all-zero vector as a code word.

\myproblem \label{prob:linear:metric} A \emph{metric}\index{metric} on a set $\seta$ is a real-valued function $d$ defined on $\seta\times\seta$ with the following properties. For all $a,b,c$ in $\seta$:
\begin{enumerate}
\item $d(a,b)\geq 0$.
\item $d(a,b)= 0$ if and only if $a=b$.
\item $d(a,b)=d(b,a)$ (symmetry).
\item $d(a,c)\leq d(a,b)+d(b,c)$ (triangle inequality).
\end{enumerate}

Let $\setc$ be an $(n,k)$ linear block code over some finite field $\field{F}$. Show that the Hamming distance $\hd$ defines a metric on $\setc$.

\myproblem \label{prob:linear:coset} Let $(\setg,+)$ be a group and for some $a\in \setg$, define $a+\setg=\{a+b\colon b\in \setg\}$.  Show that
\begin{align}
a+\setg=\setg.
\end{align} 

\myproblem\label{prob:linear:minimumweight} Let $\setc$ be a linear block code. Define
\begin{align}
A_i:=&|\{\vecx\in\setc\colon\hw(\vecx)=i\}|\\
A_i(\vecx):=&|\{\vecy\in\setc\colon\hd(\vecx,\vecy)=i\}|.
\end{align}
\begin{enumerate}
\item Show that
\begin{align}
A_i(\vecx)=A_i,\text{ for all }\vecx\in\setc.\label{eq:short:prob:ai}
\end{align}
\item Show that \eqref{eq:short:prob:ai} implies \eqref{eq:linear:minimumweight}.
\end{enumerate}

\myproblem \label{prob:linear:bhattacharyya} Let $\beta$ be the Bhattacharyya parameter of some DMC. Show that
\begin{align}
0\oleq{a} \beta\oleq{b} 1.
\end{align}
When do we have equality in (a) and (b), respectively?

\myproblem \label{prob:linear:checkmatrixd} Prove Theorem~\ref{theo:linear:checkmatrixd}.

\myproblem
A binary code $\mathcal{C}$ of length $n=5$ is used. All codewords contain exactly three $1$'s.
\begin{enumerate}
\item What is the size $|\mathcal{C}|$ of the code, i.e., how many codewords are there?
\item Can this code be linear? Give two reasons!
\item List all codewords of this code.
\end{enumerate}

\myproblem
Consider the following non-linear code
\begin{align*}
 \mathcal{C}=\{0000,0110,0001,1111\}.
\end{align*}
\begin{enumerate}
  \item Why is $\mathcal{C}$ non-linear?
\item Determine its rate $R$.
\item Assume a codeword $\vecx\in\mathcal{C}$ is transmitted over a Binary Symmetric Channel (BSC) with crossover probability $\delta<0.5$ and $\vecy=1101$ was received.
Perform an ML decoding to obtain an estimate $\hat{\vecx}$ of $\vecx$.
\item Find a linear block code $\mathcal{C}'$ that contains all codewords from $\mathcal{C}$.
Determine the dimension $k$ and the rate $R$ of $\mathcal{C}'$.
\end{enumerate}

\myproblem\label{short:prob:hammingcode} The binary \emph{repetition code} $\mathcal{C}_\text{rep}$ is a $(1,n)$ linear block code over $\fieldtwo$. 
\begin{enumerate}
\item What is the dimension of $\mathcal{C}_\text{rep}$?
\item Provide a generator matrix of $\mathcal{C}_\text{rep}$.
\item Calculate the Hamming weights of all codewords.
\end{enumerate}
The $(7,4)$ \emph{Hamming code} is a code over $\fieldtwo$ with generator matrix
\begin{align}
\matg_\text{ham} = \bpm 1&0&0&0&0&1&1\\
0&1&0&0&1&0&1\\
0&0&1&0&1&1&0\\
0&0&0&1&1&1&1
\epm.
\end{align}

\begin{enumerate}
\setcounter{enumi}{3}
\item For the BSC with crossover probability $\delta<0.5$, write a Matlab function that implements the ML decoder in the form
\begin{align}
d_\text{ML}(\vecy)=\argmin_{\vecc\in\mathcal{C}}\hd(\vecy,\vecc)
\end{align}
\item Use the repetition code for $n=1,2,\dotsc,7$ and the Hamming code for a BSC with crossover probability $\delta=0.01,0.02,\dotsc,0.4$. Plot the crossover probability $\delta$ in horizontal direction and the error probability of an ML decoder in the vertical direction for each of the codes. Use a logarithmic scale for the error probability.
\end{enumerate}

\myproblem For the binary $(n,1)$ repetition code, determine a check matrix.

\myproblem The generator matrix $\matg$ is given by
\begin{align*}
\matg=\bbm 1&1&1&0&1&0\\
1&0&0&1&1&1\\
0&0&1&0&1&1\ebm
\end{align*}
\begin{enumerate}
\item Find a generator matrix in reduced row echelon form. 
\item Find a check matrix.
\end{enumerate}

\myproblem Generator and check matrix of a binary code are given by
\begin{align*}
\matg=\bbm 1& 0& 1& 0& 1& 1\\
0& 1& 1& 1& 0& 1\\
0& 1& 1& 0& 1& 0\ebm,
\qquad
\mathh=\bbm
1&1&0&1&1&0\\
1&0&1&0&1&1\\
0&1&0&0&1&1
\ebm.
\end{align*}
Verify that $\mathh$ is a check matrix for the code generated by $\matg$.

\myproblem Let $\mathh$ be the check matrix for an $(n,k)$ linear code $\setc$. Let $\setc'$ be the extended code whose check matrix $\mathh'$ is formed by
\begin{align*}
\mathh'=\bbm
0&&\\
\vdots&\boldsymbol{H}\\
0&&\\
1&\dotsb&1
\ebm
\end{align*}
\begin{enumerate}
\item Show that every codeword in $\setc'$ has even weight.
\item Show that $C'$ can be obtained from $\setc$ by adding to each codeword an extra check bit called the \emph{overall parity bit}.
\item Let $\matg$ be a generator matrix of $\setc$. Specify a generator matrix of $\setc'$.
\end{enumerate}

\myproblem Let $\setc$ be a binary linear code with both even- and odd-weight codewords. Show that the number of even-weight codewords is equal to the number of odd-weight codewords.

\myproblem\label{short:prob:hammingcode2} Let's consider again the $(7,4)$ Hamming code from Problem~\ref{short:prob:hammingcode}. A BSC with crossover probability $\delta$ can be modeled as an additive channel over $\fieldtwo$ by
\begin{align}
Y=X+Z
\end{align}
with output $Y$, input $X$, and noise term $Z$. The noise distribution is $P_Z(1)=1-P_Z(0)=\delta$. The addition is in $\fieldtwo$. The transition probabilities describing the channel is $P_{Y|X}(b|a)=P_Z(b-a)$ for $a,b\in\fieldtwo$. The goal of this and the next two problems is the design and analysis of an efficient ML decoder for $\setc_\text{ham}$ when used on the BSC with crossover probability $0\leq \delta < 0.5$.

\begin{enumerate}
\item Construct a check matrix $\mathh$ for $\matg_\text{ham}$.
\item Suppose a code word from $\setc_\text{ham}$ was transmitted over the channel and $\vecy\in\fieldtwo^7$ is observed at the output. As we have shown in class, the ML decoder decides for the most probable error pattern $\vecz$ in the coset $\{\vecy+\setc_\text{ham}\}$. Furthermore, we have shown that 
\begin{align}
\vecz\in\{\vecy+\setc_\text{ham}\}\Leftrightarrow \vecz\mathh^T=\vecy\mathh^T.
\end{align}
For each syndrome $\vecs\in\fieldtwo^{n-k}$, find the most probable error pattern $\vecz$ with $\vecz\mathh^T=\vecs$. List all syndrome--error pattern pairs in a table.
\item An efficient ML decoder is
\begin{enumerate}
\item[i.] $\vecs=\vecy\mathh^T$.
\item[ii.] $\hat{\vecz}=f(\vecs)$.
\item[iii.] $\hat{\vecx}=\vecy-\hat{\vecz}$.
\end{enumerate}
The function $f$ performs a table lookup. Implement this decoder in Matlab.
\end{enumerate}

\myproblem{(Problem~\ref{short:prob:hammingcode2} continued)}\label{short:prob:hammingcode2:b}
\begin{enumerate}
\item Calculate all cosets for $\setc_\text{ham}$.
\item Show the following: For a BSC with crossover probability $\delta<0.5$,
\begin{align}
P_{Z^n}(\vecz_1)\geq P_{Z^n}(\vecz_2)\Leftrightarrow \hw(\vecz_1)\leq \hw(\vecz_2).\label{eq:weightproperty}
\end{align}
Generalize this to $q$-ary symmetric channels by showing that \eqref{eq:weightproperty} holds if the $q$-ary symmetric channel is not too noisy, i.e., if it fulfills (2.63) in the lecture notes.
\item List the error patterns that can be corrected by your ML decoder. Determine their weights.
\item Show the following: For a BSC with $\delta<0.5$, a code $\setc$ with minimum code word distance $\dmin$ can correct all error patterns $\vecz$ with weight 
\begin{align}
\hw(\vecz)\leq\frac{\dmin-1}{2}.\label{eq:guarantee}
\end{align}
\item Consider a binary code with code word length $n$ and minimum distance $\dmin$. Show that the error probability of an ML decoder is for a BSC with $\delta<0.5$ bounded by
\begin{align}
P_e \leq 1-\sum_{i=0}^{\lfloor\frac{\dmin-1}{2}\rfloor} \delta^i(1-\delta)^{n-i}.
\end{align} 
\end{enumerate}

\myproblem{(Problem~\ref{short:prob:hammingcode2} continued)}
\begin{enumerate}
\item An encoder for $\setc_\text{ham}$ is $\vecu\mapsto\vecx=\vecu\matg_\text{ham}$ where $\vecu\in\fieldtwo^4$. Use this encoder and your decoder from Problem~\ref{short:prob:hammingcode2}. How can you calculate an estimate $\hat{\vecu}$ from your code word estimate $\hat{\vecx}$?
\item Simulate data transmission over a BSC with $\delta=0.01,0.02,\dotsc,0.4$. Let the data bits $U^k$ be uniformly distributed on $\fieldtwo^4$. Use Monte Carlo simulation to estimate the probability of error. Plot estimates both for the code word error probability $\Pr(X^n\neq \hat{X}^n)$ and the information word error probability $\Pr(U^k\neq \hat{U}^k)$.
\item Add the bound from Problem~\ref{short:prob:hammingcode2:b}.5 to the plot.
\end{enumerate}

\myproblem
\begin{enumerate}
\item Show the following: If an $(n,k)$ binary linear block code contains the all one code word $\vecone$ then $A_i=A_{n-i}$, i.e., the number of code words of weight $i$ is equal to the number of code words of weight $n-i$ for all $i=0,1,\dotsc,n$.
\item Consider a code that has the all one vector as a code word. Suppose the code words are mapped to a signal. The duration of one binary symbol is $1$ second and $0\mapsto 1$ Volt, $1\mapsto -1$ Volt. Suppose further that the code words are used equally likely. The voltage is measured over a resistance of $1\,\Omega$. What is the average direct current (DC) through the resistance when many codewords are transmitted successively? 
\end{enumerate}

\myproblem
Your mission is to transmit 1 bit over a binary symmetric channel with crossover probability $\delta=1/4$. You use the code
\begin{align*}
\setc=\{110,001\}.
\end{align*}
\begin{enumerate}
\item Is your code linear?
\item How many errors can a minimum distance decoder correct?
\item Specify a linear code that has the same error correcting capability as $\setc$.
\item Specify a check matrix for your linear code.
\item Calculate the look up table of a syndrome decoder for your linear code.
\item Decode the observation $110$ using your syndrome decoder. 
\end{enumerate}

\myproblem
The generator matrix of a binary linear code $\setc$ is given by
\begin{align*}
\matg = \bbm
0&1&0&1\\
1&0&1&0
\ebm.
\end{align*}
\begin{enumerate}
\item Calculate all code words of $\setc$.
\item Show that the generator matrix $\matg$ is also a check matrix of $\setc$.
\end{enumerate}
The code is used on a BSC with crossover probability $\delta<\frac{1}{2}$. Each code word is used equally likely. 
\begin{enumerate}
\addtocounter{enumi}{2}
\item Calculate the rate in bits per channel use.
\item Calculate the syndrome for each error pattern of weight one. Which weight one error patterns can a syndrome decoder surely correct?
\item Add a column to the generator matrix such that the syndrome decoder can correct all weight one error patterns.
\end{enumerate}

\myproblem
The generator matrix of a binary linear code $\setc$ is given by
\begin{align*}
\matg = \bbm
1&1&0\\
1&1&1
\ebm.
\end{align*}
\begin{enumerate}
\item Calculate all code words of $\setc$.
\item Calculate the check matrix of $\setc$.
\item Show that the dual code of $\setc$ in $\fieldtwo^3$ is a subcode of $\setc$.
\item Do $(1,0,0)$ and $(0,1,0)$ belong to the same coset of $\setc$ in $\fieldtwo^3$? 
\end{enumerate}
The code is used on a BSC with crossover probability $\delta<\frac{1}{2}$. Each code word is used equally likely. 
\begin{enumerate}
\addtocounter{enumi}{4}
\item The transmitted code word is corrupted by the error pattern $(0,0,1)$. Does the syndrome decoder decode correctly?
\end{enumerate}

\myproblem
The generator matrix of a binary linear code $\setc$ is given by
\begin{align*}
\matg = \bbm
1&1&1&1\\
1&0&1&0\\
1&1&0&0
\ebm.
\end{align*}
\begin{enumerate}
\item What is the minimum distance of $\setc$?
\item Calculate a check matrix of $\setc$.
\end{enumerate}
Let $\vecc=(c_1,c_2,c_3,c_4)\in\setc$ be a codeword. The first entry $c_1$ is transmitted over $\text{BSC}_1$ with crossover probability $\delta_1=0.5$ and the bits $c_2,c_3,c_4$ are transmitted over $\text{BSC}_2$ with crossover probability $\delta_2=0.1$.
\begin{enumerate}
\addtocounter{enumi}{2}
\item Suppose the channel outputs are $\vecy=(1,1,1,0)$. Calculate its syndrome.
\item Calculate the coset of $\setc$ to which $\vecy$ belongs.
\item A syndrome decoder decodes $\vecy$. What is its codeword estimate? \emph{Hint: keep in mind that $\text{BSC}_1$ and $\text{BSC}_2$ have different crossover probabilities.} 
\end{enumerate}

%% file: cyclic.tex
\chapter{Cyclic Codes}
\label{chap:cyclic}

This chapter is about a subclass of linear codes, which is called \emph{cyclic codes}. The purpose of this chapter is threefold. First, we want to get familiar with polynomials, since these are essential for the next two chapters of this course. Second, we establish basic properties of cyclic codes, which again are going to be very useful in the upcoming chapters. Finally, we show how very efficient encoders can be built for cyclic codes.

\section{Basic Properties}

\subsection{Polynomials}
\begin{definition}\label{def:polynomials} A \emph{polynomial}\index{polynomial} $f(x)$ of degree $m$ over a field $\field{F}$ is an expression of the form
\begin{align}
f(x)=f_0+f_1x+f_2x^2+\dotsb+f_mx^m
\end{align}
where $f_i\in\field{F}, 0\leq i \leq m$, and $f_m\neq 0$. The \emph{null polynomial}\index{null polynomial} $f(x)=0$ has degree $-\infty$. The set of all polynomials over $\field{F}$ is denoted by $\field{F}[x]$.
\end{definition}
\begin{example}
Let's consider polynomials over $\fieldtwo$. According to Definition~\ref{def:polynomials}, the polynomial $1+x$ has degree $1$, the polynomial $1$ has degree $0$ and the polynomial $0$ has degree $-\infty$. The product $(1+x)(1+x)=1+x^2$ has degree $2$ and the product $(1+x)\cdot 0=0$ has degree $-\infty$.
\end{example}
\subsubsection{Modulo Arithmetic}
\index{modulo arithmetic}\index{mod@\mod \see{modulo arithmetic}}
In this chapter, we extensively need division by a polynomial. Given are two polynomials $p(x)$ (the ``dividend'') and $q(x)\neq 0$ (the ``divisor''). Then there exist unique polynomials $m(x)$ (the ``quotient'') and $r(x)$ (the ``remainder'') such that
\begin{align}
p(x)=q(x)m(x)+r(x),\quad\text{with }\deg r(x)<\deg q(x).\label{eq:cyclic:mod}
\end{align}
The expression ``$p(x)\mod q(x)$'' is defined as the remainder $r(x)$ in \eqref{eq:cyclic:mod}. The polynomials $m(x)$ and $r(x)$ can be calculated by \emph{polynomial long division}\index{polynomial long division}.
\begin{example}
Let $p(x)=1+x^3+x^4$ and $q(x)=1+x+x^3$ be two polynomials over $\fieldtwo$. Then
\begin{align*}
\begin{array}{rrrrrr}
&1&+x&&&\\
\cline{2-6}
1+x+x^3)&1&&&+x^3&+x^4\\
&&x&+x^2&&+x^4\\
\cline{2-6}
&1&+x&+x^2&+x^3&\\
&1&+x&&+x^3&\\
\cline{2-5}
&&&x^2&&
\end{array}
\end{align*}
Thus, $m(x)=1+x$ and $r(x)=x^2$, and we can write $p(x)$ as
\begin{align*}
p(x)&=1+x^3+x^4\\
&=(1+x)(1+x+x^3)+x^2\\
&=m(x)q(x)+r(x).
\end{align*}
In particular, we have shown
\begin{align}
(1+x^3+x^4)\mod (1+x+x^3)=x^2.
\end{align}
\end{example}
\subsection{Cyclic Codes}
\begin{definition}
Let $\vecv=(v_0,v_1,\dotsc,v_{n-1})$ be a vector. The vector $\vecw$ is a \emph{cyclic shift}\index{cyclic shift} of $\vecv$ if for some integer $k$
\begin{align}
\forall j=0,1,\dotsc,n-1\colon w_j=v_{(j-k)\mod n}.
\end{align}
\end{definition}
\begin{example}
The cyclic shifts of the vector $(a,b,c)$ are $(a,b,c)$, $(c,a,b)$, and $(b,c,a)$ where the entries of the original vector are shifted to the right by $k=0$, $1$, and $2$ entries, respectively.
\end{example}
Let $\setc$ be a linear code with block length $n$. Let $\vecc=(c_0,c_1,\dotsc,c_{n-1})$ be a code word. We represent it by its generating function
\begin{align}
c(x)=\sum_{i=0}^{n-1}c_i x^i.
\end{align}
We say $\setc$ is a \emph{cyclic code}\index{cyclic code} if all cyclic shifts of $\vecc$ are also code words. A code word in $\setc$ of least non-negative degree is called a \emph{generator polynomial}\index{generator polynomial}.

Let $g(x)$ be the generator polynomial of $\setc$. The following properties hold.
\begin{enumerate}
\item Let $c(x)$ be a code word and $c^{(i)}(x)$ the code word that results from a cyclic shift of the entries of $c(x)$ to the right by $i$ positions. Then
\begin{align}
c^{(i)}(x)=x^i c(x)\mod(x^n-1).
\end{align}
\item If $c(x)$ is a code word in $\setc$, then for any polynomial $p(x)$, $p(x)c(x)\mod(x^n-1)$ is also a code word.
\item If $g(x)$ is a generator polynomial, then $g_0\neq 0$.
\item If $p(x)$ is a polynomial such that $p(x)\mod (x^n-1)$ is a code word, then $g(x)$ divides $p(x)$.
\item A polynomial $g(x)$ with $g_0\neq 0$ is a generator polynomial of a cyclic code with code word length $n$ if and only if $g(x)$ divides $x^n-1$.
\item The dimension of $\setc$ is $n-\deg g(x)$.
\item Let $h(x)$ be a polynomial with
\begin{align}
h(x)g(x)=x^n-1.\label{eq:cyclic:hx}
\end{align}
Then
\begin{align}
c(x)\mod (x^n-1)\in\setc\Leftrightarrow h(x)c(x)\mod (x^n-1)=0.
\end{align}
The polynomial $h(x)$ is called a \emph{check polynomial}\index{check polynomial}.
\end{enumerate}

\subsection{Proofs}

\subsubsection{Property 1}
We show Property 1 for $i=1$. We have
\begin{align}
c^{(1)}(x)&=c_{n-1}+c_0 x+\dotsb+c_{n-2}x^{n-1}\\
xc(x)&=\phantom{c_{n-1}+} c_0 x+\dotsb+c_{n-2}x^{n-1}+c_{n-1}x^n.
\end{align}
Therefore,
\begin{align}
&c^{(1)}(x)= x c(x)-c_{n-1}(x^n-1)\\
\Rightarrow &x c(x)=c_{n-1}(x^n-1)+c^{(1)}(x)\\
\Rightarrow &c^{(1)}(x)=xc(x)\mod(x^n-1).
\end{align}
For $i>1$, the property follows by repeatedly applying the property for $i=1$. 
\subsubsection{Property 2}
We have
\begin{align}
p(x)c(x)\mod (x^n-1)&=\sum_{i=0}^{\deg p(x)} p_ix^i\cdot c(x)\mod(x^n-1)\\
&=\underbrace{\sum_{i=0}^{n-1} p_i\underbrace{c^{(i)}(x)}_{(\star)}}_{(\star\star)}.
\end{align}
$(\star)\in\setc$ follows by Property 1 and $(\star\star)\in\setc$ follows because $\setc$ is linear.
\subsubsection{Property 3}
Suppose $g_0=0$. Then $x^{-1}g(x)\in\setc$ and $\deg x^{-1}g(x)=\deg g(x)-1$. By definition, $g(x)$ is the code word of lowest degree. This is a contradiction, thus $g_0\neq 0$. 
\subsubsection{Property 4}
Suppose $g(x)\nmid p(x)$. Then $p(x)=g(x)m(x)+r(x)$  for some polynomials $m(x)$ and $r(x)$ with $\deg r(x)<\deg g(x)$. By assumption, $p(x)\mod(x^n-1)\in\setc$ and by Property 2, $g(x)m(x)\mod(x^n-1)\in\setc$. Since $\setc$ is linear, also 
\begin{align*}
r(x)=p(x)\mod(x^n-1)-m(x)g(x)\mod(x^n-1)\in\setc.
\end{align*}
 This contradicts that $g(x)$ is the code word of lowest degree. Thus 
\begin{align*}
p(x)\mod(x^n-1)\in\setc\Rightarrow g(x)\mid p(x).
\end{align*}
\subsubsection{Property 5}
``$\Rightarrow$'': Since $\setc$ is linear, $0=(x^n-1)\mod (x^n-1)\in\setc$. Thus, by Property 4, $g(x)\mid (x^n-1)$.

``$\Leftarrow$'': Suppose $g(x)\mid(x^n-1)$. We need to construct a cyclic code with $g(x)$ as the code word of least degree. Define
\begin{align}
\setc:=\Bigl\{c(x)\colon c(x)=m(x)g(x)\mod(x^n-1),\;m(x)\in\field{F}[x]\Bigr\}.
\end{align}
The set $\setc$ is linear since for two polynomials $p(x),q(x)\in\field{F}[x]$ also $p(x)+q(x)\in\field{F}[x]$. Furthermore, suppose $c(x)\in\setc$. Then
\begin{align}
c^{(i)}(x)&=x^i c(x)\mod(x^n-1)\\
&=x^i m(x)g(x)\mod(x^n-1)\\
&=m'(x)g(x)\mod(x^n-1)\in\setc.
\end{align}
Thus, $\setc$ is cyclic. Suppose $c(x)$ is a code word. We show that $c(x)$ is divisible by $g(x)$, i.e., $g(x)$ is indeed the code word of least weight and thereby the generator polynomial of $\setc$. We have
\begin{align}
c(x)\mod g(x)&=[m(x)g(x)\mod(x^n-1)]\mod g(x)\\
&=m(x)g(x)\mod g(x)\\
&=0
\end{align}
where we used $p(x)\mid q(x)\Rightarrow [r(x)\mod q(x)]\mod p(x)=r(x)\mod p(x)$ and $g(x)\mid(x^n-1)$.
\subsubsection{Property 6}
By Property 4, if $c(x)\in\setc$, then there exists a polynomial $m(x)$ with $\deg m(x)\leq n-\deg g(x)-1$ such that $c(x)=m(x)g(x)$. Thus
\begin{align}
c(x)&=\sum_{i=0}^{n-\deg g(x)-1}m_i x^i g(x)\\
&=\sum_{i=0}^{n-\deg g(x)-1}m_i g^{(i)}(x).
\end{align}
Thus $\setb=\{g^{(i)}(x),i=0,\dotsc,n-\deg g(x)-1\}$ spans $\setc$. Furthermore, the $g^{(i)}(x)$ are linearly independent. Thus, $\setb$ is a basis of $\setc$ and the dimension of $\setc$ is $|\setb|=n-\deg g(x)$.

\subsubsection{Property 7}

``$\Rightarrow$'': Suppose $c(x)\mod (x^n-1)\in\setc$. Then by Property 4, $c(x)=m(x)g(x)$ for some polynomial $m(x)$. Then
\begin{align}
c(x)h(x)\mod (x^n-1)&=m(x)g(x)h(x)\mod (x^n-1)\\
&=m(x)g(x)h(x)\mod [g(x)h(x)]\\
&=0.
\end{align}
``$\Leftarrow$'': We have
\begin{align}
&c(x)h(x)\mod (x^n-1)=0\\
\Rightarrow &c(x)h(x)\mod [g(x)h(x)]=0\\
\Rightarrow &c(x)\mod g(x)=0\\
\Rightarrow &c(x)\in\setc.
\end{align}

\section{Encoder}
\label{sec:cyclic:encoder}
\subsection{Encoder for Linear Codes}
\begin{definition}
Let $\setc$ be an $(n,k)$ linear code over $\myfield$. An \emph{encoder}\index{encoder} enc is a bijective mapping
\begin{align}
\text{enc}\colon\myfield^k\to\setc.
\end{align} 
\end{definition}
An encoder does nothing but indexing the code words in $\setc$ by vectors in $\myfield^k$ that represent data to be transmitted. Let $\matg$ be a generator matrix of $\setc$. Let $\vecu\in\myfield^k$. A natural definition of an encoder is the mapping
\begin{align}
\vecu\mapsto\vecu\matg.
\end{align}
If the generator matrix is of the form
\begin{align}
\matg=[\mati_k,\matp]
\end{align}
then
\begin{align}
\vecu\mapsto[\vecu,\vecu\matp]\label{eq:cyclic:systematic}
\end{align}
i.e., the data appears as cleartext in the code word. This is sometimes useful in practice. Encoders of the form \eqref{eq:cyclic:systematic} are called \emph{systematic}\index{systematic encoder}.
\subsection{Efficient Encoder for Cyclic Codes}
Since cyclic codes are linear, encoding can also be performed by multiplying a generator matrix with the data vector. More efficient are encoders that are based on multiplication of polynomials. Let $g(x)$ be a generator polynomial of an $(n,k)$ cyclic code. Any data vector $\vecu\in\myfield^k$ can be represented by a polynomial $u(x)\in\myfield[x]$ with $\deg u(x)\leq k-1$. The simplest encoder is the mapping
\begin{align}
u(x)\mapsto u(x)g(x).\label{eq:cyclic:encoder}
\end{align}
Suppose now that we want to have a systematic encoder. Since it simplifies the derivations, we place the data in the right part of the code word, i.e., we consider a mapping of the form
\begin{align}
u(x)\mapsto t(x)+x^{n-k}u(x)
\end{align}
where $\deg t(x)<n-k$. The polynomial $t(x)$ has to be chosen such that $t(x)+x^{n-k}u(x)$ is a code word, i.e., that it is a multiple of $g(x)$. Recall that $\deg g(x)=n-k$. We calculate
\begin{align}
&[t(x)+x^{n-k}u(x)]\mod g(x)=t(x)+x^{n-k}u(x)\mod g(x)\overset{!}{=}0\\
&\Rightarrow t(x)=-x^{n-k}u(x)\mod g(x).
\end{align}
The mapping thus becomes
\begin{align}
u(x)\mapsto -x^{n-k}u(x)\mod g(x)+x^{n-k}u(x).\label{eq:cyclic:encoderright}
\end{align}
To put the data in the left part of the code word, we recall that for a cyclic code, any shift of a code word is again a code word. Therefore, we cyclically shift the code word in \eqref{eq:cyclic:encoderright} to the right by $k$ positions. The resulting systematic encoder is
\begin{align}
u(x)\mapsto u(x)-x^k[x^{n-k}u(x)\mod g(x)].\label{eq:cyclic:encoderleft}
\end{align}
The polynomial multiplications in \eqref{eq:cyclic:encoder},\eqref{eq:cyclic:encoderright}, and \eqref{eq:cyclic:encoderleft} can be implemented very efficiently in hardware, which is one of the reasons why cyclic codes are widely used in practice.

\section{Syndromes}

Let's recall the definition of cosets from Section~\ref{subsec:cosets}. Let $\setc$ be a linear code in $\myfield^n$. The cosets of $\setc$ in $\myfield^n$ are the equivalence classes for the equivalence relation
\begin{align}
\vecv\sim\vecw\Leftrightarrow \vecv-\vecw\in\setc.\label{eq:cosettest}
\end{align}
The cosets are disjoint, each is of size $|\setc|$ and their union is $\myfield^n$. In particular, there are $|\myfield^n|/|\setc|$ different cosets. The syndrom indexes the cosets, i.e., if the syndrome of a vector $\vecv$ tells us to which coset $\vecv$ belongs. To paraphrase \eqref{eq:cosettest}, two vectors $\vecv$ and $\vecw$ belong to the same coset if their difference is a code word. Thus, to characterize syndromes for a specific class of codes, we should look for an appropriate test if a vector is a code word or not.

\subsection{Syndrome Polynomial}
\label{subsec:cyclic:syndrome}
Cyclic codes are defined by a generator polynomial $g(x)$ and the code word test is
\begin{align}
v(x)\mod g(x)\overset{?}{=}0.
\end{align}
This gives us imediately a test if two polynomials $v(x),w(x)$ belong to the same coset:
\begin{align}
&[v(x)-w(x)]\mod g(x)\overset{?}{=}0\\
\Leftrightarrow &v(x)\mod g(x)\overset{?}{=}w(x)\mod g(x).
\end{align}
Thus, for a cyclic code $\setc$ with generator polynomial $g(x)$, the polynomial $s(x)=v(x)\mod g(x)$ indexes the cosets of $\setc$ and $s(x)$ is therefore called the \emph{syndrome polynomial}\index{syndrome polynomial}.

\subsection{Check Matrix}

For general linear codes, syndromes are calculated by multiplication with a check matrix. We now show how check matrices for cyclic codes can be constructed. Let $\setc$ be an $(n,k)$ cyclic code over some field $\myfield$. Let $g(x)\in\myfield[x]$ be the generator polynomial of $\setc$. According to \eqref{eq:cyclic:hx}, the check polynomial $h(x)$ is defined by
\begin{align}
h(x)g(x)=x^n-1.
\end{align}
By Property~6 of cyclic codes, the degree of $g(x)$ is $n-k$. Therefore, the degree of $h(x)$ is $k$ and $h_0\neq 0$, i.e., $h(x)$ is of the form
\begin{align}
h(x)=h_0+h_1x+\dotsb+h_k x^k,\quad h_0,h_k\neq 0.\label{eq:hxform}
\end{align}
By Property 7 of cyclic codes, $h(x)$ defines a test for if a polynomial $c(x)$ is in the code $\setc$, i.e.,
\begin{align}
c(x)\mod (x^n-1)\in\setc\Leftrightarrow h(x)c(x)\mod (x^n-1)=0.\label{eq:polycheck}
\end{align}
We now want to use $h(x)$ to construct a check matrix for $\setc$. For clarity of exposure, we assume $\deg c(x)\leq n-1$, i.e., $c(x)\mod (x^n-1)=c(x)$. (If this is not the case, we can define $\tilde{c}(x):=c(x)\mod(x^n-1)$ and then use in the following derivation $\tilde{c}(x)$ instead of $c(x)$). We have
\begin{align}
 h(x)c(x)\mod (x^n-1)=0&\Leftrightarrow \sum_{i=0}^{n-1}c_ix^ih(x)\mod(x^n-1)=0\\
&\oiff{a} \sum_{i=0}^{n-1}c_ih^{(i)}(x)=0\\
&\Leftrightarrow \sum_{i=0}^{n-1}c_i \sum_{j=0}^{n-1}[h^{(i)}]_j x^j=0\\
&\Leftrightarrow \sum_{j=0}^{n-1} x^j  \sum_{i=0}^{n-1}c_i[h^{(i)}]_j=0\\
&\Leftrightarrow \forall j\in\{0,1,\dotsc,n-1\}\colon\sum_{i=0}^{n-1}c_i[h^{(i)}]_j=0\label{eq:sumcheck}
\end{align}
where (a) follows by Property~1 of cyclic codes. The scalar $[h^{(i)}]_j$ is the $j$th coefficient of the $i$th cyclic shift of $h(x)$, i.e.,  
\begin{align}
[h^{(i)}]_j=h_{(j-i)\mod n}.
\end{align}
We define
\begin{align}
\vech_j&:=\bigl([h^{(0)}]_j,[h^{(1)}]_j,\dotsc,[h^{(n-1)}]_j\bigr)\\
&=\bigl(h_{(j-0)\mod n},h_{(j-1)\mod n},\dotsc,h_{(j-(n-1))\mod n}\bigr).
\end{align}
The condition \eqref{eq:sumcheck} can now be written as
\begin{align}
c(x)\in\setc\Leftrightarrow \forall j=0,1,\dotsc,n-1\colon \vecc\vech_j^T=0.\label{eq:innercheck}
\end{align}
From the ``$\Rightarrow$'' direction of \eqref{eq:innercheck}, it follows that $\vech_j\in\setc^\perp$, for all $j=0,1,2,\dotsc,n-1$, see Definition~\ref{def:linear:dual}. If we can choose $n-k$ linearly independent vectors $\vech_j$, then by Proposition~\ref{prop:linear:dualdimension}, these vectors would form a basis of $\setc^\perp$. Using these vectors as rows of a matrix would form a generator matrix of $\setc^\perp$, which by Definition~\ref{def:linear:checkmatrix} is a check matrix of $\setc$.

It is convenient to choose $j=k,k+1,\dotsc,n-1$. The resulting matrix is then of the form
\begin{align}
\mathh=\bbm \vech_k\\\vech_{k+1}\\\vdots\\\vech_{n-1}\ebm
=\bbm
h_k&h_{k-1}&\dotsb&h_{0}&0&\dotsb&0\\
0&h_k&h_{k-1}&\dotsb&h_{0}&\ddots&\vdots\\
\vdots&\ddots&\ddots&\ddots&&\ddots&\\
0&\dotsb&0&h_k&h_{k-1}&\dotsb&h_{0}
\ebm.
\end{align}
By \eqref{eq:hxform}, $h_0\neq 0$ and $h_k\neq 0$, which implies that the rows of $\mathh$ are linearly independent. We conclude that $\mathh$ is a check matrix of $\setc$.

\clearpage

\section{Problems}

\myproblem Consider a $(15,11)$ binary cyclic code with generator polynomial $g(x) = 1 + x + x^4$:
\begin{enumerate}
\item Determine the check polynomial.
\item Let $u(x) = x + x^2 + x^3$. Encode $u(x)$ by each of the encoders \eqref{eq:cyclic:encoder},\eqref{eq:cyclic:encoderright}, and \eqref{eq:cyclic:encoderleft}.
\item For the code polynomial $v(x) = 1 + x + x^3 + x^4 + x^5 + x^9 + x^{10} + x^{11} + x^{13}$, determine
the data polynomial for each of the encoders \eqref{eq:cyclic:encoder},\eqref{eq:cyclic:encoderright}, and \eqref{eq:cyclic:encoderleft}.
\end{enumerate}

\myproblem
Let $g(x)=1+x^2$ be a generator polynomial of a block length $5$ cyclic code over $\myfield_2$.
\begin{enumerate}
\item What is the dimension of the code?
\item Is $v(x)=1+x+x^2+x^3+x^4$ a code word?
\item Systematically encode the bits $011$. 
\end{enumerate}

\myproblem
Let $\setc$ be a binary cyclic code with blocklength $n=4$ and dimension $k=2$.
\begin{enumerate}
\item Is $v(x)=x^2+x^3$ a code word polynomial?
\item Show that $x^2+1$ is the only generator polynomial that $\setc$ can have.
\item Is $w(x)=1+x^2+x^3$ a code word polynomial?
\item What is the minimum distance of the code?
\item A systematic encoder encodes $11\mapsto 11c_2c_3$. Calculate $c_2$ and $c_3$.
\end{enumerate}

\myproblem
Let $g(x)=1+x^3$ be the generator polynomial of a binary cyclic code $\setc$ of \text{block length $n=6$}.
\begin{enumerate}
\item What is the dimension of $\setc$?
\item What is the check polynomial of $\setc$?
\item Calculate a generator matrix of $\setc$.
\item Is $v(x)=x+x^2+x^3+x^4+x^5$ a code word?
\item List all cyclic subcodes of $\setc$ and calculate their dimension.
\end{enumerate}

\myproblem
\begin{enumerate}
\item Show that $g(x) = 1 + x + x^4 + x^5 + x^7 + x^8 + x^9$ generates a binary $(21,12)$ cyclic code.
\end{enumerate}

\myproblem For the $(15,11)$ binary Hamming code with generator polynomial $g(x) = 1 + x + x^4$:
\begin{enumerate}
\item Determine the check polynomial.
\item Determine the generator matrix $\boldsymbol{G}$ and the check matrix $\boldsymbol{H}$ for this code in non-systematic form.
\item Determine the generator matrix $\boldsymbol{G}$ and the check matrix $\boldsymbol{H}$ for this code in systematic form.
\end{enumerate}

\myproblem Let $g(x)$ be the generator polynomial of a binary cyclic code of length $n$.
\begin{enumerate}
\item Show that if $g(x)$ has $1+x$ as a factor, the code contains no codewords of odd weight.
\item Show that if $1+x$ is not a factor of $g(x)$, the code contains a codeword consisting of all ones.
\item Show that an $(n,k)$ binary cyclic code with $0<k<n$ has minimum weight at least three if $n$ is the smallest integer such that $g(x)$ divides $x^n-1$.
\end{enumerate}

\myproblem Let $v(x)$ be a code polynomial in a cyclic code of length $n$. Let $\ell$ be the smallest positive integer such that $v^{(\ell)}(x) = v(x)$. Show that $\ell$ is a factor of $n$.

\myproblem Let $\setc_1$ and $\setc_2$ be two cyclic codes of length $n$ that are generated by $g_1(x)$ and $g_2(x)$, respectively. Show that the codeword polynomials common to both $\setc_1$ and $\setc_2$ also form a cyclic code $\setc_3$. Determine the generator polynomial of $C_3$. If $d_1$ and $d_2$ are the minimum distances of $\setc_1$ and $\setc_2$, respectively, what can you say about the minimum distance of $\setc_3$?

\myproblem The polynomial $g(x) = 1 + x + x^4 + x^5 + x^7 + x^8 + x^9$ generates a binary $(21,12)$ cyclic code.
\begin{enumerate}
\item Let $r(x) = 1 + x^4 + x^{16}$ be a received polynomial. Compute the syndrome of $r(x)$.
\end{enumerate}

\myproblem Consider a blocklength $n=15$ binary Hamming code with generator polynomial $g(x) = 1 + x + x^4$. Codewords are transmitted over a BSC with crossover probability $0\leq \delta<0.5$. Your job is to implement an encoder and an ML-decoder.
\begin{enumerate}
\item What is the dimension $k$ of the code?
\item Write a Matlab function $\texttt{enc}$ that takes a binary string of length $k$ as argument and puts out a codeword $\vecc$ in vector form. \emph{Hint:} use the multiplication of polynomials in your implementation.
\item Implement a Matlab function $\texttt{bsc}$ that takes the codeword $\vecc$ and the crossover probability $\delta$ as argument and puts out a noisy version $\vecy$. \emph{Hint:} Implement your function by adding a random error pattern to $\vecc$ in $\fieldtwo$.
\item Form a lookup table with the most probable error pattern in each coset. Sort your table such that it can be indexed by the corresponding syndroms.
\item Implement the Matlab function $\texttt{mldec}$, see Problem 1.3, Exercise 5. Your function should return an estimate $\hat{\vecc}$ of the transmitted codeword.
\item Implement the Matlab function $\texttt{dec}$ that calculates from $\hat{\vecc}$ an estimate $\hat{m}$ of the original message $m$.
\item Estimate the end-to-end error probability $\Pr(m\neq\hat{m})$ of your code by Monte Carlo simulation for $\delta=0,0.1,0.2,0.3,0.4$.
\end{enumerate}

%% file: rs.tex

%
%
%
%
\chapter{Reed--Solomon Codes}
\label{chap:rs}
In this and the next chapter, we develop the most important algebraic codes, namely Reed--Solomon (RS) codes and Bose--Chaudhuri--Hocquenghem (BCH) codes. 

\section{Minimum Distance Perspective}

So far, the three parameters of interest were $R$ (the rate), $n$ (block length, delay, complexity), and $P_e$ (probability of error). The three parameters depend both on the code and the channel. We observed a trade-off between these three parameters. We now slightly change our perspective. We consider the parameters $(n,k,d)$ of a linear block code with minimum distance $d$. If a code with parameters $(n,k,d)$ is used on a channel, these three parameters can be related to $(R,n,P_e)$. We discussed this in Section~\ref{sec:linear:bound} and around \eqref{eq:linear:mldistance}. However, the two perspectives are not equivalent, e.g., fixing $(n,k)$ and searching for a code that maximizes $d$ leads in general to a code that is different from the code that results from fixing $(R,n)$ and minimizing $P_e$.
\subsection{Correcting $t$ Errors}
The ML decoder for $q$-ary symmetric channels was stated in \eqref{eq:linear:mldistance} as
\begin{align*}
\dml(\vecy)=\argmin_{\vecc\in\setc}\hw(\vecy-\vecc).
\end{align*}
This is a \emph{minimum distance decoder}\index{minimum distance decoder}, since it chooses the codeword that is closest to the observed channel output vector in terms of Hamming distance. The very same decoder can be used on any $q$-ary channel. It decodes correctly as long as the error pattern $\vecz$ is such that the channel output $\vecy$ remains close enough to the codeword that was actually transmitted. 
\begin{theorem}\label{theo:rs:mdd}
Let $\vecc$ be the codeword of a linear code $\setc\subseteq\myfield_q^n$ with minimum distance $d$ and let $\vecy=\vecc+\vecz$ be the output of a $q$-ary channel where $\vecz\in\myfield_q^n$ is the error pattern. Define
\begin{align}
\hat{\vecc}=\vecy-\argmin_{\vecz\in\{\vecy+\setc\}}\hw(\vecz).
\end{align}
Then 
\begin{align}
\hat{\vecc}=\vecc\text{ if }\hw(\vecz)\leq t:=\left\lfloor\frac{d-1}{2}\right\rfloor
\end{align}
that is, the minimum distance decoder is guaranteed to decode correctly if the transmitted codeword gets corrupted in at most $t$ coordinates.
\end{theorem}
\begin{proof}
See Problem~\ref{prob:rs:mdd}
\end{proof}
The larger the minimum distance $d$, the greater is the number $t$ of errors that we can guarantee to correct. We next relate $d$ to block length $n$ and code dimension $k$.

\subsection{Singleton Bound and MDS Codes}
\label{subsec:singleton}
Let $\setx$ be the input alphabet of some channel and let $\setc\subseteq \setx^n$ be a (not necessarily linear) block code. Let $d$ be the minimum distance of the code. This means that any two codewords differ in at least $d$ positions. After erasing the values in any $d-1$ positions, the two codewords still differ in at least one of the remaining $n-(d-1)$ positions. Therefore, there can be at most $|\setx|^{n-d+1}$ codewords in $\setc$. 
\begin{definition}
Let $\setc$ be a (not necessarily linear) code with alphabet $\setx$ and $|\setx|^k$ codewords. A set of $k$ coordinates where the codewords run through all $|\setx|^k$ possible $k$-tuples is called an \emph{information set}\index{information set}.
\end{definition}

\begin{theorem}[Singleton Bound]\label{rs:theo:singleton}
A (not necessarily linear) code $\setc$ with minimum Hamming distance $d$ over an alphabet $\setx$ can have at most 
\begin{align}
|\setc|\leq|\setx|^{n-d+1}
\end{align}
codewords. This bound is called the \emph{Singleton bound}\index{Singleton bound}. Equality holds if and only if any set of $k=n-d+1$ coordinates is an information set. A code that meets the Singleton bound with equality is called a \emph{maximum distance separable}\index{maximum distance separable code} (MDS)\index{MDS code \see{maximum distance separable code}} code.
\end{theorem} 
The only binary MDS codes are the trivial $(n,n,1)$ code, the $(n,n-1,2)$ single parity check code and the $(n,1,n)$ repetition code, see Problem~\ref{prob:rs:binarymds}. This motivates us to look at non-binary codes.

\section{Finite Fields}\label{sec:finite}

\begin{theorem}
Let $q$ be a positive integer. There exists a finite field with $q$ elements if and only if $q=p^m$ for some prime number $p$ and a positive integer $m$. All finite fields with $q$ elements are isomorphic to each other.
\end{theorem}

\subsection{Prime Fields $\field{F}_p$}

\begin{theorem}\label{theo:primefield}
For every prime number $p$, the integers $\setz$ with $\mod p$ addition and multiplication form a field $\field{F}_p$ with $p$ elements. Any field $\field{F}$ with $p$ elements is isomorphic to $\field{F}_p$ via the correspondence
\begin{align}\underbrace{1+1+\dotsb+1}_{i\text{ times}}\in\field{F}\leftrightarrow i\in\field{F}_p.
\end{align}
\end{theorem}

\subsection{Construction of Fields $\field{F}_{p^m}$}\label{sec:rs:fieldspm}

\begin{definition}
Let $f(x)$ be a polynomial of degree $m$ over the field $\field{F}$.
\begin{itemize}
\item $f(x)$ is \emph{monic}\index{monic}, if the coefficient of $x^m$ is equal to one, i.e., if $f_m=1$.
\item $f(x)$ is \emph{irreducible}\index{irreducible}, if it is not the product of two factors of positive degree in $\field{F}[x]$.
\item $f(x)$ is a \emph{prime polynomial}\index{prime polynomial}, if it is monic and irreducible.
\end{itemize}
\end{definition}
\begin{example}
Consider the polynomials over $\field{F}_3$. The polynomial $1+x^2$ is monic but the polynomial $1+2x^2$ is not monic. The polynomial $2+x^2$ is reducible since
\begin{align}
2+x^2=(x-1)(x-2).
\end{align}
The polynomial $1+x^2$ is irreducible, since if not, it would have a factor of degree $1$ and thus a root in $\field{F}_3$. However:
\begin{align}
1+x^2\big|_{x=0}&=1\neq 0\\
1+x^2\big|_{x=1}&=2\neq 0\\
1+x^2\big|_{x=2}&=2\neq 0.
\end{align}
\end{example}
There exist a number of methods to test if a polynomial is irreducible or not, see the literature on abstract algebra. The next theorem states how to construct finite fields of order $p^m$ in analogy to Theorem~\ref{theo:primefield}. The set of polynomials $\field{F}_p[x]$ takes the role of the integers $\setz$ and a prime polynomial of degree $m$ takes the role of the prime number $p$.
\begin{theorem}\label{theo:extensionfield}
Let $g(x)$ be a prime polynomial of degree $m$ over a prime field $\field{F}_p$. Then the polynomials
\begin{align}
\field{F}_p[x]\mod g(x)\label{def:field:polynomial}
\end{align}
form a field with $p^m$ elements. Any field $\field{F}$ with $p^m$ elements is isomorphic to \eqref{def:field:polynomial}. 
\end{theorem}
\begin{example} Construction of $\field{F}_{2^2}$. We first need an irreducible polynomial over $\field{F}_2$ of degree $2$. The polynomial $1+x^2$ is reducible, since it has $1$ as a root and thus $1+x$ as a factor. The polynomial $g(x)=1+x+x^2$ is irreducible, since if not, it would have a factor of degree 1 and thus a root in $\field{F}_2$. However $1+0+0^2=1\neq 0$ and $1+1+1^2=1\neq 0$. The addition table is
\begin{center}
\begin{tabular}{c|cccc}
$+$&$0$&$1$&$x$&$x+1$\\\hline
$0$&$0$&$1$&$x$&$x+1$\\
$1$&$1$&$0$&$x+1$&$x$\\
$x$&$x$&$x+1$&$0$&$1$\\
$x+1$&$x+1$&$x$&$1$&$0$
\end{tabular}
\end{center}
Note that $(\{0,1\},+)$ forms a subgroup of $(\field{F}_{2^2},+)$. The multiplication table can be obtained by performing $\mod g(x)$ multiplication of the field elements.
\begin{center}
\begin{tabular}{c|cccc}
$\cdot$&$0$&$1$&$x$&$x+1$\\\hline
$0$&$0$&$0$&$0$&$0$\\
$1$&$0$&$1$&$x$&$x+1$\\
$x$&$0$&$x$&$x+1$&$1$\\
$x+1$&$0$&$x+1$&$1$&$x$
\end{tabular}
\end{center}
Note that $\field{F}_2$ forms a subfield of $\field{F}_{2^2}$.
\end{example} 
\subsubsection{Primitive Element}
By Theorem~\ref{theo:extensionfield}, we can construct a finite fields with $p^m$ elements, given that we know an irreducible polynomial $g(x)$ in $\field{F}_p[x]$ of order $m$. To establish the multiplication table, we need to perform $\mod\,g(x)$ multiplication of polynomials. The following theorem makes the construction even more convenient.
\begin{theorem}
For any field $\field{F}_{p^m}$, the multiplicative group $\field{F}_{p^m}\setminus 0$ is \emph{cyclic}, i.e., there exists a \emph{primitive element}\index{primitive element} $\alpha\in\field{F}_{p^m}\setminus 0$ such that every element in $\field{F}_{p^m}\setminus 0$ can be written as a power of $\alpha$, i.e.,
\begin{align}
\field{F}_{p^m}=\{0,1,\alpha^1,\dotsc,\alpha^{p^m-2}\}\label{def:field:power}
\end{align}
and
\begin{align}
\alpha^i\cdot \alpha^j = \alpha^{(i+j)\mod (p^m-1)}.
\end{align}
\end{theorem}
By this theorem, the construction of a multiplication table is trivial. However, how does the corresponding addition table look like? The \emph{primitive polynomial} provides this connection.
\begin{definition}
Let $\beta$ be an element of $\field{F}_{p^m}$. The \emph{minimal polynomial}\index{minimal polynomial} of $\beta$ in $\field{F}_p[x]$ is the monic polynomial in $\field{F}_p[x]$ of lowest degree that has $\beta$ as a root. The minimal polynomial in $\field{F}_p[x]$ of a primitive element in $\field{F}_{p^m}$ is called a \emph{primitive polynomial}\index{primitive polynomial}.
\end{definition}
We will study minimal polynomials in more detail in Subsection~\ref{sec:bch:minimal}. By definition, a primitive polynomial is a prime polynomial.
\begin{theorem}
Let $g(x)$ be a primitive polynomial with the corresponding primitive element $\alpha\in\field{F}_{p^m}$. Then $g(x)$ has degree $m$ and in particular, $\field{F}_{p^m}$ is isomorphic to $\field{F}_p[x]\mod g(x)$. 
\end{theorem}
The following theorem gives the correspondence between the polynomial representation and the cyclic representation of a finite field.
\begin{theorem}
Let $\alpha$ be a primitive element of $\field{F}_{p^m}$ with the primitive polynomial $g(x)\in\field{F}_p[x]$. Then
\begin{align}
0&\leftrightarrow f(x)=0\\
\alpha^i&\leftrightarrow x^i\mod g(x),\quad i=0,1,\dotsc,p^m-2.
\end{align}
defines an isomorphism between \eqref{def:field:polynomial} and \eqref{def:field:power}.
\end{theorem}
\begin{example}
A primitive polynomial for $\field{F}_{2^3}$ is $g(x)=1+x+x^3$. The correspondence table is
\begin{align*}
0&\leftrightarrow 0\\
1&\leftrightarrow \alpha^0\\
x&\leftrightarrow \alpha^1\\
x^2&\leftrightarrow\alpha^2\\
1+x&\leftrightarrow\alpha^3\\
x+x^2&\leftrightarrow\alpha^4\\
1+x+x^2&\leftrightarrow\alpha^5\\
1+x^2&\leftrightarrow\alpha^6
\end{align*}
\end{example}
\begin{theorem}\label{theo:noname}
Over any field $\field{F}$, a monic polynomial $f(x)\in\field{F}[x]$ of degree $m$ can have no more than $m$ pairwise distinct roots in $\field{F}$. If it does have $m$ pairwise distinct roots $\beta_1,\dotsc,\beta_m$, then the unique factorization (up to permutations of the factors) of $f(x)$ is $f(x)=(x-\beta_1)\dotsb(x-\beta_m)$.
\end{theorem}

\section{Reed--Solomon Codes}

Consider a field $\field{F}_q=\{\beta_1,\beta_2,\dotsc,\beta_q\}$ with $q$ elements. A \emph{Reed--Solomon}\index{Reed--Solomon Code}\index{RS code \see{Reed--Solomon Code}} (RS) code $\setc_\text{RS}$ over $\field{F}_q$ with block length $n=q$ and dimension $k$ is defined as the image of a mapping $\enc \colon \field{F}_q^k\to\field{F}_q^q$. We represent the $k$-tuples $\vecu\in\field{F}_q^k$ by their generating function, i.e., 
\begin{align}
\vecu=(u_0,u_1,\dotsc,u_{k-1})\leftrightarrow u(x)=u_0+u_1x+\dotsb+u_{k-1}x^{k-1}.
\end{align}
The \emph{evaluation map}\index{evaluation map}\index{ev@\enc \see{evaluation map}} \enc is given by
\begin{align}
\enc \colon \field{F}_q^k&\to\field{F}_q^q\\
\vecu &\mapsto \enc(\vecu)=\bigl (u(\beta_1),u(\beta_2),\dotsc,u(\beta_q)\bigr).
\end{align}
The RS code is defined as the image of \enc, i.e.,
\begin{align}
\setc_\text{RS}=\enc (\field{F}_q^k).
\end{align}
\begin{theorem}\ \label{theo:rs:rs}
\begin{enumerate}
\item The RS code is linear.
\item The dimension of the RS code is equal to $k$.
\item The RS code is MDS.
\end{enumerate}
\end{theorem}
\begin{proof}
1. We need to show that \enc is linear. The map \enc is linear if and only if
\begin{align}
\enc (\vecu)+\beta\cdot\enc(\vecv)=\enc(\vecu+\beta\vecv),\quad \forall\vecu,\vecv\in\field{F}_q^k\text{ and }\forall \beta\in\field{F}_q.
\end{align}
The condition holds if it holds for each coordinate. For the $i$th coordinate, we have
\begin{align}
[\enc (\vecu)+\beta\cdot\enc(\vecv)]_i&=\enc (\vecu)_i+\beta\cdot\enc(\vecv)_i\\
&=u(\beta_i)+\beta\cdot v(\beta_i)\\
&=\sum_{j=0}^{k-1}u_j\beta_i^j+\beta\cdot\sum_{j=0}^{k-1}v_j \beta_i^j\\
&=\sum_{j=0}^{k-1} (u_j+\beta v_j)\beta_i^j\\
&=\enc(\vecu+\beta \vecv)_i
\end{align}
and we conclude that $\setc_\text{RS}$ is indeed linear.

2. We show that the image of the mapping $\enc$ is $k$ dimensional. $\enc$ is defined on $\field{F}_q^k$. Thus
\begin{align}
\dim(\ker(\enc))+\dim(\im(\enc))=\dim(\field{F}_q^k)=k
\end{align}
where $\ker(\enc)$ is the \emph{kernel}\index{kernel}\index{ker \see{kernel}} of $\enc$ and $\im(\enc)$\index{image}\index{im \see{image}} the \emph{image} of $\enc$. Let $\vecu\leftrightarrow u(x)$ be a non-zero $k$-tuple. Then $u(x)$ is a polynomial of degree at most $k-1$. By Theorem~\ref{theo:noname}, $u(x)$ can have at most $k-1$ distinct roots in $\field{F}_q^n$. Therefore, at most $k-1$ entries of the corresponding codeword
\begin{align}
\enc(\vecu)=\bigl(u(\beta_1),u(\beta_2),\dotsc,u(\beta_n)\bigr)
\end{align}
can be equal to zero and we conclude that the weight of each non-zero codeword is at least $n-(k-1)$. Thus, the dimension of the kernel of $\enc$ is equal to zero and the image has dimension
\begin{align}
\dim(\im(\enc))=\dim(\field{F}_q^k)-\dim(\ker(\enc))=k-0=k.
\end{align}

3. Since the RS code is linear, the minimum distance of the code is equal to the minimum weight of all non-zero codewords. As we have shown in 2., the weight of each non-zero codeword is at least $n-(k-1)$. Consequently, the minimum distance $d$ is bounded from below by
\begin{align}
d\geq n-k+1.
\end{align}
By the Singleton bound, the minimum distance of any linear code is bounded from \emph{above} by $d\leq n-k+1$. Therefore, $d=n-k+1$ must be true, which shows that the RS code is MDS.
\end{proof}

\subsection{Puncturing RS Codes}

We have defined RS codes over $\myfield_q$ for block length $n=q$. By \emph{puncturing}\index{puncturing} the set $\myfield_q$ $\ell$ times, i.e., by removing $\ell$ elements from $\myfield_q$, and then defining an evaluation map based on this new set, we get a punctured RS code. For example, if we remove $\beta_1$ and $\beta_2$ from $\myfield_q$, the evaluation map becomes
\begin{align}
&\enc_{\beta_1\beta2}\colon\myfield_q^k\to\myfield_q^{q-2}\\
&\vecu\mapsto\enc_{\beta_1\beta2}(\vecu)=\bigl(u(\beta_3),u(\beta_4),\dotsc,u(\beta_q)\bigr).
\end{align}
This defines an $(n=q-2,k)$ code. By the same arguments as in the proof of Theorem~\ref{theo:rs:rs}, punctured RS codes are MDS, i.e., the minimum distance is $d=n-k+1$.

\subsection{RS Codes via Fourier Transform}

The transform \enc provides more structure when the code is punctured in $\beta=0$ and the block length $n$ is equal to $q-1$. We denote the corresponding evaluation map by $\enc_0$. We first study the multiplicative group of finite fields in more detail and then define a Fourier transform for finite fields based on our findings. We then interpret the evaluation map $\enc_0$ as a Fourier transform in $\field{F}_q^n$.

\subsubsection{The Fourier Transform in $\field{F}_q$}

\begin{definition}
Let $\omega$ be an element in $\field{F}_q$. The \emph{order}\index{order} of $\omega$ is defined as
\begin{align}
|\omega|:=\min\{i>0\colon \omega^i=1\}.
\end{align}
\end{definition}
Note that since $0=(q-1)\mod(q-1)$, $\omega^{q-1}=1$ and $|\omega|\leq q-1$ for all elements $\omega$ in $\field{F}_q\setminus 0$.
\begin{lemma}\label{lem:nthroot}
Let $\omega\neq 1$ be an element in $\field{F}$ with $\omega^n=1$. Then
\begin{align}
\sum_{i=0}^{n-1}\omega^i=0.
\end{align}
\end{lemma}
\begin{proof}
Since by assumption $\omega^n=1=\omega^0$,
\begin{align}
\omega\sum_{i=0}^{n-1}\omega^i=\sum_{i=1}^n\omega^i=\sum_{i=0}^{n-1}\omega^i.
\end{align}
Since by assumption $\omega\neq 1$, this can only be true if $\sum_{i=0}^{n-1}\omega^i=0$.
\end{proof}
\begin{lemma}\label{lem:fourier2}
Let $\omega$ be an element of order $n$ in $\field{F}_q$ and let $j$ be an integer. Then
\begin{align}
\sum_{i=0}^{n-1}(\omega^j)^i=\begin{cases}
n,&\text{if }j\mod n=0\\
0,&\text{otherwise}
\end{cases}
\end{align}
where $n\in\field{F}_q:=\underbrace{1+1+\dotsb+1}_{n\text{ times}}$.
\end{lemma}
\begin{proof}
$j\mod n=0$: For some integer $k$, we have $j=nk$. We calculate
\begin{align} 
\omega^j=\omega^{kn}=(\omega^n)^k\oeq{a}1^k=1
\end{align}
where (a) follows because the order of $\omega$ is $n$.

$j\mod n\neq 0$: For some integers $k$ and $0<\ell<n$, we have $j=kn+\ell$. We calculate
\begin{align}
\omega^j=\omega^{kn+\ell}=(\omega^n)^k\omega^\ell=1^k\omega^\ell=\omega^\ell\overset{\text{(a)}}{\neq} 1\label{eq:proof:neqzero}
\end{align}
where (a) follows because the order of $\omega$ is $n>\ell$. Combining \eqref{eq:proof:neqzero} and Lemma~\ref{lem:nthroot} yields the statement $\sum_{i=0}^{n-1}(\omega^j)^i=0$.
\end{proof}

\begin{definition}
Let $\omega$ be an element in $\field{F}_q$. Let 
\begin{align}
\vecv=(v_o,v_1,\dotsc,v_{|\omega|-1})\leftrightarrow v(x)
\end{align}
be a vector in $\field{F}_q^{|\omega|}$. The \emph{Fourier transform}\index{Fourier transform} of $\vecv$ is defined as
\begin{align}
\four_{\omega}\colon \field{F}_q^{|\omega|}&\to \field{F}_q^{|\omega|}\nonumber\\
\vecv&\mapsto\four_\omega(\vecv)=\left(v(\omega^0),v(\omega^1),\dotsc,v(\omega^{|\omega|-1})\right)=:\vecc\label{eq:fourier}\\
\four_{\omega}^{-1}\colon \field{F}_q^{|\omega|}&\to \field{F}_q^{|\omega|}\nonumber\\
\vecc&\mapsto\four_\omega^{-1}(\vecc)=\frac{1}{|\omega|}\left(c(\omega^0),c(\omega^{-1}),\dotsc,c(\omega^{-(|\omega|-1)})\right)\nonumber\\
&\hspace{2cm}=\frac{1}{|\omega|}\left(c(\omega^0),c(\omega^{|w|-1}),\dotsc,c(\omega^{1})\right).\label{eq:fourierinverse}
\end{align}
\end{definition}
We verify that indeed $\four_\omega^{-1}(\four_\omega(\vecv))=\vecv$. Define $n:=|\omega|$. Consider some coordinate $i$, $0\leq i\leq n-1$.
\begin{align}
\four_\omega^{-1}(\four_\omega(\vecv))_i=n^{-1}c(\omega^{-i})&=n^{-1}\sum_{j=0}^{n-1}v(\omega^j)\omega^{-ij}\nonumber\\
&=n^{-1}\sum_{j=0}^{n-1}\sum_{k=0}^{n-1}v_k\omega^{jk}\omega^{-ij}\nonumber\\
&=n^{-1}\sum_{k=0}^{n-1}v_k\sum_{j=0}^{n-1}(\omega^{k-i})^j.\label{eq:fouriercheck1}
\end{align}
For $0\leq k,i\leq n-1$, we have
\begin{align}
\begin{cases}
k-i\mod n=0,&\text{if }k=i\\
k-i\mod n\neq 0,&\text{if }k\neq i.
\end{cases}
\end{align}
Therefore, by Lemma~\ref{lem:fourier2}, 
\begin{align}
\sum_{j=0}^{n-1}(\omega^{k-i})^j=\begin{cases}
n,&\text{if } k=i\\
0,&\text{otherwise}.
\end{cases}
\label{eq:fouriercheck2}
\end{align}
Using \eqref{eq:fouriercheck2} in \eqref{eq:fouriercheck1} finally yields
\begin{align}
n^{-1}c(\omega^{-i})=n^{-1}v_in=v_i.
\end{align}
This is true for each coordinate $i$, $0\leq i\leq n-1$, which shows that \eqref{eq:fourierinverse} indeed defines the inverse of the Fourier transform \eqref{eq:fourier}.
\subsubsection{RS Codes via Fourier Transform}
Consider an $(n=q-1,k)$ RS code. Let $\alpha$ be a primitive element in $\field{F}_q$. The order of a primitive element is $|\alpha|=q-1=n$. The transform $\enc_0$ can be interpreted as a Fourier transform $\four_\alpha$ in $\field{F}_q^n$ by
\begin{align}
\vecu&\leftrightarrow (u_0,u_1,\dotsc,u_{k-1},\underbrace{0,\dotsc,0}_{n-k\text{ times}})\\
&\leftrightarrow u(x)\mapsto (u(\alpha^0),u(\alpha^1),\dotsc,u(\alpha^{n-1}))=:\vecc.
\end{align}
Applying the inverse Fourier transform to $\vecc\leftrightarrow c(x)$, we have
\begin{align}
c(\alpha^{-k})=c(\alpha^{-(k+1)})=\dotsb=c(\alpha^{-(n-1)})=0.
\end{align}
Since $-k\mod n=n-k$, we equivalently have
\begin{align}
c(\alpha^{1})=c(\alpha^{2})=\dotsb=c(\alpha^{n-k})=0.\label{eq:rs:preroots}
\end{align}
Conversely, if a vector $\vecc\in\field{F}_q^n$ has the property \eqref{eq:rs:preroots}, then it is a codeword. Thus, we have the following alternative characterization of RS codes.
\begin{theorem}\label{theo:rs:roots}
Let $\setc_\text{RS}$ be an $(n=q-1,k)$ RS code. Let $\vecc$ be a vector in $\field{F}_q^n$. Then
\begin{align}
\vecc\leftrightarrow c(x)\in\setc_\text{RS}\Leftrightarrow c(\alpha^{1})=c(\alpha^{2})=\dotsb=c(\alpha^{n-k})=0.
\end{align}
\end{theorem}
Based on this theorem, we will in the remaining sections of this chapter further exploit the structure of RS codes.

\subsection{Syndromes}

Recall from Section~\ref{subsec:cosets} that two vectors $\vecv,\vecw$ belong to the same coset of a code $\setc$ if $\vecv-\vecw\in\setc$. The coset to which a vector $\vecv$ belongs is identified by the syndrome of $\vecv$. We now want to calculate the syndrome of a vector $\vecv$ in the case when an RS code is considered. RS codes are cyclic codes, so we could calculate syndromes as in Section~\ref{subsec:cyclic:syndrome}. However, we have an alternative code word test at hand. Consider an $(n=q-1,k)$ RS code $\setc_\text{RS}$ over $\myfield_q$ with primitive element $\alpha$. A polynomial $v(x)$ is a code word if
\begin{align}
&v(\alpha)=v(\alpha^2)=\dotsb=v(\alpha^{n-k})=0
\end{align}
Thus, a code word test for RS codes is
\begin{align}
[v(\alpha),v(\alpha^2),\dotsc,v(\alpha^{n-k})]\overset{?}{=}\veczero.
\end{align}
Consequently, two polynomials $v(x)$ and $w(x)$ belong to the same coset of $\setc_\text{RS}$ if the polynomial $v(x)-w(x)$ passes the test, i.e., if
\begin{align}
[v(\alpha),v(\alpha^2),\dotsc,v(\alpha^{n-k})]=[w(\alpha),w(\alpha^2),\dotsc,w(\alpha^{n-k})].
\end{align}
Thus
\begin{align}
\vecs=[v(\alpha),v(\alpha^2),\dotsc,v(\alpha^{n-k})]
\end{align}
is the syndrome of $\vecv\leftrightarrow v(x)$.

\subsection{Check Matrix for RS Codes}

Let again $\setc_\text{RS}$ be an $(n=q-1,k)$ RS code. Let $\vecc$ be a codeword. By Theorem~\ref{theo:rs:roots},
\begin{align}
c(\alpha^i) = c_0\alpha^{0i}+c_1\alpha^{1i}+\dotsb+c_{n-1}\alpha^{(n-1)i}=0,\quad i=1,2,\dotsb,n-k.\label{eq:rs:dual}
\end{align}
We define the vector
\begin{align}
\vecv_i = (\alpha^{0i},\alpha^{1i},\dotsc,\alpha^{(n-1)i}).
\end{align}
By \eqref{eq:rs:dual}, $\vecc\vecv_i^T=0$, so the vector $\vecv_i$ is in the dual code $\setc_\text{RS}^\perp$. We define the matrix $\mathh$ as
\begin{align}
\mathh=\bbm 
\vecv_1\\
\vecv_2\\
\vdots\\
\vecv_{n-k}
\ebm=\bbm
\alpha^{0}&\alpha^{1}&\dotsb&\alpha^{(n-1)}\\
\alpha^{0}&\alpha^{2}&\dotsb&\alpha^{2(n-1)}\\
\vdots&\vdots&\dotsb&\vdots\\
\alpha^{0}&\alpha^{n-k}&\dotsb&\alpha^{(n-k)(n-1)}
\ebm.\label{eq:rs:checkmatrix}
\end{align}
The equalities \eqref{eq:rs:dual} become in vector notation
\begin{align}
\vecc\mathh^T=(\vecc\vecv_1^T,
\vecc\vecv_2^T,\dotsc,
\vecc\vecv_{n-k}^T)
=(0,0,\dotsc,0).
\end{align}
The matrix $\mathh$ defines a linear mapping from $\field{F}_q^n$ to $\field{F}_q^{n-k}$. By Theorem~\ref{theo:rs:roots}, the kernel of $\mathh^T$ is the RS code $\setc_\text{RS}$, which is of dimension $k$. Therefore, the dimension of the image of $\mathh$ is
\begin{align}
\dim(\im(\mathh))=\dim(\field{F}_q^n)-\dim(\ker(\mathh))=n-k.
\end{align}
Therefore, the rows of the matrix $\mathh$ are $n-k$ linearly independent vectors in $\setc_\text{RS}^\perp$. Since by Proposition~\ref{prop:linear:dualdimension}, the dual code $\setc_\text{RS}^\perp$ is of dimension $n-k$, $\mathh$ is a generator matrix of $\setc_\text{RS}^\perp$ and therefore a check matrix of $\setc_\text{RS}$. We summarize our findings.
\begin{theorem}\label{theo:rs:syndrome}
Let $\alpha$ be a primitive element of $\field{F}_q$. Then \eqref{eq:rs:checkmatrix} is a check matrix of an $(n=q-1,k)$ RS code. Let $\vecy$ be a vector in $\field{F}_q^n$. The syndrome $\vecs$ of $\vecy\leftrightarrow y(x)$ can be calculated by
\begin{align}
\vecs=\vecy\mathh^T=[y(\alpha),y(\alpha^2),\dotsc,y(\alpha^{n-k})].
\end{align}
\end{theorem}
\subsection{RS Codes as Cyclic Codes}

By Theorem~\ref{theo:rs:roots}, a codeword $c(x)$ has the polynomial
\begin{align}
g(x)=\prod_{j=1}^{n-k}(x-\alpha^j)
\end{align}
as a factor. Since $\alpha^j$ is a root of $x^n-1$ for any integer $j$, $g(x)$ is also a factor of $x^n-1$, i.e, $g(x)$ divides $x^n-1$. By Properties 4 \& 5 of cyclic codes, the RS code is a cyclic code with generator polynomial $g(x)$.

\clearpage

\section{Problems}

\myproblem\label{prob:rs:mdd} Prove Theorem~\ref{theo:rs:mdd}.

\myproblem\label{prob:rs:binarymds} Show that the binary $(n,n-1)$ single parity check code, the binary $(n,1)$ repetition code, and the binary $(n,n)$ code are MDS. Are there any other binary MDS codes?

\myproblem
Let $\setc_{\text{RS}}$ be an $(n=2^2-1,1)$ RS code. $p(x)=1+x+x^2$ is a primitive polynomial. Let $\alpha$ be a root of $p(x)$, i.e., a primitive element of $\myfield_4$. The correspondence table is
\begin{align*}
0&\leftrightarrow 0\\
z^0&\leftrightarrow \alpha^0\\
z^1&\leftrightarrow \alpha^1\\
1+z&\leftrightarrow \alpha^2.
\end{align*}
\begin{enumerate}
\item Encode the information $\alpha^2$ by using the Fourier transform.
\item Calculate the generator polynomial $g(x)$ of $\setc_\text{RS}$.
\item Encode $\alpha^2$ by using $g(x)$.
\item The code is used on a binary channel. Calculate the rate in bits per channel use.
\item The binary representation of a codeword of $\setc_\text{RS}$ is transmitted over the binary channel. Consider the two binary error patterns $011000$ and $001100$. Which of these two error patterns can a minimum distance decoder surely correct?
\end{enumerate}

\myproblem
Let $\setc_\text{RS}$ be an $(n=5-1,2)$ RS code over $\myfield_5$. \emph{Note:} $\myfield_5$ is isomorphic to the integers modulo $5$.
\begin{enumerate}
\item Show that $2$ is a primitive element of $\myfield_5$.
\item Encode $\vecu=(1,2)$ by using the Fourier transform $\mathcal{F}_2$.
\item Calculate the generator polynomial of the code.
\end{enumerate}
The code is used on a $5$-ary channel.
\begin{enumerate}
\addtocounter{enumi}{3}
\item What is the rate of the code in bits per channel use?
\item Calculate the syndrome of $\vecy=(1,2,3,4)$. Is $\vecy$ a code word?
\end{enumerate}

\myproblem
Let $\setc_\text{RS}$ be an $(n=2^2-1,2)$ RS code over $\myfield_{2^2}$ with primitive element $\alpha$. 
The correspondence table is

\begin{align*}
0&\leftrightarrow 0\\
z^0&\leftrightarrow \alpha^0\\
z^1&\leftrightarrow \alpha^1\\
1+z&\leftrightarrow \alpha^2.
\end{align*}

\begin{enumerate}
\item What is the primitive polynomial $p(z)$ of $\alpha$?
\item Encode $(1,\alpha^2)$ using the evaluation map.
\item Compute a generator matrix $\matg$ for $\setc_\text{RS}$.
\item Does a systematic generator matrix for Reed-Solomon codes always exist? 
\end{enumerate}
Let the data symbols $\vecu=(u_1,u_2)$ be encoded using the evaluation map.
The code is used on an erasure channel. 
Suppose we receive $\vecy=(1,e,\alpha)$ where $e$ denotes the erasure symbol.
\begin{enumerate}
\addtocounter{enumi}{4}
\item Calculate the data symbols $u_1$ and $u_2$.
\end{enumerate}

\myproblem Calculate the Reed--Solomon Codes over $\myfield_2$ for $n=2$ and $k=1$ and $k=2$. Verify that the code is MDS, i.e., that the minimum distance $d$ is given by $d=n-k+1$.

\myproblem Calculate the Reed--Solomon Code over $\myfield_7$ for $n=7$ and $k=4$. 
\begin{enumerate}
\item Use Matlab to list all codewords.
\item What is the minimum distance of the code?
\item Is the code MDS?
\item How many errors can we guarantee to correct?
\item Suppose the code is used on an erasure channel. How many erasures can we guarantee to correct?
\end{enumerate}

\myproblem Consider the finite extension field $\myfield_{2^3}$ with primitive polynomial $p(x)=1+x+x^3\in\myfield_2[x]$.
\begin{enumerate}
 \item Verify that the polynomial $p(x)=1+x+x^3$ is irreducible over $\myfield_2$.
 \end{enumerate}
 Let $\alpha$ be a corresponding primitive element, i.e., $\alpha\colon p(\alpha)=0$.
\begin{enumerate}
 \addtocounter{enumi}{1}
 \item List all elements of $\myfield_{2^3}$ constructed by $p(x)$ and represent each element in terms of the primitive polynomial and the primitive element.
 \item Setup the addition and multiplication table for $\myfield_{2^3}$.
 \item Compute $\alpha^4+1$ and $\alpha^4\cdot\alpha^2$. Which representation is better for addition and which for multiplication?
 \item Write a Matlab script which outputs a list of all elements of $\myfield_{2^3}$ in terms of $p(x)$ and $\alpha$. Output the binary representation for each field elements.  
\end{enumerate}

\myproblem Consider Reed--Solomon code $\setc$ over $\myfield_{2^3}$. Use $p(z)=1+z^2+z^3$ as primitive polynomial for $\myfield_{2^3}$.
\begin{enumerate}
 \item What is the maximum possible length $n$ of a Reed--Solomon code over this field $\myfield_{2^3}$? Use this maximum length in the following.
 \item For $k=3$, determine the minimum distance $\dmin$.
 \item Construct a generator matrix $\matg$ for this code.
 \item Encode the binary information vectors $\vecu=(u_0,u_1,u_2)$, with $u_0=(010),u_1=(001),u_2=(011)$. These are binary representations of the elements of $\myfield_{2^3}$, e.g., 
\begin{align*}
u_1=010\leftrightarrow 0+1\cdot z+0\cdot z^2.
\end{align*}
 \item Write a Matlab function which takes a binary vector of length $m\cdot k=3\cdot 3$ as input and returns the corresponding Reed--Solomon codeword of length $n$ over $\myfield_{2^3}$.
 \item Implement a function in Matlab which maps a vector of length $n$ over $\myfield_{2^3}$ to a binary vector of length $m\cdot n=3\cdot n$. Use for the binary representation the correspondence defined by the primitive polynomial.  
\end{enumerate}

%% file: bch.tex
\chapter{BCH Codes}
\label{chap:bch}

By construction, every field $\field{F}_{p^m}$ has the prime field $\field{F}_p$ as a subfield, i.e., $\field{F}_p\subset\field{F}_{p^m}$ is closed under addition and multiplication. Consequently, every code that is a subspace of $\field{F}_{p^m}^n$ contains a subspace of $\field{F}_p^n$ as a subcode.
\begin{definition}\label{def:bch:bch}
The binary subcode of an RS code over $\field{F}_{2^m}$ is called a binary \emph{Bose--Chaudhuri--Hocquenghem} (BCH) \emph{code}\index{Bose-Chaudhuri-Hocquenghem code}\index{BCH code \see{Bose-Chaudhuri-\\Hocquenghem code}}. 
\end{definition}
In this course, ``BCH code'' is short for ``binary BCH code''. 
\section{Basic Properties}
Consider the binary BCH subcode of an $(n=2^m-1,k)$ RS code. Since the RS mothercode is cyclic, all codewords are multiples of the generator polynomial
\begin{align}
g(x)=\prod_{i=1}^{n-k}(x-\alpha^i)
\end{align}
and so are the binary codewords. The BCH code is also cyclic, see Problem~\ref{prob:bch:cyclic}. By definition, the generator polynomial of a cyclic code over $\field{F}_2$ is the codeword whose generating function in $\field{F}_2[x]$ is of least degree. We thus have the following.
\begin{theorem}\label{theo:bch:firstcharacerization}
The BCH subcode of an $(n=2^m-1,k)$ RS code over $\field{F}_{2^m}$ is given by
\begin{align}
\{\tilde{g}(x)p(x)\colon p(x)\in\field{F}_2[x], \deg p(x)<\tilde{k}\}
\end{align}
where $\tilde{g}(x)$ is the polynomial of least degree in $\field{F}_2[x]$ that has $\{\alpha,\dotsc,\alpha^{n-k}\}$ as roots and $\tilde{k}=n-\deg \tilde{g}(x)$.
\end{theorem}
The theorem does not tell us how to find $\tilde{g}(x)$ and we don't even know the dimension $\tilde{k}$ of the BCH code. We first study the construction of minimal polynomials, which we defined in Section~\ref{sec:rs:fieldspm}. This will then help us in constructing the generator polynomial $\tilde{g}(x)$.

\subsection{Construction of Minimal Polynomials}\label{sec:bch:minimal}

\begin{lemma}\label{lem:bch:linearpower}
Let $\beta,\omega$ be two elements in $\field{F}_{2^m}$. Let $i\geq 1$ be a positive integer. Then
\begin{align}
(\beta+\omega)^{2^i}=\beta^{2^i}+\omega^{2^i}.
\end{align}
More generally, let $\beta_1,\beta_2,\dotsc,\beta_\ell$ be elements in $\field{F}_{2^m}$. Then
\begin{align}
(\beta_1+\dotsb+\beta_\ell)^{2^i}=\beta_1^{2^i}+\dotsb+\beta_\ell^{2^i}.
\end{align}
\end{lemma}
\begin{proof}
Taking the $2^i$th power is equivalent to taking the $2$nd power $i$ times, i.e.,
\begin{align}
(\beta+\omega)^{2^i}=\left(\dotsb\left((\beta+\omega)^2\right)^2\dotsb\right)^2.
\end{align}
Therefore, if the statement is true for $i=1$, it is also true for $i>1$. For $i=1$, we calculate
\begin{align}
(\beta+\omega)^2&=\beta^2+2\beta\omega+\omega^2\\
&=\beta^2+\omega^2.
\end{align}
The generalization follows by repeatedly applying the just shown identity, i.e.,
\begin{align}
(\beta_1+\dotsb+\beta_\ell)^{2^i}&=\left(\Bigl(\bigl(\dotsb(\beta_1+\beta_2)+\dotsb\bigr)+\beta_{\ell-1}\Bigl)+\beta_\ell\right)^{2^i}\nonumber\\
&=\Bigl(\bigl(\dotsb(\beta_1+\beta_2)+\dotsb\bigr)+\beta_{\ell-1}\Bigl)^{2^i}+\beta_\ell^{2^i}\nonumber\\
&\;\;\vdots\nonumber\\
&=\beta_1^{2^i}+\dotsb+\beta_\ell^{2^i}.
\end{align}
\end{proof}
\begin{lemma}\label{lem:bch:binarypolynomialcheck}
Let $f(x)$ be a polynomial in $\field{F}_{2^m}[x]$. Then
\begin{align}
f(x)\in\field{F}_2[x]\Leftrightarrow f^2(x)=f(x^2).
\end{align}
\end{lemma}
\begin{proof}
Let $j$ be the degree of $f$, i.e.,
\begin{align}
f(x)=f_0+f_1x+\dotsb+f_jx^j.
\end{align}
By Lemma~\ref{lem:bch:linearpower},
\begin{align}
f^2(x)=f_0^2+f_1^2x^{2}+\dotsb+f_j^2x^{2j}.
\end{align}
By comparing the coefficients, this polynomial is equal to
\begin{align}
f(x^2)=f_0+f_1x^{2}+\dotsb+f_jx^{2j}
\end{align}
if and only if $f_i^2=f_i$, $i=0,1,\dotsc,j$. This equation is fulfilled for $f_i\in\{0,1\}=\field{F}_2$. By Theorem~\ref{theo:noname}, the equation cannot have more than two distinct solutions.
\end{proof}
\begin{definition}
In $\field{F}_{p^m}$ the mapping $\omega\mapsto\omega^p$ is called \emph{conjugation}\index{conjugation}. If $\beta=\omega^{p^i}$ for some positive integer $i$, then $\beta$ is called a \emph{conjugate}\index{conjugate} of $\omega$.
\end{definition}
\begin{lemma}\label{lem:bch:conjugacyclasses}
The relation
\begin{align}
\beta\sim\omega\Leftrightarrow \beta=\omega^{p^i}\text{ for some non-negative integer }i
\end{align}
defines an equivalence relation in $\field{F}_{p^m}$. The resulting equivalence classes are called \emph{conjugacy classes}\index{conjugacy class}.
\end{lemma}
\begin{proof}

\emph{reflexive:} $\beta=\beta^1\Rightarrow\beta\sim\beta$.

\emph{transitive:} For $\beta_1,\beta_2,\beta_3\in\field{F}_{p^m}$, suppose $\beta_1\sim\beta_2$ and $\beta_2\sim\beta_3$, i.e., $\beta_2=\beta_1^{p^i}$ and $\beta_3=\beta_2^{p^j}$ for some positive integers $i,j<p^m-1$. Then
\begin{align}
\beta_3=\beta_2^{p^j}=(\beta_1^{p^i})^{p^j}=\beta_1^{p^{i+j}}\Rightarrow \beta_1\sim \beta_3.
\end{align}

\emph{symmetric:} Suppose $\beta_1\sim\beta_2$, i.e., $\beta_2=\beta_1^{p^i}$ for some $i<p^m-1$. Then
\begin{align}
\beta_2^{p^{m-i}}=\beta_1^{p^ip^{m-i}}=\beta_1^{p^m}=\beta_1\Rightarrow \beta_2\sim\beta_1.
\end{align}
\end{proof}
\begin{lemma}\label{lem:bch:roots}
Let $f(x)$ be a polynomial in $\field{F}_2[x]$ with root $\beta$. Then each element in the conjugacy class $\setb$ of $\beta$ is also a root of $f(x)$. 
\end{lemma}
\begin{proof}
By Lemma~\ref{lem:bch:binarypolynomialcheck}, $f^2(x)=f(x^2)$. In particular
\begin{align}
0=0^{2^i}=f^{2^i}(\beta)=f(\beta^{2^i})
\end{align}
which shows that all elements in $\setb$ are roots of the polynomial $f(x)$.
\end{proof}
We can now state and prove the following lemma, which characterizes minimal polynomials.
\begin{lemma}\label{lem:bch:minimalpolynomial}
Let $\beta$ be an element in $\field{F}_{2^m}$. Let $\setb$ be the conjugacy class of $\beta$. Then
\begin{align}
h(x)=\prod_{\omega\in\setb}(x-\omega)
\end{align}
is the minimal polynomial of $\beta$ in $\field{F}_2[x]$.
\end{lemma}
\begin{proof}
Let $\mu(x)$ be the minimal polynomial of $\beta$. By definition, $\mu(\beta)=0$ and $\mu(x)\in\field{F}_2[x]$. By Lemma~\ref{lem:bch:roots}, all elements in $\setb$ are also roots of $\mu(x)$. This shows that $h(x)\mid \mu(x)$ in $\field{F}_{2^m}[x]$. It remains to show that $h(x)$ is in $\field{F}_2[x]$. To this end, we check the condition in Lemma~\ref{lem:bch:binarypolynomialcheck}
\begin{align}
h^2(x)&=\prod_{\omega\in\setb}(x-\omega)^2\\
&\oeq{a}\prod_{\omega\in\setb}(x^2-\omega^2)\\
&\oeq{b}\prod_{\omega\in\setb}(x^2-\omega)\\
&=h(x^2)
\end{align}
where (a) follows by Lemma~\ref{lem:bch:linearpower} and where (b) follows by Lemma~\ref{lem:bch:conjugacyclasses}. By Lemma~\ref{lem:bch:binarypolynomialcheck}, $h(x)\in\field{F}_2[x]$. This shows that $h(x)=\mu(x)$, i.e., $h(x)$ is the minimal polynomial of $\beta$, as claimed in the lemma.
\end{proof}
\subsection{Generator Polynomial of BCH Codes}
Using our results on minimal polynomials, we can now state and prove the construction of generating polynomials of BCH codes.
\begin{theorem}\label{theo:bch:generatorpolynomial}
Let $\setc_\text{RS}$ be an $(n=2^m-1,k)$ RS code over $\field{F}_{2^m}$ with primitive element $\alpha$. Let $\setr$ be a set containing one representative of each conjugacy class of the elements in $\seta=\{\alpha,\alpha^2,\dotsc,\alpha^{n-k}\}$. Let $\setb$ be the set of all conjugates of elements in $\seta$. The generator polynomial of the binary BCH subcode of $\setc_\text{RS}$ is given by
\begin{align}
\tilde{g}(x)=\prod_{\beta\in\setb}(x-\beta)=\prod_{\omega\in\setr}\Phi_\omega(x)\label{eq:bch:bchgen}
\end{align}
where $\Phi_\omega(x)$ is the minimal polynomial in $\field{F}_2[x]$ of $\omega$.
\end{theorem}
\begin{proof}
By Theorem~\ref{theo:bch:firstcharacerization} and Lemma~\ref{lem:bch:roots}, the generator polynomial of the BCH code must be the polynomial in $\field{F}_2[x]$ of least degree that has all the roots of the $\tilde{g}(x)$ defined in \eqref{eq:bch:bchgen}. The polynomial $\tilde{g}(x)$ is by Lemma~\ref{lem:bch:minimalpolynomial} the product of minimal polynomials in $\field{F}_2[x]$ and therefore $\tilde{g}(x)\in\field{F}_2[x]$. All the roots of $\tilde{g}(x)$ are distinct. By Theorem~\ref{theo:noname}, a polynomial of degree $\ell$ can have at most $\ell$ distinct roots. Therefore, there can be no polynomial with degree less than $\tilde{g}(x)$ that has all the roots of $\tilde{g}(x)$. The polynomial $\tilde{g}(x)$ must therefore be the generator polynomial of the BCH code.
\end{proof}
\begin{example}\label{ex:bch:generatorpolynomial}
For the $(n=16-1,12)$ RS code over the field $\field{F}_{16}$, we construct the generator polynomial $\tilde{g}(x)$ of the binary BCH subcode using Theorem~\ref{theo:bch:generatorpolynomial}. We have $n-k=15-12=3$. Let $\alpha$ be a primitive element of $\field{F}_{16}$. The elements $\seta=\{\alpha,\alpha^2,\alpha^3\}$ need to be roots of $\tilde{g}(x)$. The conjugacy classes are
\begin{align}
\alpha:&\{\alpha,\alpha^2,\alpha^4,\alpha^8\}\\
\alpha^2:&\surd\\
\alpha^3:&\{\alpha^3,\alpha^6,\alpha^{12},\alpha^{9}\}
\end{align}
By $\surd$ we indicate that the element is contained in an already calculated conjugacy class. The set $\setb$ of all conjugates of elements in the set $\seta$ is
\begin{align}
\setb = \{\alpha,\alpha^2,\alpha^4,\alpha^8,\alpha^3,\alpha^6,\alpha^{12},\alpha^{9}\}.
\end{align}
The set of representatives of the conjugacy classes is not unique. We choose
\begin{align}
\setr=\{\alpha,\alpha^3\}.
\end{align}
We could also have chosen $\setr'=\{\alpha^2,\alpha^{12}\}$. The degree of the generator polynomial $\tilde{g}(x)$ is $\deg\tilde{g}(x)=|\setb|=8=n-\tilde{k}$. Thus, the dimension of the BCH code is $n-\deg\tilde{g}(x)=15-8=7$. The generator polynomial is given by
\begin{align}
\tilde{g}(x)=\prod_{\beta\in\setb}(x-\beta).
\end{align}
From the expression on the right-hand side, it is not obvious that $\tilde{g}(x)$ is in $\field{F}_2[x]$, i.e., that it is a polynomial with coefficients in $\field{F}_2$ (although Theorem~\ref{theo:bch:generatorpolynomial} guarantees this). We could expand the right-hand side; instead, we look up the minimal polynomials of the representatives in $\setr$ for example in \cite[Appendix~B]{lin2004error}. We find
\begin{align}
&\Phi_{\alpha}(x) = 1+x+x^4\\
&\Phi_{\alpha^3}(x) = 1+x+x^2+x^3+x^4.
\end{align}
Thus,
\begin{align}
\tilde{g}(x)=\prod_{\beta\in\setr}\Phi_\beta(x)=(1+x+x^4)(1+x+x^2+x^3+x^4)
\end{align}
and indeed, $\tilde{g}(x)$ is in $\field{F}_2$ and $\deg\tilde{g}(x)=8$. Finally, we can expand the right-hand side and get
\begin{align}
\tilde{g}(x)=1+x+x^2+x^3&+x^4\nonumber\\
+x+x^2+x^3&+x^4+x^5\nonumber\\
&+x^4+x^5+x^6+x^7+x^8\nonumber\\
=1+x^4+x^6+x^7&+x^8.
\end{align}
\end{example}
\section{Design of BCH Codes Correcting $t$ Errors}

Recall that a BCH code is the binary subcode of an RS mothercode. RS codes are MDS, i.e., the minimum distance is given by $d=n-k+1$. The actual minimum distance $\tilde{d}$ of the BCH code is at least as large as the minimum distance $d$ of the RS mothercode, however, BCH codes are in general \emph{not} MDS, this is because their dimension $\tilde{k}$ is in most cases smaller than $k$. The true minimum distance $\tilde{d}$ can only be determined by searching over all codewords. For large block lengths, this is infeasible. We therefore use the (known) minimum distance $d$ of the RS mothercode as the \emph{design distance}\index{design distance} of the BCH subcode. Summarizing,
\begin{align}
\tilde{d}\geq d=n-k+1.
\end{align}
This inequality is very useful, since it allows the design of BCH codes for a given block length $n$ and a required number $t$ of correctable errors. We illustrate this by an example.
\begin{example}\label{ex:bch:designdistance}
Suppose we are asked to design a BCH code with block length at most $20$ that is guaranteed to correct up to $t=2$ errors. The block length has to be of the form $n=2^m-1$ for some positive integer $m$. The largest $m$ such that $n=2^m-1\leq 20$ is $m=4$, so we choose $n=15$. By Theorem~\ref{theo:rs:mdd}, we need a minimum distance of $2t+1=5$. The RS code is MDS, i.e, its dimension has to be $k=n-d+1=15-5+1=11$ and in particular $n-k=2t=4$. We construct the generator polynomial of the BCH code. According to Theorem~\ref{theo:bch:generatorpolynomial}, its roots are $\seta=\{\alpha,\alpha^2,\alpha^3,\alpha^4\}$ and their conjugates. Since $\alpha^4$ is already in the conjugacy class of  $\alpha$, the generator polynomial is the generator polynomial $\tilde{g}(x)$ that we calculated in Example~\ref{ex:bch:generatorpolynomial}. The $(15,12)$ and the $(15,11)$ RS codes have the same binary BCH subcode! The dimension of the BCH code is $\tilde{k}=n-8=7$. By the Singleton bound, the actual minimum distance $\tilde{d}$ of the BCH code is bounded as
\begin{align}
\tilde{d}\leq n-\tilde{k}+1=9.
\end{align}
The design minimum distance $d=5$ is thus $4$ coordinates away from the Singleton bound. 
\end{example}

\section{Erasure Decoding}

Let $\vecc$ be the code word of some linear $(n,k)$ code $\setc$ over $\field{F}_q$. Suppose at positions $\delta_1,\dotsc,\delta_{n-\check{k}}$, the code word entries get erased, while at the other positions $\rho_1,\rho_2,\dotsc,\rho_{\check{k}}$, the entries arrive at the receiver unaltered. In this section, we will investigate how for codes with a guaranteed minimum distance, the receiver can explore its knowledge of the correct entries for decoding.

\subsection{Erasure Decoding of MDS Codes}

Let $\setc$ be an MDS code with minimum distance $d=n-k+1$. Suppose $\check{k}=k$ entries are received correctly. Let $\matg=(\vecg_1,\vecg_2,\dotsc,\vecg_n)$ be a $k\times n$ generator matrix of $\setc$ and consider the matrix
\begin{align}
\tilde{\matg}=(\vecg_{\rho_1},\vecg_{\rho_2},\dotsc,\vecg_{\rho_k}).
\end{align}
For the information vector $\vecu=(u_1,\dotsc,u_k)$ that was encoded to $\vecc$, we have
\begin{align}
\vecu\tilde{\matg}=(c_{\rho_1},c_{\rho_2},\dotsc,c_{\rho_k}).
\end{align}
Since $\setc$ is MDS (i.e., every set of $k$ coordinates of the code forms an information set, see Theorem~\ref{rs:theo:singleton}), the $k\times k$ matrix $\tilde{\matg}$ is full rank and invertible. We can thus recover $\vecu$ from the correctly received entries of $\vecc$ by
\begin{align}
\vecu=(c_{\rho_1},c_{\rho_2},\dotsc,c_{\rho_k})\tilde{\matg}^{-1}.
\end{align}

\subsection{Erasure Decoding of BCH Codes}

Let now $\setc$ be a $(n,\tilde{k})$ binary BCH code with design minimum distance $d$. This means that after erasing $d-1$ entries of the codeword, it still differs from each other code word in at least one entry, see Subsection~\ref{subsec:singleton} where we used the same argument. Suppose now that in $\check{k}\geq n-d+1$, the code word is received correctly. By the Singleton bound, $\tilde{k}\leq n-d+1\leq \check{k}$. Consider the $\tilde{k}\times\check{k}$ matrix
\begin{align}
\check{\matg}=(\vecg_{\rho_1},\vecg_{\rho_2},\dotsc,\vecg_{\rho_{\check{k}}}).
\end{align}
Because the minimum distance of the code is $d\geq n-\check{k}+1$, this matrix maps each information vector $\vecu$ to a different length $\check{k}$ vector. Therefore, the row rank of $\check{\matg}$ is $\tilde{k}$ and in particular, $\check{\matg}$ has $\tilde{k}$ linearly independent columns  
\begin{align}
\{\vecg_{\tilde{\rho}_1},\vecg_{\tilde{\rho}_2},\dotsc,\vecg_{\tilde{\rho}_{\tilde{k}}}\}\subseteq\{\vecg_{\rho_1},\vecg_{\rho_2},\dotsc,\vecg_{\rho_{\check{k}}}	\}\label{eq:bch:full rank}
\end{align}
and the $\tilde{k}\times\tilde{k}$ matrix
\begin{align}
\tilde{\matg}:=(\vecg_{\tilde{\rho}_1},\vecg_{\tilde{\rho}_2},\dotsc,\vecg_{\tilde{\rho}_{\tilde{k}}})
\end{align}
has full rank $\tilde{k}$ and is invertible. We can now recover the encoded information vector by
\begin{align}
\vecu=(u_1,u_2,\dotsc,u_{\tilde{k}})=(c_{\tilde{\rho}_1},c_{\tilde{\rho}_2},\dotsc,c_{\tilde{\rho}_{\tilde{k}}})\tilde{\matg}^{-1}.
\end{align}
In summary, for $(n,k)$ MDS codes, we can use any $k$ correctly received code word entries to recover the transmitted information. In contrast, for $(n,\tilde{k})$ BCH codes, we need to find a set of $\tilde{k}$ correctly received entries with the property that the corresponding columns of the generator matrix are linearly independent. We are guaranteed to find such a set if we receive $\tilde{k}\geq n-(d-1)$ entries correctly, where $d$ is the design minimum distance given by $d=n-k+1$ with $k$ being the dimension of the RS mothercode.  

\section{Decoding of BCH Codes}

In this section, we derive how to decode BCH codes efficiently. We start with an example and then develop the general case. 
\subsection{Example}\label{subsec:bch:example}
Consider the $(15,7)$ BCH code we designed in Example~\ref{ex:bch:designdistance}. Suppose codeword $\vecc\leftrightarrow c(x)$ was transmitted over a binary channel and the received vector is
\begin{align}
y(x)=c(x)+e(x)
\end{align}
where $e(x)$ is the polynomial representation of the error pattern $\vece$. The codeword $c(x)$ has roots at 
\begin{align}
\{\alpha,\alpha^2,\alpha^3,\alpha^4,\alpha^6,\alpha^8,\alpha^{9},\alpha^{12}\}.
\end{align}
Thus,
\begin{align}
y(\alpha^i)=c(\alpha^i)+e(\alpha^i)=e(\alpha^i),\quad i=1,2,3,4,6,8,9,12
\end{align}
that is, we know $e(\alpha^i)$ at the roots $\alpha^i$. Recall that we designed the code to correct $2$ errors and that the minimum distance needed for this was $5$. To decode up to $2$ errors, we can use a minimum distance decoder for the RS mothercode. Since the generator polynomial for the RS mothercode has the roots $\{\alpha,\alpha^2,\alpha^3,\alpha^4\}$, by Theorem~\ref{theo:rs:syndrome}, the syndrome is
\begin{align}
\vecs=[y(\alpha),y(\alpha^2),y(\alpha^3),y(\alpha^4)]=[e(\alpha),e(\alpha^2),e(\alpha^3),e(\alpha^4)].
\end{align}
From this syndrome, we can correct all error patterns with weight smaller or equal to $2$. Note that this is the syndrome of the RS mother code. Suppose two errors occurred at the positions $i_1$ and $i_2$, i.e., the error polynomial is
\begin{align}
e(x)=x^{i_1}+x^{i_2}.
\end{align}
Define $\gamma:=\alpha^{i_1}$ and $\rho:=\alpha^{i_2}$. We have the following information about $e(x)$.
\begin{align}
&s_0=e(\alpha)=\alpha^{i_1}+\alpha^{i_2}=\gamma+\rho\\
&s_1=e(\alpha^2)=\gamma^2+\rho^2\\
&s_2=e(\alpha^3)=\gamma^3+\rho^3\\
&s_3=e(\alpha^4)=\gamma^4+\rho^4.
\end{align}
By Lemma~\ref{lem:bch:linearpower}, the second and the fourth equations are linearly dependent of the first equation. We therefore discard equation two and four and try to find $\gamma$ and $\rho$ using the system of equations
\begin{align}
&s_0=e(\alpha)=\gamma+\rho\\
&s_2=e(\alpha^3)=\gamma^3+\rho^3.
\end{align}
We know $\gamma+\rho=s_0$. We need another equation with $\gamma$ and $\rho$. We do a trick.
\begin{align}
s_0^3&=(\gamma+\rho)^3=(\gamma+\rho)(\gamma+\rho)^2\\
&=(\gamma+\rho)(\gamma^2+\rho^2)\\
&=\gamma^3+\rho^3+\gamma\rho(\gamma+\rho)\\
&=s_2+\gamma\rho s_0
\end{align}
We solve for $\gamma\rho$ and get
\begin{align}
\gamma\rho=(s_0^3-s_2)\cdot s_0^{-1}=s_0^2-s_2\cdot s_0^{-1}.
\end{align}
We can now solve for $\gamma$ and $\rho$. We can write this step as the problem of factoring a polynomial since
\begin{align}
x^2-(\gamma+\rho)x+\gamma\rho=(x-\gamma)(x-\rho).
\end{align}
The roots $\gamma$ and $\rho$ of this polynomial ``locate'' the errors. In the general case, this polynomial will be called the \emph{error locator polynomial}.

\subsection{Linear Recurrence Relations}

\begin{definition}
A right-infinite sequence $a_0,a_1,\dotsc$ is a \emph{linear recurrence sequence}\index{linear recurrence sequence} of order $k$, if
\begin{align}
a_n = c_1a_{n-1}+c_2a_{n-2}+\dotsb+c_ka_{n-k},\quad \forall n\geq k
\end{align}
where $c_k\neq 0$.
\end{definition}
Let $a(x)$ be the generating function of the sequence $a_0,a_1,\dotsc$. We define the polynomial
\begin{align}
b(x)=1-c_1x-c_2x^2-\dotsb-c_kx^k.
\end{align}
Multiplying $a(x)$ with $b(x)$, the $n$th coefficient is
\begin{align}
d_n=a_n-c_1a_{n-1}-c_2a_{n-2}-\dotsb-c_ka_{n-k}.
\end{align}
By the definition of $a(x)$, $d_n=0$ for $n\geq k$. We define
\begin{align}
d(x)=d_0+d_1x+\dotsb+d_{k-1}x^{k-1}=a(x)b(x)
\end{align}
and write $a(x)$ as
\begin{align}
a(x)=\frac{d(x)}{b(x)}.
\end{align}
The coefficients of $d(x)$ are the initial values of $a(x)$. The polynomial $b(x)$ is the \emph{recurrence}\index{recurrence}. The degree of $b(x)$ is $\deg b(x)=k$ since $c_k\neq 0$. The degree of $d(x)$ is smaller or equal to $k-1$, depending on the initial values. Thus, we have the following result.
\begin{lemma}\label{lem:bch:recurrencepolynomial}
A power series $a(x)$ is a linear recurrence sequence if there exist two polynomials $d(x)$ and $b(x)$ with $\deg d(x)<\deg b(x)<\infty$ such that
\begin{align}
a(x)=\frac{d(x)}{b(x)}.\label{eq:bch:recurrence}
\end{align}
\end{lemma}
\begin{example}
The sequence $a_0=1,a_1=2,a_2=4,\dotsc$ fulfills for $c_1=2$ the equation
\begin{align}
a_n=c_1\cdot a_{n-1},\quad n\geq 1
\end{align}
so it is a linear recurrence sequence of order $1$. Multiplying $a(x)$ by $1-2x$, we get
\begin{align}
a(x)(1-2x)=(1+2x+4x^2+\dotsb)(1-2x)=1
\end{align}
so 
\begin{align}
a(x)=1+2x+4x^2+\dotsb=\sum_{i=0}^\infty (2x)^i=\frac{1}{1-2x}
\end{align}
which is the well-known geometric series formula.
\end{example}
Note that in \eqref{eq:bch:recurrence}, we can multiply both numerator and denominator by another polynomial $h(x)$ to get
\begin{align}
a(x)=\frac{d(x)h(x)}{b(x)h(x)}.
\end{align}
Then, $b(x)h(x)$ is another recurrence relation of $a(x)$. The recurrence of $a(x)$ in $\myfield[x]$ of least degree is called the \emph{minimal recurrence}\index{minimal recurrence} of $a(x)$ in \myfield[x]. The following theorem is the reason why the minimal polynomial is unique up to a scalar multiplication. We state it without a proof.
\begin{theorem}\label{theo:rs:ufd}
Let $p(x)$ be a polynomial in $\myfield[x]$. Then $p(x)$ can be written as the product of prime polynomials in $\myfield[x]$, scaled by a field element. This factorization is unique up to permutation of the prime polynomials.
\end{theorem}

\begin{lemma}
Suppose $a(x)$ is a linear recurrence sequence with recurrence $\mu(x)$ that is minimal in $\myfield[x]$. 
\begin{enumerate}
\item Any other recurrence is a multiple of $\mu(x)$. 
\item If $b(x)\in\myfield[x]$ is a recurrence and $a(x)=d(x)/b(x)$ with $\deg d(x)<\deg b(x)$ then $b(x)$ is minimal in $\myfield[x]$ if and only if $d(x)$ and $b(x)$ are co-prime in $\myfield[x]$, i.e., have no common factors in $\myfield[x]$. 
\end{enumerate}
\end{lemma}
\begin{proof}
To prove the lemma, we think of each polynomial written in its prime factorization in $\myfield[x]$, which is unique by Theorem~\ref{theo:rs:ufd}.
 
1. Let $a(x)=\nu(x)/\mu(x)$. Since $\mu(x)$ is minimal, $\nu(x)$ and $\mu(x)$ are co-prime, because otherwise we could cancel out the common factor, which would contradict that $\mu(x)$ is minimal. Let $a(x)=d(x)/b(x)$ for some other recurrence $b(x)$. Then from
\begin{align}
&\frac{\nu(x)}{\mu(x)}=\frac{d(x)}{b(x)}\label{eq:bch:minimal}
\end{align}
we see that $b(x)$ is a multiple of $\mu(x)$, since $\nu(x)$ and $\mu(x)$ share no common factors. 

2. Suppose now $d(x)$ and $b(x)$ are co-prime and $a(x)=d(x)/b(x)$. Then
\begin{align}
\frac{d(x)}{b(x)}=\frac{\nu(x)}{\mu(x)}
\end{align}
which shows that $b(x)$ must be a scalar multiple of $\mu(x)$ and consequently a minimal recurrence.
\end{proof}

\subsection{Syndrome Polynomial as Recurrence}

Suppose we have a BCH code that guarantees by its design distance the correction of up to $t$ errors. The number $t$ relates to the parameters of the RS mothercode by
\begin{align}
t=\left\lfloor\frac{d-1}{2}\right\rfloor=\left\lfloor\frac{n-k+1-1}{2}\right\rfloor=\left\lfloor\frac{n-k}{2}\right\rfloor.
\end{align}
We therefore set
\begin{align}
2t=n-k
\end{align}
and write in the following $2t$ instead of $n-k$. The approach that we took in our introductory example was to find the coefficients of the expansion of the polynomial 
\begin{align}
f(x)=(x-\gamma_1)\dotsb(x-\gamma_t)
\end{align}
and then to search for the roots $\gamma_j$. Following the literature, we equivalently use in the following the polynomial
\begin{align}
\ell(x):=\prod_{j=1}^t(1-\gamma_j x).
\end{align}
The polynomial $\ell(x)$ is called the \emph{error locator polynomial}\index{error locator polynomial}. Note that $\ell(x)=x^tf(1/x)$. We explicitly allow $\gamma_j=0$, i.e, $\ell(x)$ can represent any number of $0$ up to $t$ errors.

Suppose we have a $t$-error correcting BCH code, i.e., the generator polynomial of the RS mothercode has roots $\alpha,\alpha^2,\dotsc,\alpha^{2t}$. Let $i_1,\dotsc,i_t$ be the $t$ (unknown) positions. Define $\gamma_j=\alpha^{i_j}$. Then we have the relations
\begin{align}
\gamma_1+\dotsb+\gamma_t &= s_0\\
\gamma_1^2+\dotsb+\gamma_t^2 &= s_1\\
\gamma_1^3+\dotsb+\gamma_t^3 &= s_2\\
&\vdots\nonumber\\
\end{align}
These relations define an infinite sequence $s_0,s_1,\dotsc$, of which we know the first $2t$ numbers, since $s_0,s_1,\dotsc,s_{2t-1}$ is the syndrome of the observed channel output. We denote the polynomial of the infinite sequence by $\sigma(x)$. The polynomial $\sigma(x)$ and the syndrome polynomial $s(x)$ are identical in the first $2t$ coefficients.
\begin{theorem}\ 
The polynomial $\sigma(x)$ and the error locator polynomial $\ell(x)$ relate as
\begin{align}
\sigma(x)=\frac{-\ell'(x)}{\ell(x)}.
\end{align}
where $\ell'(x)$ denotes the \emph{formal derivative} of $\ell(x)$. Furthermore, $\ell(x)$ is the minimal recurrence of $\sigma(x)$.
\end{theorem}
\begin{proof}
We first show the identity.
\begin{align}
\ell'(x)&=-\sum_{j=1}^t\gamma_j\frac{\ell(x)}{1-\gamma_jx}\\
&=-\ell(x)\sum_{j=1}^t\gamma_j\sum_{i=0}^\infty(\gamma_jx)^i\\
&=-\ell(x)\sum_{i=0}^\infty \Bigl(\underbrace{\sum_{j=1}^t\gamma_j^{i+1}}_{=s_i}\Bigr)x^i\\
&=-\ell(x)\sigma(x).
\end{align}
For the degrees, we have $\deg\ell(x)=\deg \ell'(x)+1=t$, so the conditions of Lemma~\ref{lem:bch:recurrencepolynomial} are fulfilled. Since all roots of $\ell(x)$ are distinct, $\ell'(x)$ does not have a root of $\ell(x)$ as a factor, i.e., $\ell(x)$ and $\ell'(x)$ are co-prime. This shows that $\ell(x)$ is a minimal recurrence of $\sigma(x)$.
\end{proof}

\subsection{Berlekamp-Massey Algorithm}

If the number of errors is smaller or equal to $t$, the \emph{Berlekamp-Massey Algorithm}\index{Berlekamp-Massey Algorithm} \cite{massey1969shift},\cite[Figure~8.4]{massey1997applied} applied to $s_0,s_1,\dotsc,s_{2t-1}$ finds the minimal recurrence $\ell(x)$ of $\sigma(x)$. Decoding by minimal recurrence is an instance of a minimum distance decoder, since the degree of the minimal recurrence is exactly the weight of the error pattern of minimum weight that explains the observed channel output.

\clearpage

\section{Problems}

\myproblem\label{prob:bch:cyclic} Conclude from Definition~\ref{def:bch:bch} that BCH codes are cyclic.

\myproblem Suppose the BCH code from Example~\ref{ex:bch:designdistance} is used on a BSC with crossover probability $\delta=0.11$. Show that the probability of error $P_e$ of an ML decoder is bounded as
\begin{align}
P_e\leq 1-\sum_{\ell=0}^2 {15\choose\ell}(1-\delta)^{15-\ell}\delta^\ell.
\end{align}
Place the operating point of the BCH code in Figure~\ref{fig:optimal_codes_bsc_pb}. Use the bound as an estimate for $P_e$.

\myproblem Consider the $(n=16-1,12)$ RS code over the field $\field{F}_{16}$. Construct the generator polynomial $\tilde{g}(x)$ of the binary BCH subcode. 

\myproblem For the example in Subsection~\ref{subsec:bch:example}, suppose the channel output is
\begin{align}
\vecy=000010111000010.
\end{align}
Decode by applying the procedure suggested in Subsection~\ref{subsec:bch:example}.

\myproblem For the example in Subsection~\ref{subsec:bch:example}, suppose again the channel output is
\begin{align}
\vecy=000010111000010.
\end{align}
Decode by applying the Berlekamp-Massey Algorithm.

\myproblem
Consider the BCH subcode of an $(n=2^3-1,5)$ RS code. A primitive polynomial for $\field{F}_{2^3}$ is $p(z)=1+z+z^3$ with root $\alpha$. The correspondence table is
\begin{align*}
0&\leftrightarrow 0\\
1&\leftrightarrow \alpha^0\\
z&\leftrightarrow \alpha^1\\
z^2&\leftrightarrow\alpha^2\\
1+z&\leftrightarrow\alpha^3\\
z+z^2&\leftrightarrow\alpha^4\\
1+z+z^2&\leftrightarrow\alpha^5\\
1+z^2&\leftrightarrow\alpha^6
\end{align*}
\begin{enumerate}
\item What is the design minimum distance of the code?
\item Calculate the generator polynomial of the BCH code.
\item What is the dimension of the BCH code?
\end{enumerate}
 A codeword is transmitted over a binary channel. One bit is corrupted. The observed vector at the output of the channel is
\begin{align*}
\vecy=1001111\leftrightarrow 1+x^3+x^4+x^5+x^6.
\end{align*}
\begin{enumerate}
\addtocounter{enumi}{3}
\item Calculate the syndrome of $\vecy$.
\item Which codeword was transmitted?
\end{enumerate}
\myproblem\label{bch:prob:GF16}
The BCH subcode of an $(n=2^4-1,9)$ RS code is used. A primitive polynomial for $\myfield_{2^4}$ is $p(z)=1+z+z^4$ with root $\alpha$. The correspondence table is
\begin{center}
\begin{tabular}{rcl|rcl}
$0$&$\leftrightarrow$&$0$&$1+z+z^3$&$\leftrightarrow$&$\alpha^7$\\
$1$&$\leftrightarrow$&$1$&$1+z^2$&$\leftrightarrow$&$\alpha^8$\\
$z$&$\leftrightarrow$&$\alpha$&$z+z^3$&$\leftrightarrow$&$\alpha^9$\\
$z^2$&$\leftrightarrow$&$\alpha^2$&$1+z+z^2$&$\leftrightarrow$&$\alpha^{10}$\\
$z^3$&$\leftrightarrow$&$\alpha^3$&$z+z^2+z^3$&$\leftrightarrow$&$\alpha^{11}$\\
$1+z$&$\leftrightarrow$&$\alpha^4$&$1+z+z^2+z^3$&$\leftrightarrow$&$\alpha^{12}$\\
$z+z^2$&$\leftrightarrow$&$\alpha^5$&$1+z^2+z^3$&$\leftrightarrow$&$\alpha^{13}$\\
$z^2+z^3$&$\leftrightarrow$&$\alpha^6$&$1+z^3$&$\leftrightarrow$&$\alpha^{14}$
\end{tabular}
\end{center}
\begin{enumerate}
\item Calculate the generator polynomial $\tilde{g}(x)$ of the BCH code.
\item How many code words are in the BCH code?
\end{enumerate}
The BCH code is used on a BEC. The codeword $\vecc$ is transmitted over the BEC. The channel output is
\begin{align*}
\vecy = 11101100e0e000
\end{align*}
\begin{enumerate}
\addtocounter{enumi}{2}
\item Modify $\vecy$ by replacing the erasures by $0$s and calculate for the modified output $\hat{\vecy}$ the first 3 entries of the RS syndrome $\vecs=(s_0,s_1,s_2,s_3,s_4,s_5)$.
\item The modified channel output can be written as $\hat{\vecy}=\vecc+\vece$. Set up a system of linear equations for the unknown coefficients of $\vece$.
\item Solve your system of linear equations and determine $\vecc$. 
\end{enumerate}

\myproblem
The BCH subcode of an $(n=2^4-1,13)$ RS code is used. 
A primitive polynomial for $\myfield_{2^4}$ is $p(z)=1+z+z^4$ with root $\alpha$. The correspondence table is provided in Problem~\ref{bch:prob:GF16}
\begin{enumerate}
\item What is the design minimum distance of the code?
\item Can the code correct $2$ erasures?
\item Calculate the generator polynomial $\tilde{g}(x)$ of the BCH code.
\item What is the dimension of the BCH code?
\item What is the relation between $\tilde{g}(x)$ and $p(x)=1+x+x^4$?
\end{enumerate}

\myproblem Some company provides the following specification of an error correcting code for a BSC:
\begin{center}
\begin{tabular}{rl}
block length in bits&800\\
rate&0.95 bits/channel use
\end{tabular}
\end{center}
\begin{enumerate}
\item Your colleague suggests to use an $(n=2^7-1,k)$ RS code together with shortening.
\begin{enumerate}
\item Specify the shortening procedure and the code dimension $k$ such that the specification is fulfilled exactly, i.e., 800 uses of the binary channel are needed to transmit one code word and the rate is 0.95 bits/channel use.
\item How many bit errors can your colleague guarantee to correct by a minimum distance decoder?
\end{enumerate}
\item You propose to design a BCH code that meets the requirements. How many bit errors can you guarantee to correct by a minimum distance decoder? \emph{Hint:} Appendix C of \cite{lin2004error} may be helpful. 
\item Both for the RS code suggested by your colleague and for your BCH code, plot block error probability upper -bounds for BSC crossover probabilities\begin{align*}\delta=10^{-1}, 10^{-2}, 10^{-3}, 10^{-4}.\end{align*}
\end{enumerate}